\documentclass[11pt]{article}
\usepackage{leeseok}
\title{Rapid mixing of quantum spin chains at any finite temperature}
\author{Leeseok Kim\thanks{\text{leeseok@unm.edu}}\\
University of New Mexico}
\date{}

\begin{document}

\maketitle

\begin{abstract}
We prove that a quasi-local quantum Gibbs sampler rapidly mixes to the Gibbs state of any 1D finite-range Hamiltonian at all finite temperatures in $\mathcal{O}\left(\log (n/\varepsilon)\right)$ time. Each update traces out a contiguous block of constant length and reconstructs it using the Petz recovery map, following early proposals for quantum heat-bath Gibbs samplers by Kastoryano and Brand\~ao~\cite{kastoryano2016quantum}. Our proof directly bounds the distance between the updated state and the target Gibbs state using quantum Wasserstein distance of order $1$. We further provide $\mathcal{O}(\text{polylog}(n/\varepsilon))$ depth of quantum circuits to prepare these Gibbs states. We hope these techniques will facilitate the analysis and design of other Gibbs samplers.
\end{abstract}

\tableofcontents

\section{Introduction}

Gibbs states describe quantum systems in thermal equilibrium and are central to the study of quantum many-body physics such as finite-temperature phase transitions, quantum magnetism, and strongly correlated matter~\cite{alhambra2023quantum}. They also provide computational resources for quantum machine learning~\cite{amin2018quantum} and for optimization algorithms such as quantum semidefinite-programming solvers~\cite{brandao2017quantum}. Preparing Gibbs states on quantum computers is therefore an important task.

A variety of quantum algorithms have been developed for Gibbs-state preparation~\cite{terhal2000problem,poulin2009sampling,bilgin2010preparing,temme2011quantum,riera2012thermalization}. One popular approach is to design a Lindbladian whose dynamics converges to the target Gibbs state from any initial state. Such dynamics, known as quantum Gibbs samplers, extend the classical paradigm of sampling from Gibbs distributions using Markov chains~\cite{levin2017markov,norris1997markov}. Various quantum Gibbs samplers have been proposed~\cite{kastoryano2016quantum,davies1974markovian,davies1976quantum,chen2023an,ding2025efficient,bakshi2026a,scandi2026thermalization}. A major question is to determine how quickly these samplers converge to the Gibbs states of Hamiltonians of interests, which is quantified by the \emph{mixing time} detnoed by $t_{\rm mix}$.  Given a fixed target accuracy, \emph{fast mixing} refers to a $t_{\rm mix} \sim \text{poly}(n)$, whereas \emph{rapid mixing} refers to the stronger requirement of $t_{\rm mix} \sim \text{polylog}(n)$. From an algorithmic perspective, one would like $t_{\rm mix}$ to be as small as possible, since it directly governs the runtime of Gibbs-state preparation based on the corresponding dynamics. On the other hand, \emph{slow mixing}, where $t_{\rm mix}$ grows superpolynomially with $n$, is also of independent interest for various reasons~\cite{gamarnik2024slow,macieszczak2016towards,rose2016metastability,bergamaschi2025logical,rakovszky2026bottlenecks}.

A standard route to proving fast mixing is to establish an inverse-polynomial lower bound on the spectral gap of the Lindbladian. For detailed-balanced Gibbs samplers of bounded local Hamiltonians, such a bound implies a polynomial mixing time at fixed target accuracy~\cite{temme2013lower}. Rapid mixing, however, does not generally follow from this spectral-gap estimate alone. It can instead be established through modified logarithmic Sobolev inequalities~\cite{kastoryano2013quantum} or contraction arguments that track discrepancies between local sites~\cite{bubley1997path,rouze2026optimal,bakshi2026a}. These approaches have established rapid mixing in several important but restricted regimes. For commuting 1D Hamiltonians, rapid mixing is known at every fixed finite temperature~\cite{bardet2023rapid,kochanowski2025rapid}. For noncommuting Hamiltonians, it has been restricted to a sufficiently high temperature~\cite{rouze2026optimal,bakshi2026a} or a perturbative regime~\cite{smid2025rapid}. (See Section~\ref{subsec:related-works} for other results.) For general noncommuting Hamiltonians, however, all-temperature rapid mixing beyond the perturbative regime had remained open in 1D (although fast mixing was recently shown in Ref.~\cite{bergamaschi2026fast}).

In this work, we show that quantum block heat-bath dynamics \emph{rapidly mixes} to the Gibbs state of any finite-range, bounded 1D Hamiltonians at every finite temperature. Following the earlier heat-bath Gibbs sampler~\cite{kastoryano2016quantum}, each update in our dynamics traces out a contiguous block and reconstructs it using the Petz recovery map~\cite{petz1986sufficient} associated with the target Gibbs state. (These heat-bath Gibbs samplers have been studied previously~\cite{kastoryano2016quantum,temme2015fast,capel2018quantum,bardet2021approximate,bardet2021on}.) Importantly, rapid mixing holds for a block length independent of the system size, and the resulting Lindbladian is quasi-local.

Our proof is basically based on the path-coupling argument in the sense that it directly bounds contraction in the quantum Wasserstein distance of order $1$~\cite{depalma2021the}. The key observation is that if two states become identical after tracing out a site $i$, then any block update containing $i$ maps them to the same state. Updates on blocks that do not contain $i$ may increase the Wasserstein distance between the two states, but we show that the additional cost decays exponentially with the distance from $i$ to the updated block. In 1D, these contributions can be summed to obtain a bound independent of the system size. For a sufficiently large constant block length, therefore, the updates that erase single-site discrepancies dominate the additional influence from other updates.

The dynamics may initially appear impractical because each update is defined using the target Gibbs state itself. Our second main result, however, provides an explicit algorithm for approximating these dynamics using only classical descriptions of local portions of the Hamiltonian. Combined with rapid mixing, this yields a quantum circuit of depth $\operatorname{polylog}(n/\varepsilon)$ that prepares the target Gibbs state to trace-distance accuracy $\varepsilon$, matching the circuit-depth bound of Ref.~\cite{bergamaschi2026fast}.

\subsection{Main results}

We consider a family $\mathbf{H}$ of one-dimensional and nearest-neighbor Hamiltonians on periodic chains $\Lambda_n = \mathbb{Z}/n\mathbb{Z}$ (the extension to open chains is straightforward as we discuss later), with fixed finite local Hilbert-space dimension $d$. Each $H \in \mathbf{H}$ acts on $(\mathbb{C}^d)^{\otimes n}$ and takes the form
\begin{align}
\label{eq:hamiltonian}
H = \sum_{i \in \Lambda_n} h^{(i)}_{i,i+1}, \qquad \|h^{(i)}_{i,i+1}\| \leq 1,
\end{align}
where $h \in \mathcal{B}(\mathbb{C}^d \otimes \mathbb{C}^d)$ is a two-site interaction and site indices are understood modulo $n$. This includes arbitrary interactions of uniformly bounded range, since these can be reduced to the nearest-neighbor form by blocking a constant number of sites (and correspondingly increasing the local dimension). Also, importantly, we assume neither commutativity nor translation invariance.

Our goal is to determine how quickly a given Lindbladian drives an arbitrary initial state toward the Gibbs state
\begin{align}
\sigma_\beta = \frac{e^{-\beta H}}{\tr[e^{-\beta H}]}, \qquad 0\le \beta<\infty.
\end{align}
Specifically, we consider a quantum block heat-bath dynamics in one dimension, whose generator is constructed following a similar idea as the classical heat-bath (Glauber) dynamics~\cite{glauber1963timedependent}, and is inspired by the constructions of Refs.~\cite{kastoryano2016quantum,bardet2021on}. For a contiguous block $B\subset\Lambda_n$ of length $\ell$, the update first discards the state on $B$ and then reconstructs it using the Petz recovery map~\cite{petz1986sufficient} associated with the Gibbs state $\sigma_\beta$. The corresponding Lindbladian is
\begin{align}
\mathcal{L}_{\ell}(\rho) = \frac{1}{\ell} \sum_{\substack{B\subset\Lambda_n\\ |B|=\ell}} \left[ \Phi_B(\rho) - \rho \right], \qquad \Phi_B(\rho) := \sigma_\beta^{1/2} \left( I_B\otimes \sigma_{\beta,B^c}^{-1/2} \tr_B[\rho]  \sigma_{\beta,B^c}^{-1/2} \right) \sigma_\beta^{1/2},
\label{eq:block-petz-generator-intro}
\end{align}
where $\sigma_{\beta,B^c}=\tr_B[\sigma_\beta]$, and the sum runs over contiguous blocks of length $\ell$. Equivalently, $\Phi_B=\mathcal R^{\sigma_\beta}_{B^c\to\Lambda_n}\circ\tr_B$, where $\mathcal R^{\sigma_\beta}_{B^c\to\Lambda_n}$ is the Petz recovery map associated with $\tr_B$ and $\sigma_\beta$. The Petz recovery construction ensures that each block update preserves the Gibbs state, $\Phi_B(\sigma_\beta)=\sigma_\beta$, and hence $\mathcal{L}_{\ell}(\sigma_\beta)=0$. Thus, $\sigma_\beta$ is a stationary state of the dynamics. The dynamics also satisfies KMS detailed balance. We will show that, for a suitable constant block length $\ell$, the resulting Lindbladian mixes rapidly and is quasi-local.

\begin{table}[t]
  \centering
  \small
  \setlength{\tabcolsep}{4pt}
  \renewcommand{\arraystretch}{1.3}
  \begin{tabularx}{\linewidth}{@{}
    >{\raggedright\arraybackslash}p{0.23\linewidth}
    >{\raggedright\arraybackslash}X
    >{\raggedright\arraybackslash}p{0.145\linewidth}
    >{\raggedright\arraybackslash}p{0.170\linewidth}
    >{\raggedright\arraybackslash}p{0.16\linewidth}@{}}
    \toprule
    Reference & Geometry/Interaction & Temperature & Dynamics & Mixing time \\
    \midrule
    Bardet et al.~\cite{bardet2021on}
      & 1D/Commuting
      & Conditional\textsuperscript{a}
      & Single-site ~~~~~ Petz heat-bath~\cite{kastoryano2016quantum} 
      & $\mathcal O(\log(n/\varepsilon))$ \\[2pt]
    Kochanowski et al.~\cite{kochanowski2025rapid}
      & 1D/Commuting
      & All
      & Davies~\cite{davies1976quantum}
      & $\mathcal O(\log(n/\varepsilon))$ \\[2pt]
    Rouz\'{e} et al.~\cite{rouze2026optimal}
      & $D$-dimensional lattice/General
      & High
      & CKG~\cite{chen2023an}
      & $\mathcal O(\log(n/\varepsilon))$ \\[2pt]
    Bakshi et al.~\cite{bakshi2026a}
      & Bounded-degree graphs/General
      & High 
      & CKG-inspired
      & $\mathcal O(\log(n/\varepsilon))$ \\[2pt]
    \v{S}m\'{\i}d et al.~\cite{smid2025rapid}
      & Lattice/Weakly interacting spins
      & All
      & DLL~\cite{ding2025efficient} 
      & $\mathcal O(\log(n/\varepsilon))$ \\[2pt]
    Bergamaschi~\cite{bergamaschi2026all}
      & $\cdot$/All-to-all $k$-local
      & High
      & CKG~\cite{chen2023an}
      & $\mathcal O(n+\log(1/\varepsilon))$ \\[2pt]
    Bergamaschi-Chen~\cite{bergamaschi2026fast}
      & 1D/General
      & All
      & CKG~\cite{chen2023an}
      & $\mathcal O(n+\log(1/\varepsilon))$ \\
    \midrule
    {This work}
      & {1D/General}
      & {All}
      & {Constant-block Petz heat-bath}
      & ${\mathcal O(\log(n/\varepsilon))}$ \\
    \bottomrule
  \end{tabularx}
  \caption{A partial list of fast and rapid mixing-time results for spin systems (see Section~\ref{subsec:related-works} for a more thorough list). All temperature means every fixed finite inverse temperature $\beta$. General denotes bounded finite-range spin systems without commutativity assumptions.}
  \label{tab:related-mixing}
  \vspace{2pt}
  \begin{minipage}{\linewidth}
    \footnotesize
    \textsuperscript{a}At temperatures satisfying uniform mixing and strong quasi-factorization assumptions.
  \end{minipage}
\end{table}

The normalization $1/\ell$ is chosen so that each site is updated at total rate one. For a classical Hamiltonian, i.e., diagonal in a fixed on-site product basis, this dynamics reduces to the usual classical block heat-bath Glauber dynamics~\cite{levin2017markov,martinelli1999lectures,chen2023optimalmixing}, where the state of the block is discarded and resampled from the Gibbs distribution conditioned on its complement. Equation~\eqref{eq:block-petz-generator-intro} therefore provides its quantum analogue, with the classical conditional resampling replaced by Petz recovery.

Our first main result shows that a constant block size $\ell_\star$ is sufficient for rapid mixing. 

\begin{theorem}[Rapid mixing of quantum spin chains at all temperatures]
\label{thm:main}
Fix $\beta, d < \infty$. There exists a constant block length $\ell_\star = \mathcal{O}(1)$, independent of the system size $n$, such that the heat-bath dynamics generated by $\mathcal{L}_{\ell_\star}$ in Eq.~\eqref{eq:block-petz-generator-intro} mixes rapidly to the Gibbs state $\sigma_\beta$ in time 
\begin{align}
t_{\mathrm{mix}}(\varepsilon) = \mathcal O\left(\log\frac{n}{\varepsilon} \right).
\end{align}
\end{theorem}

Table~\ref{tab:related-mixing} compares our main result with selected mixing-time bounds (further related work is discussed below). To our knowledge, this is the first all-temperature rapid-mixing result for any finite-range, bounded-strength 1D Hamiltonian. As mentioned, the sampler is also quasi-local and satisfies KMS detailed balance, making it a natural object for further study.

Our second main result facilitates this further investigation. Specifically, we give a quantum algorithm that implements the Lindbladian dynamics for time $t$ using quantum circuits of depth $\mathcal{O}(\operatorname{poly}(t + \log(n))$ for fixed error. Combining this result with Theorem~\ref{thm:main} yields $\operatorname{polylog}(n)$-depth quantum circuits for preparing the Gibbs state.

\begin{theorem}[Polylogarthmic-depth quantum circuits for Gibbs state preparation, Corollary~\ref{cor:gibbs-preparation}]
\label{thm:main-algo}
For $\beta, d < \infty$. There exists a randomized quantum algorithm which prepares the Gibbs state $\sigma_\beta$ in trace-distance $\varepsilon$ using only $\mathcal{O}\left(\operatorname{polylog}(n/\varepsilon)\right)$ depth of one- and two- qubit gates (with ancillas) on a 1D nearest-neighbor geometry.
\end{theorem}

This matches the circuit-depth bound achieved in Ref.~\cite{bergamaschi2026fast}. Although the dynamics in Eq.~\eqref{eq:block-petz-generator-intro} are defined using the target Gibbs state $\sigma_\beta$, our algorithm requires only a classical description of the local Hamiltonian terms. The key idea is that each block update can be approximated to high accuracy using only the Hamiltonian on a window of $\mathcal O(\log(n/\varepsilon))$ sites. Each approximate update can be implemented with $\operatorname{polylog}(n/\varepsilon)$ gates. We then exploit the locality of these updates to parallelize the dynamics, obtaining the claimed depth at the mixing time.


\subsection{Related works}
\label{subsec:related-works}

\paragraph{Quantum Gibbs samplers.} One of the earliest examples of a quantum Gibbs sampler is the Davies generator, which naturally arises in the weak-coupling limit of a quantum system interacting with a thermal bath~\cite{davies1974markovian,davies1976quantum}. Another approach generalizes classical heat-bath Glauber dynamics, which updates a single site by resampling its spin from the Gibbs distribution conditioned on all other spins~\cite{glauber1963timedependent}. Following this idea, Ref.~\cite{kastoryano2016quantum} constructed a quantum heat-bath generator for commuting local Hamiltonians using minimal conditional expectations. Their generator coincides with the single-site ($\ell=1$) case of the dynamics considered here in Eq.~\eqref{eq:block-petz-generator-intro}. Subsequent developments include finite-time analyses and algorithmic implementations of Davies-type thermalization~\cite{chen2021fast,rall2023thermal,chen2023quantum}, as well as dissipative sampling schemes based on conditionally stopped processes~\cite{zhang2023dissipative}. For noncommuting Hamiltonians, efficiently implementable Lindbladians satisfying exact KMS detailed balance have been constructed and analyzed~\cite{chen2023an,ding2025efficient,bakshi2026a,chen2025quantum}.
Recent work has developed Gibbs-state preparation protocols based on repeated interactions with a single-qubit bath~\cite{ding2025simple}, with polynomial-time guarantees for several noncommuting systems~\cite{slezak2026polynomial}. Related discrete-time approaches include the original quantum Metropolis algorithm~\cite{temme2011quantum}, quantum generalizations of Glauber and Metropolis dynamics~\cite{gilyen2024quantum}, and quantum Metropolis sampling via weak measurements~\cite{jiang2024quantum}.

\paragraph{Mixing time results for commuting Hamiltonians.}
Ref.~\cite{kastoryano2016quantum} related a system-size-independent spectral gap of the Gibbs sampler to strong clustering of the Gibbs state, yielding fast mixing for 1D commuting Hamiltonians for both Davies and Petz heat-bath samplers. For the same single-site Petz heat-bath dynamics, Ref.~\cite{bardet2021on} proved a system-size-independent modified log-Sobolev constant, and hence rapid mixing, under certain assumptions on the Gibbs state. For Davies generators, Refs.~\cite{bardet2023rapid,bardet2024entropy} established all-temperature rapid mixing for finite-range, translation-invariant commuting spin chains, initially with a logarithmically decreasing MLSI constant. Ref.~\cite{kochanowski2025rapid} subsequently proved the optimal rapid mixing, also covering non-translation-invariant commuting chains. Thus, optimal all-temperature mixing in 1D was previously known for commuting interactions, but those arguments do not cover the noncommuting Hamiltonians. Beyond these 1D results, fast and rapid mixing have also been established for various other commuting systems ~\cite{capel2020modified,depalma2022quantum,lucia2023thermalization,stengele2026css,stengele2026abelian,hwang2026gibbs}.

\paragraph{Mixing time results beyond the commuting setting.}
Proving mixing time bounds for noncommuting systems is more challenging and is still an active area of research. For local Hamiltonians on lattices, Ref.~\cite{rouze2025efficient} proved a system-size-independent spectral gap for the CKG Lindbladian~\cite{chen2023an} at high temperature, yielding fast mixing, which was later strengthened to rapid mixing~\cite{rouze2026optimal}. Related rapid-mixing results for closely related high-temperature Gibbs samplers were obtained via a quantum Dobrushin condition~\cite{bakshi2026a} and in the presence of arbitrarily strong on-site fields~\cite{bakshi2026rapid}. At all finite temperatures, structured results include a fast mixing for weakly interacting fermions~\cite{tong2025fast} and rapid mixing for noninteracting and weakly interacting spins among others~\cite{smid2025rapid}. More recently, Ref.~\cite{bergamaschi2026fast} proved a fast mixing of the CKG Lindbladian of every finite-range 1D Hamiltonian at any finite temperature. Ref.~\cite{bergamaschi2026all} further proved a system-size-independent gap, and hence fast mixing, for the CKG Lindbladian of all-to-all $k$-local Hamiltonians at high temperature.

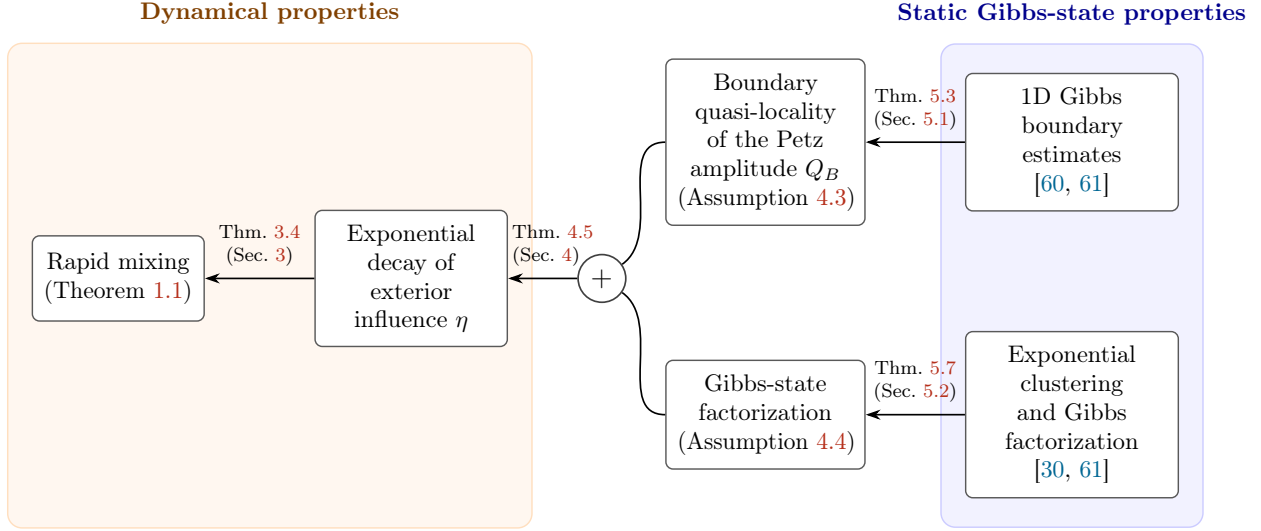
\begin{figure}[t]
\centering
\resizebox{\linewidth}{!}{%
\begin{tikzpicture}[
    box/.style={
        draw=black!65,
        fill=white,
        line width=0.6pt,
        rounded corners=3pt,
        align=center,
        text width=2.5cm,
        minimum height=1.25cm,
        inner sep=6pt,
        font=\small
    },
    arrow/.style={
        -{Stealth[length=2.2mm,width=1.5mm]},
        line width=0.7pt
    },
    branch/.style={
        line width=0.7pt
    },
    annotation/.style={
        midway,
        above=3pt,
        align=center,
        font=\scriptsize,
        inner sep=2pt
    },
    heading/.style={
        anchor=south,
        yshift=4pt,
        font=\small\bfseries
    }
]

\node[box, text width=2.1cm] (mix) at (0,0)
    {Rapid mixing\\
     (Theorem~\ref{thm:main})};

\node[box, text width=2.4cm] (influence) at (4.3,0)
    {Exponential decay of exterior influence $\eta$};

\node[
    circle,
    draw=black!75,
    fill=white,
    line width=0.7pt,
    minimum size=7mm,
    inner sep=0pt,
    font=\large
] (plus) at (7.1,0) {$+$};

\node[box] (ql) at (9.5,2.0)
    {Boundary quasi-locality\\
     of the Petz amplitude $Q_B$\\
     (Assumption~\ref{asm:QL})};

\node[box] (gf) at (9.5,-2.0)
    {Gibbs-state factorization\\
     (Assumption~\ref{asm:GF})};

\node[box, text width=2.7cm] (boundary) at (14.0,2.0)
    {1D Gibbs boundary estimates\\
     \cite{scalet2025classical,bluhm2022exponential}};

\node[box, text width=2.7cm] (clustering) at (14.0,-2.0)
    {Exponential clustering
     and Gibbs factorization\\
     \cite{bergamaschi2026fast,bluhm2022exponential}};

\draw[arrow] (influence.west) --
    node[annotation]
    {Thm.~\ref{thm:influence-to-curvature}\\
     (Sec.~\ref{sec:curvature})}
    (mix.east);

\draw[arrow] (plus.west) --
    node[annotation]
    {~~~Thm.~\ref{thm:static-to-influence}\\
     ~(Sec.~\ref{sec:petz-influence})}
    (influence.east);

\draw[branch]
    (ql.west)
    to[out=180,in=35]
    (plus.35);

\draw[branch]
    (gf.west)
    to[out=180,in=-35]
    (plus.-35);

\draw[arrow] (boundary.west) --
    node[annotation]
    {Thm.~\ref{prop:verify-QL}\\
     (Sec.~\ref{sec:verify-QL})}
    (ql.east);

\draw[arrow] (clustering.west) --
    node[annotation]
    {Thm.~\ref{thm:verify-GF}\\
     (Sec.~\ref{sec:verify-GF})}
    (gf.east);

\begin{scope}[on background layer]

    \node[
        fit=(boundary)(clustering),
        fill=blue!5,
        draw=blue!20,
        line width=0.6pt,
        rounded corners=6pt,
        inner xsep=10pt,
        inner ysep=12pt
    ] (staticpanel) {};

    \node[
        fit={
            (mix)
            (influence)
            (mix.center |- staticpanel.north)
            (mix.center |- staticpanel.south)
        },
        fill=orange!6,
        draw=orange!25,
        line width=0.6pt,
        rounded corners=6pt,
        inner xsep=10pt,
        inner ysep=0pt
    ] (dynamicpanel) {};

\end{scope}

\node[
    heading,
    text=orange!50!black
] at (dynamicpanel.north)
    {Dynamical properties};

\node[
    heading,
    text=blue!55!black
] at (staticpanel.north)
    {Static Gibbs-state properties};

\end{tikzpicture}%
}
\caption{Proof outline of Theorem~\ref{thm:main}, read from right to left. The blue panel collects existing results on static properties of 1D Gibbs states \cite{bergamaschi2026fast,bluhm2022exponential,scalet2025classical}, which prove Assumptions~\ref{asm:QL} and~\ref{asm:GF}. Together, these imply the dynamical consequences shown in the orange panel: exponentially decaying exterior influence, exponential contraction in quantum $W_1$ distance, and rapid mixing. The paper is intentionally organized in the reverse order.}
\label{fig:rapid-mixing-proof-outline}
\end{figure}

\subsection{Informal overview of the proof}
\label{subsec:proof-overview}

In this subsection, we give a brief, informal overview of the proof and its main ideas. Figure~\ref{fig:rapid-mixing-proof-outline} summarizes its overall structure. We keep technical terminology to a minimum, deferring precise definitions and statements to later sections. We recommend reading this overview before getting into the weeds. Although most of the proof is carried out in the Heisenberg picture, we present this overview in the Schr\"odinger picture for simplicity, as the underlying ideas are the same.

Our goal is to bound the worst-case trace distance to $\sigma_\beta$. Rather than analyzing trace distance directly, we analyze the \emph{quantum Wasserstein distance of order $1$}, introduced in Ref.~\cite{de2021the}, which we abbreviate as the $W_1$ distance (Definition~\ref{def:W1-primal}). Since the $W_1$ distance upper bounds trace distance (Proposition~\ref{prop:W1-trace}), exponential contraction in $W_1$ yields the desired trace-distance bound (Eq.~\eqref{eq:desired-bound}). 

Trace distance quantifies how well two states can be distinguished by an optimal global measurement, but does not indicate how their discrepancy is distributed across sites. In contrast, the $W_1$ distance accounts for this spatial structure by decomposing the discrepancy into contributions associated with individual sites and summing their costs. Specifically, for two states $\rho$ and $\sigma_\beta$, consider decompositions
\begin{align}
\rho-\sigma_\beta=\sum_i X_i, \qquad X_i=X_i^\dagger, \qquad \tr_i[X_i]=0.
\end{align}
The condition $\tr_i[X_i]=0$ means that the contribution $X_i$ vanishes upon tracing out site $i$. In this sense, $X_i$ represents the difference between the two states anchored at $i$th site. The $W_1$ distance is the minimum of $\frac12\sum_i\|X_i\|_1$ over all such decompositions.

Now fix a site $i$ and consider a single block update $\Phi_B$. There are two cases. If $i\in B$, then $\Phi_B(X_i)=0$, because the update first traces out $B$ and $\tr_i[X_i]=0$. Thus, the discrepancy anchored at site $i$ is erased and $\|\Phi_B(X_i)\|_{W_1} = 0$. If $i\notin B$, imagine an ideal situation in which the update leaves this discrepancy unchanged, so that $\Phi_B(X_i)=X_i$, thus $\|\Phi_B(X_i)\|_{W_1} = \|X_i\|_{W_1}$. There are $\ell$ blocks containing $i$, each updated at rate $1/\ell$, so the total rate at which the discrepancy is erased is $\ell\times(1/\ell)=1$. The remaining $n-\ell$ blocks are updated at total rate $(n-\ell)/\ell$, but leave $X_i$ unchanged in this ideal scenario. Consequently, each of the $\ell$ blocks containing $i$ contributes $-X_i/\ell$, while all other blocks contribute zero, giving $\mathcal L_\ell(X_i)=-X_i$, and hence $e^{t\mathcal L_\ell}(X_i)=e^{-t}X_i$. Summing over $i$, we obtain $\|e^{t\mathcal{L}_{\ell}}(\rho - \sigma)\|_{W_1} =e^{-t}\|\rho-\sigma\|_{W_1} \leq n e^{-t}$, and thus rapid mixing result.

In general, however, $\Phi_B$ affects and changes $X_i$ when $i\notin B$. Indeed, the update in Eq.~\eqref{eq:block-petz-generator-intro} involves the target global Gibbs state $\sigma_\beta$ and its marginal, so its action can extend beyond $B$. We quantify the possible increase in the $W_1$ distance through an \emph{exterior influence decay function} $\eta$ (Definition~\ref{def:exterior-influence} and Proposition~\ref{prop:single-discrepancy-influence}):
\begin{align}
\|\Phi_B(X_i)\|_{W_1}
\leq
\|X_i\|_{W_1}
+\eta\bigl(\operatorname{dist}(i,B)\bigr)\|X_i\|_{W_1}.
\end{align}
The key is to show that this additional cost decays with the distance between $i$ and $B$.

To understand why such decay is natural, decompose $H=H_B+H_{B^c}+V_{\partial B}$, where $H_B$ and $H_{B^c}$ contain the interactions supported entirely within $B$ and $B^c$, respectively, and $V_{\partial B}$ contains the interactions crossing their boundary. The block update admits the exact representation (Eq.~\eqref{eq:petz-amplitude-schrodinger})
\begin{align}
\Phi_B(\rho) = Q_B\bigl(\tau_B\otimes\tr_B[\rho]\bigr)Q_B^\dagger,
\qquad  Q_B =\sigma_\beta^{1/2}
\left(\tau_B^{-1/2}\otimes\sigma_{\beta,B^c}^{-1/2}\right),
\qquad \tau_B=\frac{e^{-\beta H_B}}{\tr[e^{-\beta H_B}]}.
\end{align}
This expression separates the update into replacement of the block by its isolated Gibbs state $\tau_B$ and a correction, encoded by the \emph{conditional Petz amplitude} $Q_B$, that accounts for the coupling across the boundary.

If $V_{\partial B} = 0$, the Gibbs state factorizes as $\sigma_\beta=\tau_B\otimes\sigma_{\beta,B^c}$, so $Q_B=I$ and the update simply resets $B$ to $\tau_B$. Although this reset need not leave $X_i$ unchanged, it cannot increase its $W_1$ norm when $i\notin B$. When $V_{\partial B}\neq0$, we expect the correction $Q_B$ to be concentrated near the boundary, with its influence weakening farther away. We make this intuition precise by showing that $Q_B$ can be approximated exponentially well by operators supported near the boundary of $B$. Together with a decay estimate for correlations within $\tau_B$, this boundary localization yields an exponentially decaying exterior influence function $\eta$.

To make this intuition rigorous, we replace $Q_B$ by an approximation $\widehat Q_{B,r}$ supported within distance $r$ of the boundary of $B$, and define
\begin{align}
\Psi_{B,r}^\dagger(\rho) :=\widehat Q_{B,r} \bigl(\tau_B\otimes\tr_B[\rho]\bigr)
\widehat Q_{B,r}^\dagger.
\end{align}
We initially assume exponential boundary quasi-locality of $Q_B$ (Assumption~\ref{asm:QL}), which essentially states $\|Q_B-\widehat Q_{B,r}\|\leq Ce^{-\mu r}$ while ensuring that $\Psi_{B,r}^\dagger$ is a quantum channel. We prove this assumption later in Section~\ref{sec:verify-QL} (Theorem~\ref{prop:verify-QL}).

For an exterior site $i$, choose $R$ so that the approximating amplitude does not reach $i$, and rewrite
\begin{align}
\Phi_B(X_i) = \Psi_{B,R}^\dagger(X_i) + \bigl(\Phi_B-\Psi_{B,R}^\dagger\bigr)(X_i).
\end{align}
The first term comes from a channel that acts trivially on site $i$. It therefore preserves the condition $\tr_i[X_i]=0$ and, by trace-norm contractivity, satisfies $\|\Psi_{B,R}^\dagger(X_i)\|_{W_1} \leq\|X_i\|_{W_1}$. Thus, any increase in the $W_1$ norm comes from the second term.

Controlling this remainder requires more than operator-norm accuracy of the amplitudes, as converting that accuracy directly into a $W_1$ bound can introduce a factor depending on the system size. We therefore use a second assumption, which states that the boundary and a distant interior region of $\tau_B$ approximately factorize, with an error that decays exponentially with their separation (Assumption~\ref{asm:GF}, proved in Theorem~\ref{thm:verify-GF}). Together with earlier construction, this lets us control the $W_1$ cost of each successive enlargement of the boundary neighborhood, $\|(\Psi_{B,r}^\dagger-\Psi_{B,r-1}^\dagger)(X_i)\|_{W_1} \lesssim re^{-\mu r}\|X_i\|_{W_1}$. Telescoping these increments beyond radius $R$ gives
\begin{align}
\|(\Phi_B-\Psi_{B,R}^\dagger)(X_i)\|_{W_1} \lesssim \sum_{r>R} r e^{-\mu r}\|X_i\|_{W_1}.
\end{align}
Combining this estimate with the nonexpansion of $\Psi_{B,R}^\dagger$ gives
\begin{align}
\|\Phi_B(X_i)\|_{W_1}
\leq \bigl(1+\eta(d)\bigr)\|X_i\|_{W_1},
\qquad
\eta(d)\lesssim\sum_{r>R(d)}r e^{-\mu r},
\qquad d=\operatorname{dist}(i,B),
\end{align}
where $R(d)$ is the largest radius for which the support of $\widehat Q_{B,R(d)}$ still avoids $i$, and $\eta(d)$ bounds the relative increase in the $W_1$ norm caused by an update at distance $d$. Since $R(d)$ differs from $d$ by at most a fixed constant, $\eta$ decays exponentially.

Finally, in 1D, at most two blocks lie at each distance $d$ from $i$. Since $a_\eta:=2\sum_{d\ge1}\eta(d)$ is finite and independent of $n$ and $\ell$, the total exterior contribution is at most $a_\eta/\ell$. Therefore, choosing $\ell>a_\eta$ makes the removal rate dominate the exterior influence discussed earlier, yielding exponential contraction and hence rapid mixing.

The rest of the paper is organized as follows.
\begin{itemize}
    \item In Section~\ref{sec:background}, we introduce the quantum $W_1$ distance and the block heat-bath dynamics, together with their basic properties.
    
    \item In Section~\ref{sec:curvature}, we show that if $\sum_{d\geq1}\eta(d)$ is bounded by a constant independent of the system size and block length, then a sufficiently large constant block length $\ell_{\star} = \mathcal{O}(1)$ yields rapid mixing (Theorem~\ref{thm:influence-to-curvature} and Corollary~\ref{cor:abstract-trace-mixing}).

    \item In Section~\ref{sec:petz-influence}, we prove that $\eta$ decays exponentially with the distance from the site to the updated block, assuming boundary quasi-locality of the Petz amplitude and boundary-to-interior Gibbs factorization (Assumptions~\ref{asm:QL} and~\ref{asm:GF}; Theorem~\ref{thm:static-to-influence}).

    \item In Section~\ref{sec:gibbs-inputs}, we prove that both assumptions hold for the one-dimensional Gibbs states considered here at every fixed finite temperature (Theorems~\ref{prop:verify-QL} and~\ref{thm:verify-GF}).

    \item In Section~\ref{sec:dynamics-implementation}, we present an explicit algorithm for simulating the dynamics and derive a polylogarithmic-depth circuit for Gibbs-state preparation, proving Theorem~\ref{thm:main-algo}.
    
\end{itemize}

\section{Background}
\label{sec:background}

\paragraph{Notation.} For a finite region $S\subseteq\Lambda_n$, we write $\mathcal A_S=\mathcal B((\mathbb C^d)^{\otimes |S|})$, and identify $\mathcal A_S$ with operators on the full chain that act trivially on $S^c$. We denote by $\tr_S$ the partial trace over $S$, and by $S^c=\Lambda_n\setminus S$ its complement. The operator, trace, and diamond norms are denoted by $\|\cdot\|$, $\|\cdot\|_1$, and $\|\cdot \|_\diamond$, respectively. For a subset $S\subseteq\Lambda_n$ and $r\geq0$, we denote its radius-$r$ neighborhood by $N_r(S) := \left\{ i\in\Lambda_n: \operatorname{dist}(i,S)\leq r \right\}$. The distance between subsets $S$ and $T$ of the periodic chain is the graph distance, $\mathrm{dist}(S,T):=\min_{i\in S,j\in T}\mathrm{dist}(i,j)$. We write $\mathcal B_{n,\ell}$ for the family of all cyclic intervals $B\subseteq\Lambda_n$ of length $\ell$.

The evolution generated by $\mathcal L$ in the Schr"odinger picture is denoted by $e^{t\mathcal L}$, while $e^{t\mathcal L^\dagger}$ denotes its Heisenberg-picture adjoint, defined by
\begin{align}
    \tr[e^{t\mathcal L}(\rho)A]=\tr[\rho e^{t\mathcal L^\dagger}(A)].
\end{align}
Similarly, for a quantum channel $\Phi$, we denote its Heisenberg-picture adjoint as $\Phi^\dagger$. 

If $\sigma$ is the unique stationary state of $e^{t\mathcal L}$, we define its mixing time by
\begin{align}
t_{\mathrm{mix}}(\varepsilon):=\inf\left\{t\geq0:\frac12\|e^{t\mathcal L}(\rho)-\sigma\|_1\leq\varepsilon, ~ \forall \rho\right\}.
\label{eq:mixing-time}
\end{align}

\subsection{Quantum Wasserstein distance of order $1$}

We introduce the quantum Wasserstein distance of order $1$ introduced in Ref.~\cite{de2021the}.

\begin{definition}[Quantum Wasserstein distance~{\cite[Definition~6]{de2021the}}]
\label{def:W1-primal}
For a traceless Hermitian operator $X\in\mathcal A_{\Lambda_n}$, define
\begin{align}
\|X\|_{W_1}:=\frac{1}{2}\inf\left\{\sum_{i\in\Lambda_n}\|X_i\|_1:X=\sum_{i\in\Lambda_n}X_i, X_i=X_i^\dagger, \tr_i[X_i]=0\right\}.
\label{eq:w1-norm}
\end{align}
For two states $\rho$ and $\sigma$, their quantum Wasserstein distance of order $1$ is $\|\rho-\sigma\|_{W_1}$.
\end{definition}
Informally, each $X_i$ represents a component of the discrepancy that becomes invisible once site $i$ is traced out, as expressed by $\tr_i[X_i]=0$, while $\frac12\|X_i\|_1$ quantifies its magnitude. Thus, $\|X\|_{W_1}$ is the minimum total cost of decomposing $X$ into such site-associated changes. In particular, for $X=\rho-\sigma$, it captures not only the overall distinguishability of the states, but also how broadly their discrepancy is distributed across the sites.

The quantum $W_1$ distance admits a dual formulation in terms of a local oscillation.

\begin{definition}[Local oscillations and Lipschitz constant]
\label{def:local-oscillation}
For an observable $A\in\mathcal A_{\Lambda_n}$, define its
oscillation at site $i$ by
\begin{align}
\delta_i(A) := 2\inf_{K\in\mathcal A_{i^c}}\|A-K\|,
\end{align}
and its associated Lipschitz constant by
\begin{align}
\|A\|_{L}:=\max_{i\in\Lambda_n}\delta_i(A).
\label{eq:lip-norm}
\end{align}
\end{definition}

\begin{proposition}[Dual form~{\cite[Proposition~9]{de2021the}}]
\label{prop:w1-duality}
For every traceless Hermitian operator $X$,
\begin{align}
\|X\|_{W_1} = \sup\left\{ \tr[AX]: A=A^\dagger,  \|A\|_{L}\leq 1 \right\}.
\label{eq:w1-duality}
\end{align}
\end{proposition}

Thus, the quantum $W_1$ norm is the largest expectation-value difference detectable by observables whose sensitivity to changing any single site is at most one. This locality constraint distinguishes it from the trace norm, whose dual constraint is instead the global bound $\|A\|\leq1$. For example, for $n$ qubits let $A=\sum_{j=1}^n Z_j$. Then $\|A\|_{L}=2$, independently of $n$, whereas $\|A\|=n$. Hence $A/2$ is an admissible witness in Eq.~\eqref{eq:w1-duality}. Using this witness for $\rho=\ket{0^n} \bra{0^n}$ and $\sigma=\ket{1^n} \bra{1^n}$ gives $\|\rho-\sigma\|_{W_1}=n$, while their trace distance is $1$. 
A more extreme example is provided by comparing $\lvert 0\rangle^{\otimes n}$ with $\lvert 1\rangle\lvert 0\rangle^{\otimes(n-1)}$ and $\lvert 1\rangle^{\otimes n}$. Both pairs have trace distance $1$, since the corresponding states are perfectly distinguishable. In contrast, their quantum $W_1$ distances are $1$ and $n$, respectively, reflecting whether the difference involves a single site or the entire system.

The quantum $W_1$ distance can be related to trace distance. 

\begin{proposition}[Comparison with the trace norm~{\cite[Proposition~2]{de2021the}}]
\label{prop:W1-trace}
For every traceless Hermitian operator $X$,
\begin{align}
\frac{1}{2}\|X\|_1\leq\|X\|_{W_1}\leq\frac{n}{2}\|X\|_1.
\label{eq:w1-trace-comparison}
\end{align}
\end{proposition}

We next quantify mixing through the contraction of the $W_1$ distance under the dynamics.

\begin{definition}[$W_1$ curvature]
\label{def:W1-curvature}
We say that $e^{t\mathcal L}$ has $W_1$ curvature at least $\gamma>0$ if
\begin{align}
\|e^{t\mathcal L}(\rho) - e^{t\mathcal L}(\sigma)\|_{W_1} \leq e^{-\gamma t}\|\rho - \sigma\|_{W_1}
\end{align}
for all states $\rho,\sigma$ and all $t\geq0$.
\end{definition}

If $\sigma$ is stationary, positive $W_1$ curvature immediately yields
\begin{align}
\label{eq:desired-bound}
\frac{1}{2}\|e^{t\mathcal L}(\rho)-\sigma\|_1 \leq \|e^{t\mathcal L}(\rho)-\sigma\|_{W_1} \leq e^{-\gamma t}\|\rho - \sigma\|_{W_1} \leq n e^{-\gamma t},
\end{align}
where the first and last inequalities follow from Eq.~\eqref{prop:W1-trace}. Consequently, $t_{\mathrm{mix}}(\varepsilon)\leq\frac{1}{\gamma}\log\left(\frac{n}{\varepsilon}\right)$. Thus, rapid mixing follows once a positive $W_1$ curvature bound is established uniformly in the system size.

\subsection{Quantum block heat-bath dynamics}
\label{subsec:heat-bath-dynamics}

The classical heat-bath Glauber dynamics updates a region by discarding its current configuration and resampling it from the Gibbs conditional distribution given its complement. Motivated by the quantum heat-bath constructions of Refs.~\cite{kastoryano2016quantum,bardet2021on}, we use the Petz recovery map as the noncommutative analogue of conditional resampling.

\begin{definition}[Petz recovery map]
Let $B\subseteq\Lambda_n$, and let $\sigma_{\beta,B^c}=\tr_B[\sigma_\beta]$. The Petz recovery map associated with the partial trace over $B$ and reference state $\sigma_\beta$ is
\begin{align}
\mathcal R^{\sigma_\beta}_{B^c\to\Lambda_n}(X) := \sigma_\beta^{1/2} \left( I_B\otimes \sigma_{\beta,B^c}^{-1/2} X \sigma_{\beta,B^c}^{-1/2} \right) \sigma_\beta^{1/2}.
\label{eq:petz-recovery}
\end{align}
\end{definition}

The corresponding heat-bath update first erases the degrees of freedom in $B$ and then reconstructs them using the Gibbs reference state.

\begin{definition}[Petz block heat-bath dynamics]
\label{def:petz-heat-bath-block-dynamics}
For a block $B$, define the heat-bath update
\begin{align}
\Phi_B := \mathcal R^{\sigma_\beta}_{B^c\to\Lambda_n}\circ\tr_B.
\label{eq:block-petz-map}
\end{align}
For a fixed block length $\ell$, the corresponding continuous-time
generator is
\begin{align}
\mathcal L_\ell = \frac1\ell \sum_{B\in\mathcal B_{n,\ell}} (\Phi_B-\mathcal I).
\label{eq:block-petz-generator}
\end{align}
\end{definition}


We first state that, for any fixed block length $\ell$ independent of the system size, this Lindbladian is quasi-local, which will be essential for the efficient implementation of the dynamics later. We defer the proof to Section~\ref{sec:dynamics-implementation} and Appendix~\ref{app:finite-window}.

\begin{lemma}[Quasi-locality of $\mathcal L_\ell$]
\label{lem:quasi-locality}
Fix $\beta, d, \ell$, and let $\mathcal L_\ell=\sum_B\mathcal L_B$, where $\mathcal L_B=\ell^{-1}(\Phi_B-I)$. For every $r\ge0$, each $\mathcal L_B$ admits a Lindbladian approximation $\mathcal L_B^{(r)}$ supported on $N_r(B)$ such that
\begin{align}
\|\mathcal L_B-\mathcal L_B^{(r)}\|_\diamond \le C e^{-\mu r},
\end{align}
where $C,\mu>0$ are indepedent of the system size. 
\end{lemma}

\begin{proof}
In Lemma~\ref{lem:finite-window-approximation}, we show that each $\Phi_B$ is quasi-local by constructing a channel $\Phi_B^{(r)}$ supported on $N_r(B)$ such that $\|\Phi_B-\Phi_B^{(r)}\|_\diamond\le Ce^{-\mu r}$. Hence, taking $\mathcal L_B^{(r)} = \ell^{-1}(\Phi_B^{H_{N_r(B)}}\otimes I_{N_r(B)^c}-I)$ and applying Lemma~\ref{lem:finite-window-approximation} give the proof. 
\end{proof}

The next elementary facts ensure that Eq.~\eqref{eq:block-petz-generator} is KMS detailed-balanced.

\begin{proposition}[KMS-detailed-balanced condition of $\mathcal{L}_\ell$]
\label{prop:petz-basic}
For every block $B$, the map $\Phi_B$ is CPTP and $\Phi_B(\sigma_\beta)=\sigma_\beta$. Consequently, $\mathcal L_\ell(\sigma_\beta)=0$, and $\mathcal L_\ell$ generates a CPTP semigroup $e^{t\mathcal L_\ell}$. Moreover, each $\Phi_B$ is KMS-reversible with respect to $\sigma_\beta$, and hence $\mathcal L_\ell$ satisfies KMS detailed balance with respect to $\sigma_\beta$.
\end{proposition}

\begin{proof}
Since $\sigma_\beta$ is strictly positive for every finite, its marginal $\sigma_{\beta,B^c}$ is also strictly positive. Complete positivity of $\mathcal R^{\sigma_\beta}_{B^c\to\Lambda_n}$ follows directly from Eq.~\eqref{eq:petz-recovery}. Moreover,
\begin{align}
\tr\left[ \mathcal R^{\sigma_\beta}_{B^c\to\Lambda_n}(X) \right] &= \tr\left[ \sigma_{\beta,B^c}^{-1/2} X\sigma_{\beta,B^c}^{-1/2} \sigma_{\beta,B^c} \right] =\tr[X],
\end{align}
so the Petz map is trace preserving. Since the partial trace is CPTP,
$\Phi_B$ is CPTP. Furthermore,
\begin{align}
\Phi_B(\sigma_\beta)=\sigma_\beta^{1/2}\left(I_B\otimes\sigma_{\beta,B^c}^{-1/2}\sigma_{\beta,B^c}\sigma_{\beta,B^c}^{-1/2}\right)\sigma_\beta^{1/2}=\sigma_\beta.
\end{align}
It follows immediately that $\mathcal L_\ell(\sigma_\beta)=0$.

Let $N_\ell=|\mathcal B_{n,\ell}|$ and define $\overline{\Phi}_\ell := \frac{1}{N_\ell} \sum_{B\in\mathcal B_{n,\ell}}\Phi_B$. Then $\overline{\Phi}_\ell$ is a quantum channel and $\mathcal L_\ell = \frac{N_\ell}{\ell} \left(\overline{\Phi}_\ell-\mathcal I\right)$. Hence, for every $t\geq0$,
\begin{align}
e^{t\mathcal L_\ell} = e^{-N_\ell t/\ell} \sum_{k=0}^{\infty} \frac{(N_\ell t/\ell)^k}{k!} \overline{\Phi}_\ell^{k}.
\end{align}
This is a convex combination of quantum channels and is therefore CPTP.

It remains to verify KMS condition. For a full-rank state $\omega$, define $\Gamma_\omega(X):=\omega^{1/2}X\omega^{1/2}$. Since the adjoint of the partial trace is $\tr_B^\dagger(Y)=I_B\otimes Y$, the block update can be written as
\begin{align}
\Phi_B=\Gamma_{\sigma_\beta}\circ\tr_B^\dagger\circ\Gamma_{\sigma_{\beta,B^c}}^{-1}\circ\tr_B.
\end{align}
Taking the adjoint gives $\Phi_B^\dagger=\tr_B^\dagger\circ\Gamma_{\sigma_{\beta,B^c}}^{-1}\circ\tr_B\circ\Gamma_{\sigma_\beta}$. Consequently, $\Gamma_{\sigma_\beta}^{-1}\circ\Phi_B\circ\Gamma_{\sigma_\beta}=\Phi_B^\dagger$, which is the KMS detailed-balance condition for $\Phi_B$. Therefore,
\begin{align}
\Gamma_{\sigma_\beta}^{-1} \circ\mathcal L_\ell \circ\Gamma_{\sigma_\beta} =
\frac1\ell\sum_{B\in\mathcal B_{n,\ell}} \left( \Gamma_{\sigma_\beta}^{-1} \circ\Phi_B \circ\Gamma_{\sigma_\beta} -\mathcal I \right) = \frac1\ell\sum_{B\in\mathcal B_{n,\ell}} (\Phi_B^\dagger-\mathcal I) = \mathcal L_\ell^\dagger.
\end{align}
Hence, $\mathcal L_\ell$ is KMS detailed balance with respect to $\sigma_\beta$.
\end{proof}


Recall that we denote by $\Phi_B^\dagger$ the (Hilbert-Schmidt) adjoint of $\Phi_B$, defined by 
\begin{align} 
\tr \left[\Phi_B(\rho)A\right] = \tr \left[\rho \Phi_B^\dagger(A)\right]. 
\end{align} 
Since our influence bounds are formulated in the Heisenberg picture, we will frequently use the adjoint map $\Phi_B^\dagger$. Explicitly, it admits the following form.

\begin{proposition}[Explicit form of $\Phi_B^\dagger$]
\label{prop:block-petz-adjoint}
For every observable $A$, the adjoint block update is 
\begin{align}
\Phi_B^\dagger(A)=I_B\otimes\bigl[\sigma_{\beta,B^c}^{-1/2}\tr_B(\sigma_\beta^{1/2}A\sigma_\beta^{1/2})\sigma_{\beta,B^c}^{-1/2}\bigr].
\end{align}
In particular, $\Phi_B^\dagger(A)$ acts trivially on $B$.
\end{proposition}
\begin{proof}
For every operator $X$, cyclicity of the trace and the defining property of the partial trace give
\begin{align}
\tr[\Phi_B(X)A] &= \tr\left[\left(I_B\otimes\sigma_{\beta,B^c}^{-1/2}\tr_B[X]\sigma_{\beta,B^c}^{-1/2}\right)\sigma_\beta^{1/2}A\sigma_\beta^{1/2}\right] \nonumber \\
&= \tr_{B^c}\left[\tr_B[X] \sigma_{\beta,B^c}^{-1/2}\tr_B\left(\sigma_\beta^{1/2}A\sigma_\beta^{1/2}\right)\sigma_{\beta,B^c}^{-1/2}\right] \nonumber \\
&= \tr\left[X\left(I_B\otimes\sigma_{\beta,B^c}^{-1/2}\tr_B\left(\sigma_\beta^{1/2}A\sigma_\beta^{1/2}\right)\sigma_{\beta,B^c}^{-1/2}\right)\right].
\end{align}
Comparing the last expression with the definition of the Hilbert-Schmidt adjoint proves the claim.
\end{proof}

A basic feature of the dynamics, which will provide the contractive part of our argument, is that an update exactly removes any local discrepancy contained inside the refreshed block.

\begin{proposition}[Exact block killing]
\label{prop:exact-block-killing}
Let $i\in B$, and let $X=X^\dagger$ satisfy $\tr_i[X]=0$. Then $\Phi_B(X)=0$. Equivalently, for every observable $A$, $\delta_i \left(\Phi_B^\dagger(A)\right)=0$. 
\end{proposition}

\begin{proof}
Since $i\in B$, we have $\tr_B[X] = \tr_{B\setminus\{i\}}\left[\tr_i X\right] = 0$. Hence, $\Phi_B(X) = \mathcal R^{\sigma_\beta}_{B^c\to\Lambda_n} \bigl(\tr_B[X]\bigr) = 0$. The dual statement follows immediately from Proposition~\ref{prop:block-petz-adjoint}, since $\Phi_B^\dagger(A)$ acts trivially on $B$.
\end{proof}

\section{From block influence to rapid mixing}
\label{sec:curvature}

In this section, we isolate the dynamical part of the argument from the static part of the argument (that only depends on the properties of 1D Gibbs states) used later. The only input will be a summable bound on the influence of a block update at sites outside the updated block. Together with the exact block-killing property (Proposition~\ref{prop:exact-block-killing}), this bound implies positive $W_1$ curvature (Definition~\ref{def:W1-curvature}) by a simple one-dimensional counting argument.

To see why single-site influence is the relevant quantity, let $X=\rho-\sigma$ and consider an admissible $W_1$ decomposition $X=\sum_i X_i$, where $\tr_i[X_i]=0$. By linearity and the triangle inequality, controlling $\|e^{t \mathcal{L}_\ell}(\rho - \sigma)\|_{W_1}$ reduces to controlling $\|\Phi_B(X_i)\|_{W_1}$ for each pair $i,B$. By the duality of the $W_1$ distance (Proposition~\ref{prop:w1-duality}), this is in turn controlled by the local oscillation $\delta_i(\Phi_B^\dagger(A))$ (Definition~\ref{def:local-oscillation}). Hence, in this section we begin with which condition on $\delta_i(\Phi_B^\dagger(A))$ imposed imply constant-rate contraction of $\|e^{t \mathcal{L}_\ell}(\rho - \sigma)\|_{W_1}$, which results in the rapid mixing.


We assume throughout this section that $1\leq \ell<n$, so that $\mathcal B_{n,\ell}$ consists of $n$ cyclic intervals.

\subsection{Exterior influence decay function}
\label{subsec:exterior-influence}

The adjoint map $\Phi_B^\dagger$ acts trivially on the sites of $B$, but it need not act trivially on $B^c$.  We quantify this exterior action through the following function.

\begin{definition}[Exterior influence decay function]
\label{def:exterior-influence}
Let $\eta:\mathbb N\to[0,\infty)$.  We say that the block updates have exterior influence decay function $\eta$ if, for every block $B\in\mathcal B_{n,\ell}$, site $i\notin B$, and observable $A$,
\begin{align}
\delta_i\bigl(\Phi_B^\dagger(A)\bigr) & \leq \delta_i(A)+\eta\bigl(\mathrm{dist}(i,B)\bigr)\max_j \delta_j(A) \nonumber \\
& = \delta_i(A)+\eta\bigl(\mathrm{dist}(i,B)\bigr)\|A\|_L.
\label{eq:exterior-influence}
\end{align}
We call the decay function summable if
\begin{align}
a_\eta:=2\sum_{r=1}^{\infty}\eta(r)<\infty.
\label{eq:total-exterior-influence}
\end{align}
\end{definition}

The first term in the RHS of Eq.~\eqref{eq:exterior-influence} is the oscillation already present at site $i$, whereas the second term is the oscillation created by the update on $B$.  
The factor $2$ in Eq.~\eqref{eq:total-exterior-influence} anticipates that, in one dimension, there are at most two blocks at any fixed positive distance from a given site.

The above definition can be converted into a Schrödinger-picture estimate via the $W_1$ duality.

\begin{proposition}[Schr\"odinger picture of the influence decay function]
\label{prop:single-discrepancy-influence}
Suppose that the block updates have exterior influence decay function $\eta$. Let $X_i=X_i^\dagger$ satisfy $\tr_i[X_i]=0$. If $i\in B$, then $\Phi_B(X_i)=0$. If $i\notin B$, then
\begin{align}
\|\Phi_B(X_i)\|_{W_1}\leq \left(1+\eta\bigl(\mathrm{dist}(i,B)\bigr)\right)\|X_i\|_{W_1}.
\label{eq:single-discrepancy-influence}
\end{align}
\end{proposition}

\begin{proof}
The case $i\in B$ is Proposition~\ref{prop:exact-block-killing}. Suppose now that $i\notin B$. Since $\tr_i[X_i]=0$, the single-component decomposition is admissible in the definition of the $W_1$ norm. Together with Proposition~\ref{prop:W1-trace}, this gives $\|X_i\|_{W_1}=\frac12\|X_i\|_1$. Moreover, $\tr[X_iK]=0$ for every $K\in\mathcal A_{i^c}$. Hence, for every observable $A=A^\dagger$ with $\|A\|_L\leq1$,
\begin{align}
\bigl|\tr[X_i\Phi_B^\dagger(A)]\bigr|
&= \bigl|\tr\left[X_i\bigl(\Phi_B^\dagger(A)-K\bigr)\right]\bigr| \nonumber \\
&\leq 2\|X_i\|_{W_1}\|\Phi_B^\dagger(A)-K\| &&& \text{(H\"older's inequality)}.
\end{align}
Taking the infimum over $K\in\mathcal A_{i^c}$, we obtain
\begin{align}
\bigl|\tr[X_i\Phi_B^\dagger(A)]\bigr|
&\leq \|X_i\|_{W_1}\delta_i\left(\Phi_B^\dagger(A)\right) &&& \text{(Definition~\ref{def:local-oscillation})} \nonumber \\
&\leq \|X_i\|_{W_1}\left(\delta_i(A)+\eta(\operatorname{dist}(i,B))\|A\|_L\right) &&& \text{(Definition~\ref{def:exterior-influence})} \nonumber \\
&\leq \left(1+\eta(\operatorname{dist}(i,B))\right)\|X_i\|_{W_1},
\end{align}
where in the last line we used $\|A\|_L\leq1$, and hence $\delta_i(A)\leq1$. Using $W_1$ duality (Proposition~\ref{prop:w1-duality}) together with $\tr[A\Phi_B(X_i)]=\tr[X_i\Phi_B^\dagger(A)]$ gives Eq.~\eqref{eq:single-discrepancy-influence}.
\end{proof}

This separates the two effects of a block update: a block containing the anchor $i$ removes the discrepancy exactly, whereas a block missing $i$ can increase its $W_1$ norm only by an amount controlled by the exterior influence at that distance. To compare these two effects after averaging over all blocks, we next simply count how many blocks contain the anchor $i$ and how many can occur at each distance from it.


\begin{proposition}[Counting cyclic blocks]
\label{prop:block-counting}
Fix $i\in\Lambda_n$.  Exactly $\ell$ blocks in
$\mathcal B_{n,\ell}$ contain $i$.  Moreover, for every $r\geq1$, at most two
blocks $B\in\mathcal B_{n,\ell}$ satisfy $i\notin B$ and
$\mathrm{dist}(i,B)=r$.  Consequently,
\begin{align}
\sum_{\substack{B\in\mathcal B_{n,\ell}\\ i\notin B}}
\eta\bigl(\mathrm{dist}(i,B)\bigr)
\leq a_\eta.
\label{eq:total-block-influence}
\end{align}
\end{proposition}

\begin{proof}
A cyclic interval of length $\ell$ contains $i$ precisely when its left
endpoint is one of the $\ell$ sites from $i-\ell+1$ through $i$.  This proves
the first assertion.  If a block does not contain $i$, a nearest point of the
block can be reached from $i$ in one of the two directions along the cycle.
For a fixed positive distance $r$, there is at most one such block in each
direction.  Summing the corresponding influence over $r$ and using
Eq.~\eqref{eq:total-exterior-influence} proves
Eq.~\eqref{eq:total-block-influence}.
\end{proof}

Thus the number of blocks that erase a discrepancy is exactly $\ell$, while the total exterior influence contributed by all remaining blocks is bounded by $a_\eta$, as illustrated in Fig.~\ref{fig:euler-step-balance}. We now turn this comparison into a contraction estimate for a single Euler step of the dynamics.

\begin{figure}[t!]
    \centering
    \includegraphics[width=0.8\linewidth]{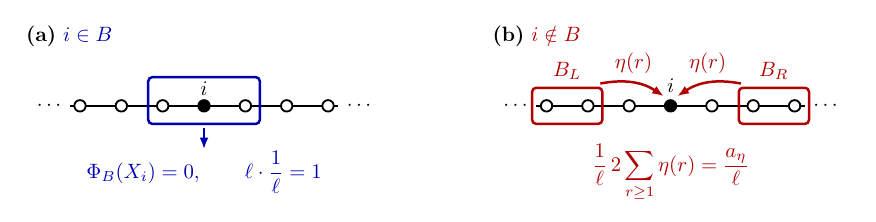}
    \caption{Erasure-versus-influence balance for an $i$-anchored discrepancy $X_i$. (a) The $\ell$ blocks containing $i$ erase $X_i$ exactly, with total rate $1$. (b) Blocks missing $i$ contribute an additional influence rate at most $a_\eta/\ell$.} 
    \label{fig:euler-step-balance}
\end{figure}

\subsection{Summable $\eta$ implies rapid mixing}
\label{subsec:one-step-contraction}

For $0\leq h\leq\ell/n$, define the Euler step $\mathcal T_h:=\mathcal I+h\mathcal L_\ell =\left(1-hn/\ell\right)\mathcal I +(h/\ell)\sum_{B\in\mathcal B_{n,\ell}}\Phi_B$. 

\begin{theorem}[Exterior influence implies positive $W_1$ curvature]
\label{thm:influence-to-curvature}
Suppose that the block updates have a summable exterior influence decay function $\eta$. If $\ell>a_\eta$, then the semigroup generated by $\mathcal L_\ell$ has $W_1$ curvature at least $\gamma_\ell:=1-a_\eta/\ell>0$. Equivalently, every traceless Hermitian operator $X$ and every $t\geq0$ satisfy
\begin{align}
\|e^{t\mathcal L_\ell}(X)\|_{W_1} \leq e^{-\gamma_\ell t}\|X\|_{W_1}.
\label{eq:semigroup-w1-contraction}
\end{align}
\end{theorem}

\begin{proof}
We first establish a contraction estimate for the Euler step. Let $X_i=X_i^\dagger$ satisfy $\tr_i[X_i]=0$. The $\ell$ blocks containing $i$ annihilate $X_i$ by Proposition~\ref{prop:exact-block-killing}. Applying Proposition~\ref{prop:single-discrepancy-influence} to the remaining blocks gives
\begin{align}
\|\mathcal T_h(X_i)\|_{W_1}
&\leq \left(1-\frac{hn}{\ell}\right)\|X_i\|_{W_1}+\frac h\ell\sum_{\substack{B\in\mathcal B_{n,\ell}\\i\notin B}}\|\Phi_B(X_i)\|_{W_1} \nonumber\\
&\leq \|X_i\|_{W_1}\left[1-\frac{hn}{\ell}+\frac h\ell\sum_{\substack{B\in\mathcal B\in\mathcal B_{n,\ell}\\i\notin B}}\left(1+\eta\bigl(\mathrm{dist}(i,B)\bigr)\right)\right] &&& \text{(Proposition~\ref{prop:single-discrepancy-influence})} \nonumber\\
&= \|X_i\|_{W_1}\left[1-h+\frac h\ell\sum_{\substack{B\in\mathcal B_{n,\ell}\\i\notin B}}\eta\bigl(\mathrm{dist}(i,B)\bigr)\right] \nonumber\\
&\leq (1-\gamma_\ell h)\|X_i\|_{W_1} &&& \text{(Proposition~\ref{prop:block-counting})}.
\label{eq:euler-single-discrepancy}
\end{align}

Now let $X=\sum_iX_i$ be any admissible $W_1$ decomposition. By linearity, the triangle inequality, and Eq.~\eqref{eq:euler-single-discrepancy},
\begin{align}
\|\mathcal T_h(X)\|_{W_1}\leq\sum_i\|\mathcal T_h(X_i)\|_{W_1}\leq(1-\gamma_\ell h)\sum_i\|X_i\|_{W_1}.
\end{align}
Since $\|X_i\|_{W_1}=\frac12\|X_i\|_1$ for every component with $\tr_i[X_i]=0$, the definition of the $W_1$ norm gives
\begin{align}
\inf_{\substack{X=\sum_iX_i\\X_i=X_i^\dagger,\ \tr_i[X_i]=0}}\sum_i\|X_i\|_{W_1}=\|X\|_{W_1}.
\end{align}
Taking the infimum over all admissible decompositions therefore yields
\begin{align}
\|\mathcal T_h(X)\|_{W_1}\leq(1-\gamma_\ell h)\|X\|_{W_1}.
\label{eq:euler-step-contraction}
\end{align}

Now fix $t\geq0$ and choose $m$ sufficiently large that $t/m\leq\ell/n$. Iterating Eq.~\eqref{eq:euler-step-contraction} gives
\begin{align}
\left\|\left(\mathcal I+\frac{t}{m}\mathcal L_\ell\right)^m(X)\right\|_{W_1}\leq\left(1-\frac{\gamma_\ell t}{m}\right)^m\|X\|_{W_1}.
\end{align}
As $m\to\infty$, the operator on the left converges to $e^{t\mathcal L_\ell}(X)$, while the scalar factor on the right converges to $e^{-\gamma_\ell t}$. This proves Eq.~\eqref{eq:semigroup-w1-contraction}. Taking $X=\rho-\sigma$ proves the claim.
\end{proof}

The counting above separates the rate-one erasure from the additional exterior influence. As illustrated in Figure~\ref{fig:euler-step-balance}, blocks containing $i$ erase the $i$-anchored discrepancy at total rate $1$, whereas blocks missing $i$ generate additional exterior influence at total rate at most $a_\eta/\ell$. Thus, $\ell>a_\eta$ makes erasure dominate influence, yielding the positive curvature rate $\gamma_\ell=1-a_\eta/\ell$.

Combining Theorem~\ref{thm:influence-to-curvature} with stationarity of the Gibbs state and the comparison between $W_1$ and trace distance yields the main result.

\begin{corollary}[Summable $\eta$ implies rapid mixing]
\label{cor:abstract-trace-mixing}
Under the assumptions of Theorem~\ref{thm:influence-to-curvature}, the Gibbs state $\sigma_\beta$ is the unique stationary state and
\begin{align}
\sup_\rho\frac12 \|e^{t\mathcal L_\ell}(\rho)-\sigma_\beta\|_1 \leq ne^{-\gamma_\ell t}.
\label{eq:abstract-trace-mixing}
\end{align}
Consequently,
\begin{align}
t_{\mathrm{mix}}(\varepsilon) \leq \frac1{\gamma_\ell}\log\frac n\varepsilon.
\label{eq:abstract-mixing-time}
\end{align}
\end{corollary}

\begin{proof}
Every block update fixes $\sigma_\beta$, so $e^{t\mathcal L_\ell}(\sigma_\beta)=\sigma_\beta$.  By
Proposition~\ref{prop:W1-trace} and
Theorem~\ref{thm:influence-to-curvature},
\begin{align}
\frac12\|e^{t\mathcal L_\ell}(\rho)-\sigma_\beta\|_1 \leq \|e^{t\mathcal L_\ell}(\rho)-\sigma_\beta\|_{W_1} \leq e^{-\gamma_\ell t}\|\rho-\sigma_\beta\|_{W_1} \leq ne^{-\gamma_\ell t}.
\end{align}
This proves Eq.~\eqref{eq:abstract-trace-mixing}, and solving $ne^{-\gamma_\ell t}\leq\varepsilon$ gives Eq.~\eqref{eq:abstract-mixing-time}. Finally, if $\sigma$ is another stationary state, then $\|\sigma-\sigma_\beta\|_{W_1}\leq e^{-\gamma_\ell t}\|\sigma-\sigma_\beta\|_{W_1}$ for every $t\geq0$. Since $\gamma_\ell>0$, this forces $\sigma=\sigma_\beta$.
\end{proof}

Corollary~\ref{cor:abstract-trace-mixing} completes the dynamical part of the proof.  The remainder of the paper is devoted to establishing $\eta$ function for the Petz block updates, ultimately showing that it is indeed summable.

\section{Exponential decay of exterior influence function}
\label{sec:petz-influence}

Section~\ref{sec:curvature} reduced rapid mixing to a summable exterior influence decay function for the block updates. We now derive such a decay function from two assumptions. The first states that the conditional Petz amplitude (defined below), which captures a nonlocal effect of $\Phi_B$, is exponentially concentrated near the two ends of the block. The second states that the boundary region of the block Gibbs state becomes approximately independent of a central region as the buffers separating them grow. 
Section~\ref{sec:gibbs-inputs} will verify both properties for the one-dimensional Gibbs states considered in this paper.

Fix a block $B \subsetneq \Lambda_n$. We denote by $\partial B$ the set of sites adjacent on the two edges separating $B$ from $B^c$, and by $N_r(\partial B)$ its radius-$r$ neighborhood in the chain (cf.~Fig.~\ref{fig:static-inputs}(a)).

\subsection{The conditional Petz amplitude}
\label{subsec:petz-amplitude}

Let $\tau_B$ be a full-rank state on $B$. Later, $\tau_B$ will be chosen as the Gibbs state generated by the interactions contained entirely inside $B$.

\begin{definition}[Product conditional expectation and conditional Petz amplitude]
Associated with $\tau_B$, define the product conditional expectation
\begin{align}
\mathcal{E}_B(A) := I_B\otimes \tr_B\left[ (\tau_B^{1/2}\otimes I_{B^c}) A (\tau_B^{1/2}\otimes I_{B^c}) \right],
\label{eq:product-conditional-expectation}
\end{align}
and the conditional Petz amplitude
\begin{align}
Q_B:=\sigma_\beta^{1/2}\left(\tau_B^{-1/2}\otimes\sigma_{\beta,B^c}^{-1/2}\right),\qquad\sigma_{\beta,B^c}:=\tr_B[\sigma_\beta].
\label{eq:conditional-petz-amplitude}
\end{align}
\end{definition}

We note that the map $\mathcal E_B$ is inspired by and is indeed the special case of the minimal conditional expectation $E_B^\rho$ in Ref.~{\cite[Eq.~(13)]{kastoryano2016quantum}} by choosing $\rho=\tau_B\otimes\tau_{B^c}$ with any full-rank state $\tau_{B^c}$.

The map $\mathcal{E}_B$ averages out the degrees of freedom in $B$ with respect to the reference state $\tau_B$, leaving an observable on $B^c$ (embedded as $I_B\otimes\mathcal A_{B^c}$). Indeed, for a product observable $X_B\otimes Y_{B^c}$, one has $\mathcal{E}_B(X_B\otimes Y_{B^c})=\tr(\tau_B X_B) I_B\otimes Y_C$, which makes its interpretation as a conditional expectation over $B$ transparent. In particular, it has the bimodule property $\mathcal{E}_B(K_1AK_2)=K_1\mathcal{E}_B(A)K_2$ for $K_1,K_2\in\mathcal A_{B^c}$. It is also unital and completely positive, which can be checked straightforwardly.

The role of $Q_B$ becomes transparent through the following exact factorization.

\begin{lemma}[Petz-amplitude factorization]
\label{lem:petz-amplitude-factorization}
For every observable $A$,
\begin{align}
\Phi_B^\dagger(A)=\mathcal{E}_B(Q_B^\dagger A Q_B),  \qquad \mathcal{E}_B(Q_B^\dagger Q_B)=I.
\label{eq:petz-amplitude-factorization}
\end{align}
\end{lemma}

\begin{proof}
By definition,
$Q_B(\tau_B^{1/2}\otimes I_{B^c}) = \sigma_\beta^{1/2}(I_B\otimes\sigma_{\beta,B^c}^{-1/2})$.
Substituting this into Eq.~\eqref{eq:product-conditional-expectation} gives
\begin{align}
\mathcal{E}_B(Q_B^\dagger A Q_B) = I_B\otimes\sigma_{\beta,B^c}^{-1/2} \tr_B\left[\sigma_\beta^{1/2}A\sigma_\beta^{1/2}\right] \sigma_{\beta,B^c}^{-1/2},
\end{align}
which is $\Phi_B^\dagger(A)$ from Proposition~\ref{prop:block-petz-adjoint}. The second one follows by setting $A=I$, since $\Phi_B^\dagger$ is unital.
\end{proof}
Equivalently, in the Schr\"odinger picture, 
\begin{align} \Phi_B(\rho) = Q_B\bigl(\tau_B\otimes\tr_B[\rho]\bigr)Q_B^\dagger. \label{eq:petz-amplitude-schrodinger} 
\end{align} 
This expression separates replacement of the entire block by $\tau_B$ from the correction $Q_B$ that accounts for its coupling to the complement. In the Heisenberg picture, the same factorization becomes conjugation by $Q_B$ followed by the product conditional expectation $\mathcal E_B$.

The operator $Q_B$ encodes the deviation of the true Gibbs state from the product reference $\tau_B\otimes\sigma_{\beta,B^c}$ and is therefore the only part of the factorization that can carry \emph{nonlocal} dependence across the boundary of $B$. While $\mathcal{E}_B$ simply averages the $B$ degrees of freedom against $\tau_B$, $Q_B$ involves the global Gibbs state $\sigma_\beta^{1/2}$ and the exact marginal $\sigma_{\beta,B^c}^{-1/2}$, and hence need not be localized near $\partial B$. The quasi-locality assumption below shows that this apparently nonlocal amplitude can nevertheless be approximated exponentially well by operators supported near $\partial B$.

\subsection{The two static assumptions}
\label{subsec:static-inputs}

We now state two static properties that will be assumed throughout this section and verified for 1D Gibbs states in Section~\ref{sec:gibbs-inputs}. Their geometries are illustrated in Fig.~\ref{fig:static-inputs}.

\begin{assumption}[Boundary quasi-locality of the Petz amplitude]
\label{asm:QL}
There exist constants $c_0,C_Q,M_Q,\mu_Q>0$ such that, for every integer $r\geq0$, there is an operator $\widetilde Q_{B,r}\in\mathcal A_{N_{r+c_0}(\partial B)}$ satisfying
\begin{align}
\label{eq:Q-boundary-approximation}
\|Q_B-\widetilde Q_{B,r}\| \leq C_Qe^{-\mu_Qr},
\qquad \|Q_B\|+\|\widetilde Q_{B,r}\| \leq M_Q.
\end{align}
\end{assumption}
Thus, although $Q_B$ is a priori nonlocal, its dependence away from the boundary is exponentially suppressed: it can be approximated by a uniformly bounded operator supported within distance $r+c_0$ of $\partial B$, with error decaying exponentially in $r$. 
This localization allows the distance from the updated block to enter the influence estimate. 

\begin{figure}[t!]
    \centering
    \includegraphics[width=1\linewidth]{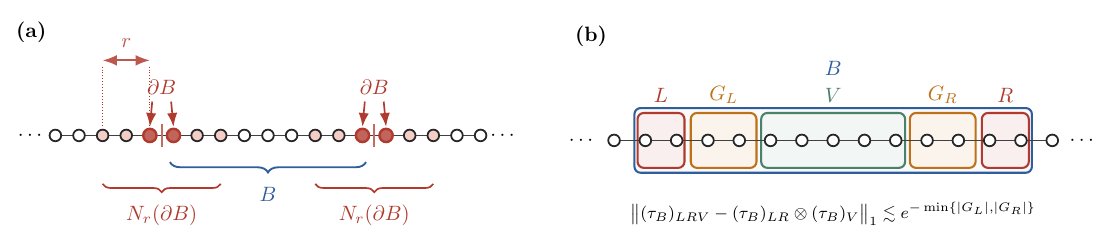}
    \caption{The two static inputs used to control the exterior influence of a block update. (a) Assumption~\ref{asm:QL}: the conditional Petz amplitude $Q_B$ admits exponentially accurate, uniformly bounded approximations supported in a growing neighborhood of $\partial B$ ($c_0 = 0$ in the schematic). (b) Assumption~\ref{asm:GF}: for a decomposition $B=L G_L V G_R R$, the boundary region $L\cup R$ and the central region $V$ approximately factorize, with an error that decays exponentially in $\min\{|G_L|,|G_R|\}$.}
    \label{fig:static-inputs}
\end{figure}

\begin{assumption}[Boundary-to-interior Gibbs factorization]
\label{asm:GF}
There exist constants $C_F,\mu_F>0$ with the following property. For any decomposition of $B$ into five consecutive intervals (cf.~Fig.~\ref{fig:static-inputs}(b)),
\begin{align}
B=L,G_L,V,G_R,R,
\label{eq:boundary-interior-decomposition}
\end{align}
set $m:=\min\{|G_L|,|G_R|\}$.  Then
\begin{align}
\|(\tau_B)_{LRV} - (\tau_B)_{LR}\otimes(\tau_B)_V\|_1 \leq C_Fe^{-\mu_Fm}.
\label{eq:boundary-interior-factorization}
\end{align}
\end{assumption}

Here, $LR:=L\cup R$, and similarly $LRV:=L\cup R\cup V$. This formulation is deliberately restricted to the geometry used below: we do not assume factorization for arbitrary unions of intervals. 

The two assumptions above address different sources of nonlocality. Assumption~\ref{asm:QL} localizes the Petz amplitude near the boundary, while Assumption~\ref{asm:GF} controls the boundary-to-interior correlations that remain after this localization. The main result of this section combines them into the summable decay function required by Theorem~\ref{thm:influence-to-curvature}.

\begin{theorem}[Assumptions~\ref{asm:QL} and \ref{asm:GF} imply exponentially decaying  $\eta$]
\label{thm:static-to-influence}
Suppose that Assumptions~\ref{asm:QL} and \ref{asm:GF} hold. Then there exist constants $C_\eta,\xi>0$ such that, for every block $B$, site $i\notin B$, and observable $A$,
\begin{align}
\delta_i\bigl(\Phi_B^\dagger(A)\bigr) \leq \delta_i(A) + C_\eta e^{-\mathrm{dist}(i,B)/\xi}\|A\|_L.
\label{eq:exponential-petz-influence}
\end{align}
For $i\in B$, one has $\delta_i(\Phi_B^\dagger(A))=0$.
\end{theorem}
Assuming for the moment that the theorem and both assumptions hold, we may take $\eta(r)=C_\eta e^{-r/\xi}$, which is summable. Hence, for any constant block length $\ell>a_\eta$, Corollary~\ref{cor:abstract-trace-mixing} gives positive $W_1$ curvature and thus proves rapid mixing. The rest of this section is devoted to provide a proof of Theorem~\ref{thm:static-to-influence} based on the two assumptions.

\subsection{Proof overview}

To prove Theorem~\ref{thm:static-to-influence}, we approximate the (adjoint) Petz update by a map $\Psi$ and use
\begin{align}
\delta_i\bigl(\Phi_B^\dagger(A)\bigr) \leq \delta_i\bigl(\Psi(A)\bigr) +2\bigl\|(\Phi_B^\dagger-\Psi)(A)\bigr\|.
\label{eq:proof-ov-step1}
\end{align}
The goal is to make the first term no larger than $\delta_i(A)$ and bound the error by an exponentially decaying multiple of $\|A\|_L$. Lemma~\ref{lem:petz-amplitude-factorization} gives $\Phi_B^\dagger(A)=\mathcal{E}_B(Q_B^\dagger A Q_B)$, so the boundary-localized amplitudes $\widetilde Q_{B,r}$ from Assumption~\ref{asm:QL} provide a natural starting point. However, two difficulties must be addressed.

First, substituting $\widetilde Q_{B,r}$ for $Q_B$ need not preserve unitality. This matters because $\|cI\|_L=0$: an approximation error controlled by $\|A\|_L$ must vanish exactly on constant observables. Lemma~\ref{lem:unitalized-approximants} in Section~\ref{subsubsec:unitalization} corrects the normalization without enlarging the support or losing exponential accuracy, producing $\widehat Q_{B,r}$ such that
\begin{align}
\mathcal{E}_B(\widehat Q_{B,r}^\dagger\widehat Q_{B,r})=I, \qquad  \Psi_{B,r}(A):=  \mathcal{E}_B(\widehat Q_{B,r}^\dagger A\widehat Q_{B,r}).
\end{align}
The resulting maps are unital and completely positive, hence contractive. For an exterior site $i$ outside the support of $\widehat Q_{B,r}$, they also preserve observables trivial at $i$. Consequently, $\delta_i(\Psi_{B,r}(A))\leq\delta_i(A)$, giving the desired first term in \eqref{eq:proof-ov-step1}.

Second, a direct approximation estimate still yields $\|\Phi_B^\dagger(A)-\Psi_{B,r}(A)\| \lesssim e^{-\mu r}\|A\|$, whereas $\|A\|$ can grow with the system size even when $\|A\|_L$ stays bounded (e.g., $A = \sum_j Z_j$). We therefore compare consecutive approximations and localize the observable's variation. If the two amplitudes $X,Y$ are supported on a boundary neighborhood $S$, write $A=(A-\mathcal D_S(A))+\mathcal D_S(A)$, where $\mathcal D_S$ depolarizes the sites in $S$. The first part satisfies $\|A-\mathcal D_S(A)\|\leq |S|\|A\|_L$. The second part reduces to $\mathcal{E}_B(C_0\mathcal D_S(A))$, where $C_0=X^\dagger X-Y^\dagger Y$ and exact normalization gives $\mathcal{E}_B(C_0)=0$. This contribution need not vanish, because averaging against $\tau_B$ can still couple the boundary to the interior. As illustrated in Figure~\ref{fig:static-inputs}(b), $C_0$ is localized near the boundary $(L, R)$, while $D_S(A)$ may retain support in the interior ($V$), with the two regions separated by buffers. In Section~\ref{subsec:anchored-increment}, we control this term by successively depolarizing interior layers and applying the boundary-to-interior factorization assumption (Assumption~\ref{asm:GF}). Each layer costs only $\mathcal O(\|A\|_L)$, while its covariance decays exponentially with its depth. Summing these contributions gives the anchored increment estimate below, with a cost proportional to $|S|$ rather than the block length.

Since a 1D boundary neighborhood has $\mathcal O(r)$ sites and consecutive amplitudes differ exponentially little, this yields
\begin{align}
\bigl\|(\Psi_{B,r}-\Psi_{B,r-1})(A)\bigr\| \lesssim \underbrace{(r+1)\|A\|_L}_{\text{observable variation}}  \underbrace{e^{-\mu r}}_{\text{strength}}.
\end{align}
Intuitively, the observable can vary by at most $\mathcal O(r)\|A\|_L$ over the boundary neighborhood, while the effect of enlarging the localized amplitude by one shell is exponentially small in $r$.
Finally, choose a boundary radius $R$ just before the approximating amplitude reaches $i$, and telescope
\begin{align}
    \Phi_B^\dagger(A)=\Psi_{B,R}(A) + \sum_{r>R}(\Psi_{B,r}-\Psi_{B,r-1})(A).
\end{align}
The first term is nonexpansive at $i$, and the remaining sum is an exponential tail in $\operatorname{dist}(i,B)$, up to fixed support offsets. Substitution into \eqref{eq:proof-ov-step1} proves the claimed exterior-influence bound, uniformly in the block length.

\subsection{From boundary localization to Lipschitz control}
\label{subsec:two-ingredients}

\subsubsection{Exact unitalization of boundary approximants}
\label{subsubsec:unitalization}

\begin{lemma}[Unitalized boundary approximants]
\label{lem:unitalized-approximants}
There exist constants $C_U,M_U,\mu_U>0$, independent of $n$ and $B$, and operators $\widehat Q_{B,r}$ for every integer $r\geq0$ such that
\begin{align}
\mathcal{E}_B(\widehat Q_{B,r}^\dagger\widehat Q_{B,r})=I.
\label{eq:exact-unitalization}
\end{align}
Moreover, $\operatorname{supp}\widehat Q_{B,r}\subseteq N_{r+c_0}(\partial B)$, $\|\widehat Q_{B,r}\|\leq M_U$, and $\|\widehat Q_{B,r}-Q_B\|\leq C_Ue^{-\mu_Ur}$. 
\end{lemma}

\begin{proof}
For the boundary approximant $\widetilde Q_{B,r}$ from Assumption~\ref{asm:QL}, define $G_{B,r}\in\mathcal A_{B^c}$ by
\begin{align}
I_B\otimes G_{B,r} := \mathcal{E}_B(\widetilde Q_{B,r}^\dagger \widetilde Q_{B,r}).
\end{align}
Since $\mathcal{E}_B(Q_B^\dagger Q_B)=I$ (Lemma~\ref{lem:petz-amplitude-factorization}) and $\mathcal{E}_B$ is contractive,
\begin{align}
\|G_{B,r}-I_{B^c}\| \leq \|\widetilde Q_{B,r}^\dagger\widetilde Q_{B,r} -Q_B^\dagger Q_B\| \leq \bigl(\|\widetilde Q_{B,r}\|+\|Q_B\|\bigr) \|\widetilde Q_{B,r}-Q_B\| \leq M_QC_Qe^{-\mu_Qr}.
\label{eq:G-close-identity}
\end{align}
Hence there exists an integer $r_0 \geq 1$, independent of $n$ and $B$, such that $\|G_{B,r}-I_{B^c}\|\leq 1/2$ for all $r\geq r_0$. In particular, $G_{B,r}$ is positive and invertible, with spectrum contained in $[1/2,3/2]$.

For $r\geq r_0$, set
\begin{align}
\widehat Q_{B,r} := \widetilde Q_{B,r}(I_B\otimes G_{B,r}^{-1/2}).
\end{align}
Using the previously mentioned $\mathcal A_{B^c}$-bimodule property of $\mathcal{E}_B$,
\begin{align}
\mathcal{E}_B(\widehat Q_{B,r}^\dagger\widehat Q_{B,r}) = (I_B\otimes G_{B,r}^{-1/2}) \mathcal{E}_B(\widetilde Q_{B,r}^\dagger\widetilde Q_{B,r}) (I_B\otimes G_{B,r}^{-1/2}) =I.
\end{align}
Moreover, $G_{B,r}$ is supported on $N_{r+c_0}(\partial B)\cap C$, so the correction by $G_{B,r}^{-1/2}$ does not enlarge the support of $\widetilde Q_{B,r}$. Hence $\operatorname{supp}\widehat Q_{B,r} \subseteq N_{r+c_0}(\partial B)$. 

Since $x\mapsto x^{-1/2}$ is Lipschitz on $[1/2,3/2]$,
Eq.~\eqref{eq:G-close-identity} gives $\|G_{B,r}^{-1/2}-I_{B^c}\| \leq C\|G_{B,r}-I_{B^c}\| \leq C M_QC_Qe^{-\mu_Qr}$, and therefore
\begin{align}
\|\widehat Q_{B,r}-Q_B\| \leq \|\widetilde Q_{B,r}-Q_B\| + \|\widetilde Q_{B,r}\| \|G_{B,r}^{-1/2}-I_{B^c}\| \leq C_Ue^{-\mu_Ur},
\end{align}
after adjusting the constants. The same estimates, together with the uniform bound on $\widetilde Q_{B,r}$, give $\|\widehat Q_{B,r}\|\leq M_U$.

Finally, for the finitely many radii $0\leq r<r_0$, set $\widehat Q_{B,r}=I$. Then $\mathcal{E}_B(\widehat Q_{B,r}^\dagger\widehat Q_{B,r})=I$, and the support condition is automatic. Since $r_0$ is fixed and $\|Q_B\|$ is uniformly bounded, enlarging $C_U$ and $M_U$ if necessary makes the same bounds valid for all $r\geq0$.
\end{proof}

The normalization in Lemma~\ref{lem:unitalized-approximants} allows us to associate to each boundary approximant an exactly unital Heisenberg map. Define
\begin{align}
\Psi_{B,r}(A) := \mathcal{E}_B(\widehat Q_{B,r}^\dagger A\widehat Q_{B,r}), \qquad r\geq0.
\label{eq:unital-shell-channel}
\end{align}
Each $\Psi_{B,r}$ is unital and completely positive. Moreover, $\Psi_{B,r}(A)\to\Phi_B^\dagger(A)$ in operator norm for every fixed $A$, since
\begin{align}
\|\Psi_{B,r}(A)-\Phi_B^\dagger(A)\| \leq (\|\widehat Q_{B,r}\|+\|Q_B\|) \|\widehat Q_{B,r}-Q_B\| \|A\|.
\label{eq:shell-channel-convergence}
\end{align}
The differences between consecutive unitalized amplitudes inherit the same exponential decay.

\begin{corollary}
\label{cor:amplitude-shell-increments}
There exist constants $C_\varepsilon,\mu_\varepsilon>0$, such that for every $r\geq1$,
\begin{align}
\|\widehat Q_{B,r}-\widehat Q_{B,r-1}\| \leq C_\varepsilon e^{-\mu_\varepsilon r}.
\label{eq:amplitude-increment-decay}
\end{align}
In particular, $\sum_{r\geq1}(r+1)^k \|\widehat Q_{B,r}-\widehat Q_{B,r-1}\|<\infty$ for every fixed integer $k\geq0$.
\end{corollary}

\begin{proof}
By the triangle inequality and Lemma~\ref{lem:unitalized-approximants},
\begin{align}
\|\widehat Q_{B,r}-\widehat Q_{B,r-1}\| \leq \|\widehat Q_{B,r}-Q_B\| + \|Q_B-\widehat Q_{B,r-1}\| \leq C_\varepsilon e^{-\mu_\varepsilon r}.
\end{align}
The weighted summability follows immediately.
\end{proof}

Corollary~\ref{cor:amplitude-shell-increments}, together with a direct expansion, gives an exponentially decaying bound on $\Psi_{B,r}(A)-\Psi_{B,r-1}(A)$ in terms of $\|A\|$. To obtain the $\|A\|_L$ bound required for $W_1$, we separate the part of $A$ that is sensitive to the boundary neighborhood from the part that is locally constant there.

\subsubsection{From operator norm to Lipschitz constant}
\label{subsubsec:oscillator-localization}

For a site $j$ and for $S\subseteq\Lambda_n$, let
\begin{align}
\mathcal D_j(A):=\frac{I_j}{d}\otimes\tr_j[A], \qquad \mathcal D_S:=\prod_{j\in S}\mathcal D_j.
\label{eq:site-depolarization}
\end{align}
The following elementary estimates allow us to separate the dependence of an observable on a boundary neighborhood from the part that is locally constant there.

\begin{proposition}[Depolarization localizes oscillations]
\label{prop:depolarization-localization}
For every Hermitian $A$ and every $S\subseteq\Lambda_n$,
\begin{align}
\|A-\mathcal D_S(A)\| \leq \sum_{j\in S}\delta_j(A) \leq |S|\|A\|_L,
\label{eq:depolarization-localization}\\
\delta_i(\mathcal D_S(A)) \leq \delta_i(A) \qquad \text{for every }i\in\Lambda_n.
\label{eq:depolarization-oscillation-contraction}
\end{align}
\end{proposition}

\begin{proof}
We first prove Eq.~\eqref{eq:depolarization-oscillation-contraction}.  Take $K\in\mathcal A_{i^c}$. Then $\mathcal D_S(K)\in\mathcal A_{i^c}$, and hence
\begin{align}
\delta_i(\mathcal D_S(A)) \leq 2\|\mathcal D_S(A-K)\| \leq 2\|A-K\|.
\end{align}
Taking the infimum over $K$ proves the claim.

Next let $K\in\mathcal A_{j^c}$. Then $\mathcal D_j(K)=K$, and contractivity of $\mathcal D_j$ gives $\|A-\mathcal D_j(A)\|\leq2\|A-K\|$. Taking the infimum over $K$ yields $\|A-\mathcal D_j(A)\|\leq\delta_j(A)$. We may now telescope over the sites of $S$; Eq.~\eqref{eq:depolarization-oscillation-contraction} bounds the oscillation at each intermediate step by its initial value. This proves Eq.~\eqref{eq:depolarization-localization}.
\end{proof}

Proposition~\ref{prop:depolarization-localization} will control the portion of an observable that depends directly on the boundary. The remaining portion may still correlate a boundary operator with degrees of freedom deep inside the block. To control that contribution without a dimension-dependent ancilla factor, we use the following complete covariance estimate.

\begin{lemma}[Covariance bound with ancillary systems]
\label{lem:complete-covariance}
Let $U,V,E,F$ be pairwise disjoint regions, and let $\tau_{UV}$ be a state on $U\cup V$, with marginals $\tau_U$ and $\tau_V$. For every $P\in\mathcal A_{U\cup E}$ and $Z\in\mathcal A_{V\cup F}$,
\begin{align}
\left\| \tr_{UV}[\tau_{UV}PZ] -\tr_U[\tau_U P]\tr_V[\tau_V Z] \right\| \leq \|\tau_{UV}-\tau_U \otimes \tau_V\|_1\|P\|\|Z\|.
\label{eq:complete-covariance}
\end{align}
\end{lemma}

\begin{proof}
Set $\Delta:=\tau_{UV}-\tau_U\otimes\tau_V$. The support conditions on $P$ and $Z$ imply that $\tr_{UV}[(\tau_U\otimes\tau_V)PZ] =\tr_U[\tau_U P]\tr_V[\tau_V Z]$. Hence the difference to be bounded is $\tr_{UV}[\Delta PZ]$.

Write the spectral decomposition $\Delta=\sum_a\lambda_a|a\rangle\langle a|$, and define the isometry $V_a:\mathcal H_{EF}\to\mathcal H_{UV}\otimes\mathcal H_{EF}$ by $V_a|\psi\rangle=|a\rangle\otimes|\psi\rangle$. Then $\tr_{UV}[\Delta PZ]=\sum_a\lambda_a V_a^\dagger PZV_a$. Since each $V_a$ is an isometry,
\begin{align}
\left\|\tr_{UV}[\Delta PZ]\right\| \leq \sum_a|\lambda_a| \|V_a^\dagger PZV_a\| \leq \sum_a|\lambda_a| \|PZ\| \leq \|\Delta\|_1\|P\|\|Z\|.
\end{align}
This proves the claim.
\end{proof}


\subsection{Lipschitz-constant control of boundary increments}
\label{subsec:anchored-increment}

Let $S=N_s(\partial B)$ be a boundary neighborhood. The next proposition compares two exactly normalized amplitudes supported on $S$. Exact normalization cancels the constant part of the observable, while Assumption~\ref{asm:GF} controls the remaining correlations between $S\cap B$ and the interior of $B$. The resulting estimate therefore depends on $\|A\|_L$, rather than $\|A\|$.

\begin{lemma}[Anchored increment estimate]
\label{lem:anchored-increment}
Suppose Assumption~\ref{asm:GF} holds. Let $X,Y\in\mathcal A_S$ satisfy
\begin{align}
\mathcal{E}_B(X^\dagger X)=\mathcal{E}_B(Y^\dagger Y)=I, \qquad \|X\|,\|Y\|\leq M.
\label{eq:anchored-normalization}
\end{align}
Then there is a constant $C$ such that every Hermitian $A$ satisfies
\begin{align}
\left\| \mathcal{E}_B(X^\dagger AX-Y^\dagger AY) \right\| \leq C (|S|+1) \|X-Y\| \|A\|_L.
\label{eq:anchored-increment}
\end{align}
\end{lemma}

\begin{proof}
Decompose
\begin{align}
A=A_{\mathrm{near}}+A_0, \qquad A_{\mathrm{near}}:=A-\mathcal D_S(A), \qquad A_0:=\mathcal D_S(A).
\label{eq:near-far-decomposition}
\end{align}
By Proposition~\ref{prop:depolarization-localization}, $\|A_{\mathrm{near}}\|\leq|S|\|A\|_L$.  Expanding the difference and using contractivity of $\mathcal{E}_B$ give
\begin{align}
\left\| \mathcal{E}_B(X^\dagger A_{\mathrm{near}}X -Y^\dagger A_{\mathrm{near}}Y) \right\| &\leq 
\left\| \mathcal{E}_B\left((X-Y)^\dagger A_{\mathrm{near}}X + Y^\dagger A_{\mathrm{near}}(X-Y)\right)\right\| \nonumber \\
& \leq 2M \|X-Y\| |S|\|A\|_L.
\label{eq:anchored-near-bound}
\end{align}

\begin{figure}[t!]
    \centering
    \includegraphics[width=0.95\linewidth]{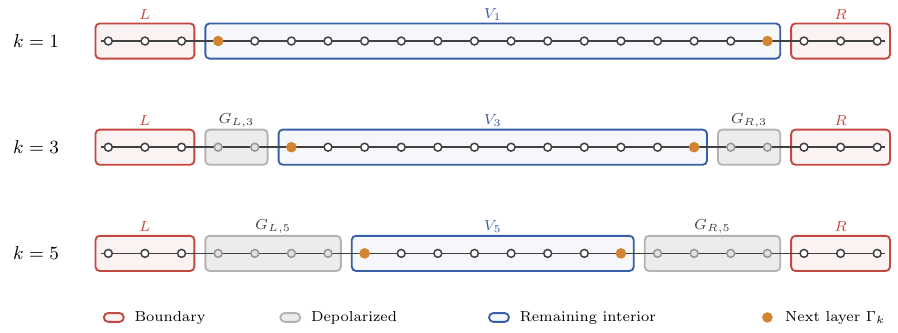}
    \caption{Layer decomposition for Lemma~\ref{lem:anchored-increment}. Depolarized buffers $G_{L,k}$ and $G_{R,k}$ (gray) separate the fixed boundaries $L,R$ (red) from the remaining interval $V_k$ (blue). Orange sites form the next layer $\Gamma_k$. Although $|\Gamma_k|\leq2$, the $B$ support of $\Delta_k A$ may extend throughout $V_k$. Only $B$ is shown.}
    \label{fig:anchored-layer-decomposition}
\end{figure}

It remains to control $A_0$.  Since $A_0$ is trivial on $S$, it commutes with $X$ and $Y$.  Setting $C_0:=X^\dagger X-Y^\dagger Y$, we therefore have
\begin{align}
\mathcal{E}_B(X^\dagger A_0X-Y^\dagger A_0Y)
=\mathcal{E}_B(C_0A_0).
\label{eq:anchored-far-reduction}
\end{align}
Exact normalization in Eq.~\eqref{eq:anchored-normalization} yields $\mathcal{E}_B(C_0)=0$, while $\|C_0\|\leq(\|X\|+\|Y\|)\|X-Y\|\leq2M \|X-Y\|$. This exact cancellation is the point at which unitalization removes the constant observable mode.

Let $S_B:=S\cap B$. If $S_B=B$, then $A_0\in\mathcal A_{B^c}$ and $\mathcal{E}_B(C_0A_0)=\mathcal{E}_B(C_0)A_0=0$, so there is nothing further to prove for the far part. We may therefore assume that $B\setminus S_B$ is nonempty. Write the two components of $S_B$ as $L$ and $R$, ordered from left to right in the interval $B$.  Starting at $L\cup R$ and moving toward the center, partition $B\setminus S_B$ into layers $\Gamma_1,\ldots,\Gamma_m$ of equal distance from $S_B$. Each layer contains at most two sites.  Define $A^{(0)}:=A_0$ and, recursively,
\begin{align}
A^{(k)}:=D_{\Gamma_k}(A^{(k-1)}), \qquad \Delta_k A:=A^{(k-1)}-A^{(k)}.
\label{eq:interior-layer-decomposition}
\end{align}
Then $A_0=A^{(m)}+\sum_{k=1}^m \Delta_k A$.  The final operator $A^{(m)}$ is trivial on all of $B$ and hence belongs to $\mathcal A_{B^c}$. By $\mathcal A_{B^c}$-bimodularity and $\mathcal{E}_B(C_0)=0$, we have $\mathcal{E}_B(C_0A^{(m)})=0$.

Depolarization contracts local oscillations, so Proposition~\ref{prop:depolarization-localization} and $|\Gamma_k|\leq2$ imply
\begin{align}
\|\Delta_k A\|\leq\sum_{j\in\Gamma_k}\delta_j(A^{(k-1)})\leq2\|A\|_L.
\label{eq:layer-increment-bound}
\end{align}
To apply Assumptions~\ref{asm:GF}, we identify the geometry of each layer increment, as illustrated in Fig.~\ref{fig:anchored-layer-decomposition}. Define the remaining central interval
\begin{align}
V_k:=B\setminus\left(S_B\cup\Gamma_1\cup\cdots\cup\Gamma_{k-1}\right),
\label{eq:remaining-interior}
\end{align}
Both $A^{(k-1)}$ and $A^{(k)}$ are trivial on $S_B$ and on all previously depolarized layers. Consequently, the $B$ support of $\Delta_k A$ is contained in $V_k$. Although $\Delta_k A$ is produced by depolarizing the at most two sites in $\Gamma_k$, it need not be supported only on $\Gamma_k$: its support may extend throughout the remaining interval $V_k$.

The previously removed layers form a left buffer $G_{L,k}$ and a right buffer $G_{R,k}$, giving the decomposition
\begin{align}
B=L G_{L,k} V_k G_{R,k} R, \qquad \min\{|G_{L,k}|,|G_{R,k}|\}\geq k-1.
\label{eq:layer-five-interval-decomposition}
\end{align}
This is precisely the five-interval geometry of Fig.~\ref{fig:static-inputs}(b), with the previously depolarized layers forming the buffers $G_{L,k}$ and $G_{R,k}$.

On the complement $C$, the supports of $C_0$ and $\Delta_k A$ are contained in $S\cap C$ and $C\setminus S$, respectively. Both operators are trivial on the buffers $G_{L,k}$ and $G_{R,k}$, so these buffers can be traced out first. We may therefore apply Lemma~\ref{lem:complete-covariance} with $U=L\cup R, V=V_k, E=S\cap C, F=C\setminus S,  P=C_0, Z=\Delta_k A, \tau_{UV}=(\tau_B)_{LRV_k}$.

Assumption~\ref{asm:GF} applied to Eq.~\eqref{eq:layer-five-interval-decomposition} yields 
\begin{align}
\|(\tau_B)_{LRV_k}-(\tau_B)_{LR}\otimes(\tau_B)_{V_k}\|_1\leq C_Fe^{-\mu_F(k-1)}.
\label{eq:layer-factorization}
\end{align}
Combining Eq.~\eqref{eq:anchored-far-reduction} with the preceding estimates, we obtain
\begin{align}
\left\|\mathcal E_B(X^\dagger A_0X-Y^\dagger A_0Y)\right\|
&\leq \sum_k \left\|\mathcal E_B(C_0\Delta_k A)\right\| \nonumber\\
&= \sum_k \left\|
\mathcal E_B(C_0\Delta_k A)
-\mathcal E_B(C_0)\mathcal E_B(\Delta_k A)
\right\|
&&& \bigl(\mathcal E_B(C_0)=0\bigr) \nonumber\\
&\leq \|C_0\|\sum_k
\left\|(\tau_B)_{LRV_k}
-(\tau_B)_{LR}\otimes(\tau_B)_{V_k}\right\|_1
\|\Delta_k A\|
&&& \text{(Lemma~\ref{lem:complete-covariance})} \nonumber\\
&\leq C_F\|C_0\|\sum_k
e^{-\mu_F(k-1)}\|\Delta_k A\|
&&& \text{(Assumption~\ref{asm:GF})} \nonumber\\
&\leq 4MC_F\|X-Y\|\|A\|_L
\sum_k e^{-\mu_F(k-1)} \nonumber\\
&\leq \frac{4MC_F}{1-e^{-\mu_F}}
\|X-Y\|\|A\|_L.
\label{eq:anchored-far-bound}
\end{align}
Combining Eqs.~\eqref{eq:anchored-near-bound} and \eqref{eq:anchored-far-bound} proves Eq.~\eqref{eq:anchored-increment}.
\end{proof}

Lemma~\ref{lem:anchored-increment} converts exact normalization and trace-norm factorization into an oscillator-norm estimate with only a linear cost in the shell size.  Consecutive unitalized Petz amplitudes have exponentially small differences, so this linear cost remains summable.  We make this conclusion explicit before passing to the influence at a fixed exterior site.

\begin{corollary}[Summable shell increments]
\label{cor:summable-channel-increments}
\label{cor}
There exist constants $C,\mu>0$ such that for every observable $A$, 
\begin{align}
\left\| \bigl(\Psi_{B,r}-\Psi_{B,r-1}\bigr)(A) \right\| \leq C(r+1)e^{-\mu r}\|A\|_L
\label{eq:summable-channel-increments}
\end{align}
\end{corollary}

\begin{proof}
Apply Lemma~\ref{lem:anchored-increment} with $X=\widehat Q_{B,r}$ and $Y=\widehat Q_{B,r-1}$.  Both amplitudes are exactly normalized and uniformly bounded by Lemma~\ref{lem:unitalized-approximants}.  Their union of supports is contained in a boundary neighborhood $S_r$ with $|S_r|\leq c(r+1)$, while Corollary~\ref{cor:amplitude-shell-increments} gives $\|X-Y\|\leq C_\varepsilon e^{-\mu_\varepsilon r}$.  Substitution into Eq.~\eqref{eq:anchored-increment} proves the claim after adjusting the constants.
\end{proof}

Corollary~\ref{cor:summable-channel-increments} controls the error made by enlarging the boundary shell.  Before a shell reaches a given exterior site, the corresponding approximate channel cannot increase the oscillation at that site.  This elementary locality observation provides the anchor for the final telescoping sum.

\begin{lemma}[Nonexpansion outside the support]
\label{lem:shell-nonexpansion}
Let $i\notin B$. If $i\notin\operatorname{supp}\widehat Q_{B,r}$, then
\begin{align}
\delta_i(\Psi_{B,r}(A)) \leq \delta_i(A)
\label{eq:shell-nonexpansion}
\end{align}
for every Hermitian $A$.
\end{lemma}

\begin{proof}
If $K\in\mathcal A_{i^c}$, then $\widehat Q_{B,r}^\dagger K\widehat Q_{B,r}$ is still trivial at $i$. Since $i\in C$, applying $\mathcal{E}_B$ does not introduce support at $i$, and hence $\Psi_{B,r}(K)\in\mathcal A_{i^c}$.  The map $\Psi_{B,r}$ is unital and completely positive, so it is an operator-norm contraction~{\cite[Theorem 3.39 (Russo-Dye)]{watrous2018theory}}.  Therefore,
\begin{align}
\delta_i(\Psi_{B,r}(A))
&\leq
2\|\Psi_{B,r}(A-K)\|
\leq
2\|A-K\|.
\end{align}
Taking the infimum over $K\in\mathcal A_{i^c}$ proves
Eq.~\eqref{eq:shell-nonexpansion}.
\end{proof}

\begin{figure}[t!]
    \centering
    \includegraphics[width=0.82\linewidth]{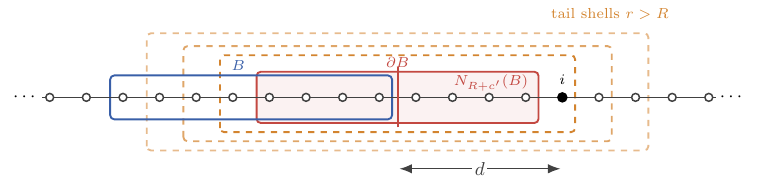}
    \caption{Shell decomposition used in the proof of Theorem~\ref{thm:static-to-influence}. For a site $i\notin B$ at distance $d=\operatorname{dist}(i,B)$, the radius $R$ is chosen so that $\operatorname{supp}\widehat Q_{B,R}$ does not reach $i$. The successive shell enlargements with $r>R$ form the tail in the telescoping expansion of $\Phi_B^\dagger(A)$ in Eq.~\eqref{eq:channel-shell-telescoping}. Only the boundary of $B$ closest to $i$ is shown.}
    \label{fig:shell-tail}
\end{figure}

We now have both ingredients for Theorem~\ref{thm:static-to-influence}: an approximate channel that is nonexpansive until its support reaches $i$, and an exponentially summable bound on every subsequent shell.  The theorem follows by choosing the largest shell that still misses $i$.

\begin{proof}[Proof of Theorem~\ref{thm:static-to-influence}]
The assertion for $i\in B$ is Proposition~\ref{prop:exact-block-killing}.  Fix $i\notin B$ and set $d:=\mathrm{dist}(i,B)$.  By Lemma~\ref{lem:unitalized-approximants} and the inclusion $\partial B\subseteq N_1(B)$, there is a fixed support offset $c'$ such that $\operatorname{supp}\widehat Q_{B,r} \subseteq N_{r+c'}(B)$ for all $r$.  Choose $R:=\max\{0,d-c'-1\}$.  If $R>0$, the support of $\widehat Q_{B,R}$ does not reach $i$.  If $R=0$, we have $\widehat Q_{B,0}=I$, so the same conclusion holds. This choice of $R$ and the subsequent shell enlargements are illustrated in Fig.~\ref{fig:shell-tail}.

By Eq.~\eqref{eq:shell-channel-convergence}, $\Psi_{B,r}(A)\to\Phi_B^\dagger(A)$ in operator norm, so we may telescope from radius $R$:
\begin{align}
\Phi_B^\dagger(A) = \Psi_{B,R}(A) + \sum_{r>R}(\Psi_{B,r}-\Psi_{B,r-1})(A).
\label{eq:channel-shell-telescoping}
\end{align}
Using $\delta_i(Z)\leq2\|Z\|$, Lemma~\ref{lem:shell-nonexpansion}, and Corollary~\ref{cor:summable-channel-increments}, we obtain
\begin{align}
\delta_i(\Phi_B^\dagger(A)) \leq \delta_i(A) + 2C_\Delta \sum_{r>R}(r+1)e^{-\mu_\Delta r}\|A\|_L.
\label{eq:influence-tail}
\end{align}
Define
\begin{align}
\eta(d) := 2C_\Delta \sum_{r>\max\{0,d-c'-1\}} (r+1)e^{-\mu_\Delta r}.
\label{eq:constructed-influence-decay function}
\end{align}
The tail of a polynomially weighted exponential is exponentially decaying. After increasing the prefactor and decreasing the exponent if necessary, therefore, $\eta(d)\leq C_\eta e^{-d/\xi}$ for all $d\geq1$. Equation~\eqref{eq:influence-tail} is then precisely Eq.~\eqref{eq:exponential-petz-influence}.
\end{proof}

Theorem~\ref{thm:static-to-influence} supplies an exponentially decaying, and hence summable, exterior influence decay function. Combined with Theorem~\ref{thm:influence-to-curvature}, it completes the dynamical part of the argument. What remains is therefore to prove Assumptions~\ref{asm:QL} and \ref{asm:GF}, which we do next.

\section{Verification of the two assumptions for 1D Gibbs states}
\label{sec:gibbs-inputs}

Section~\ref{sec:petz-influence} reduced the remaining task to verifying Assumptions~\ref{asm:QL} and \ref{asm:GF}. We now prove both statements for the 1D Hamiltonians defined in Eq.~\eqref{eq:hamiltonian}.

\paragraph{Setup.} Fix a Hamiltonian $H$ as in Eq.~\eqref{eq:hamiltonian}. Throughout this section, $\beta<\infty$ and the local dimension $d$ are fixed. All constants below may depend on $\beta$ and $d$, but are independent of $n$.

Let $B=[a,b]\subsetneq\Lambda_n$ be a block. Define the Hamiltonians obtained by retaining only the interactions contained entirely within $B$ and $B^c$,
\begin{align}
H_B := \sum_{\{i,i+1\}\subset B} h^{(i)}_{i,i+1}, \qquad H_{B^c} := \sum_{\{i,i+1\}\subset B^c} h^{(i)}_{i,i+1}.
\end{align}
Since the interaction is nearest-neighbor, the only terms not contained in $H_B+H_{B^c}$ are the two interactions crossing the boundary of $B$. Hence
\begin{align}
H = H_B+H_{B^c}+V_{\partial B}, \qquad V_{\partial B} := h^{(a-1)}_{a-1,a} + h^{(b)}_{b,b+1},
\end{align}
and $\|V_{\partial B}\|\le 2$. We also denote 
\begin{align}
\tau_B := \frac{e^{-\beta H_B}}{\tr_B[e^{-\beta H_B}]}, \qquad \sigma_{\beta,B^c}:=\tr_B[\sigma_\beta]
\end{align}
for the Gibbs state defined for $H_B$ and the marginal of the full Gibbs state on $B^c$, respectively.


\subsection{Proof of Assumption~\ref{asm:QL}}
\label{sec:verify-QL}

We now prove Assumption~\ref{asm:QL}. Before giving the details, we briefly outline the argument.

\paragraph{Proof outline.}
We prove boundary locality of $Q_B$ by complex interpolation. Consider the analytic family
\begin{align}
F_B(z) := \sigma_\beta^z \bigl(\tau_B^{-z}\otimes\sigma_{\beta,B^c}^{-z}\bigr), \qquad 0\leq\operatorname{Re}z\leq1.
\label{eq:QL-analytic-family}
\end{align}
Since $F_B(1/2)=Q_B$, our goal is to establish spatial locality at the midpoint of the strip. For $S_r:=N_{r+c_0}(\partial B)$, recall from Eq.~\eqref{eq:site-depolarization} that $\mathcal D_{S_r^c}(X)=\operatorname{tr}_{S_r^c}(X)\otimes I_{S_r^c}/d^{|S_r^c|}$. So $\mathcal D_{S_r^c}(Q_B)$ is supported on $S_r$, and it suffices to show that the nonlocal tail $(I-\mathcal D_{S_r^c})Q_B$ decays exponentially in $r$.

The two boundaries of the interpolation strip are the imaginary axis $\operatorname{Re}z=0$ and the vertical line $\operatorname{Re}z=1$. On the imaginary axis, $F_B(it)$ is a unitary evolution satisfying $\frac{d}{dt}F_B(it)=i(e^{-it\beta H}\Delta_Be^{it\beta H})F_B(it)$, where $\Delta_B:=\log\sigma_\beta-\log(\tau_B\otimes\sigma_{\beta,B^c})$ is the difference between the effective Hamiltonians of the two states. Using results from Ref.~\cite{scalet2025classical}, we show that $\Delta_B$ is uniformly bounded and exponentially localized near $\partial B$. Since $H$ is a 1D nearest-neighbor Hamiltonian, a Lieb-Robinson bound controls the propagation of this boundary correction under the evolution generated by $\beta H$. This yields $\|(I-\mathcal D_{S_r^c})F_B(it)\|\lesssim e^{-\nu r+v|t|}$, establishing boundary locality on $\operatorname{Re}z=0$ (Lemma~\ref{lem:left-boundary}).

On the other boundary, $\operatorname{Re}z=1$, results on Araki expansionals in Ref.~\cite{bluhm2022exponential} give a uniform bound on $\|F_B(1+it)\|$ (Lemma~\ref{lem:right-boundary}). The left-boundary estimate still grows exponentially in $|t|$, so we apply Hadamard's three-lines theorem (see Appendix~\ref{app:QL-cited-results}) to the Gaussian-damped nonlocal part
\begin{align}
G_{B,r}(z):=e^{z^2}(I-\mathcal D_{S_r^c})F_B(z).
\end{align}
The Gaussian factor makes the estimates on both boundaries uniform in $t$. Interpolating at $z=1/2$ therefore gives exponential decay of $\|(I-\mathcal D_{S_r^c})Q_B\|$, while applying the same interpolation directly to $F_B$ gives a uniform bound on $\|Q_B\|$. By contractivity of $\mathcal D_{S_r^c}$, the approximants $\mathcal D_{S_r^c}(Q_B)$ are uniformly bounded as well. Hence, taking $\widetilde Q_{B,r}:=\mathcal D_{S_r^c}(Q_B)$ proves Assumption~\ref{asm:QL}.

We first establish the boundary locality of $F_B$ on the left boundary $\operatorname{Re}z=0$.

\begin{lemma}[Boundary locality at $\operatorname{Re}z=0$]
\label{lem:left-boundary}
There exist constants $c_0,C,\nu,v>0$ such that, for every integer $r\geq0$ and every $t\in\mathbb R$,
\begin{align}
\|F_B(it)-\mathcal D_{S_r^c}\left(F_B(it)\right)\| \leq Ce^{-\nu r+v|t|}.
\label{eq:left-boundary-locality}
\end{align}
\end{lemma}

\begin{proof}[Proof sketch]
The effective-Hamiltonian locality result of Ref.~\cite[Proposition~1 and Appendix~A]{scalet2025classical}, together with the relative-operator bounds of Ref.~\cite[Corollaries~3.4 and~4.4]{bluhm2022exponential}, gives $\|\Delta_B-\mathcal D_{S_r^c}\left(\Delta_B\right)\|\leq Ae^{-ar}$ and a uniform bound on $\|\Delta_B\|$. Applying a Lieb-Robinson bound~\cite{nachtergaele2010lieb} to the generator $e^{-it\beta H}\Delta_Be^{it\beta H}$ controls its spatial tail. Replacing this generator by its conditional expectation under $\mathcal D_{S_r^c}$ produces a unitary supported near $\partial B$, which then bounds the difference from $F_B(it)$ and yields Eq.~\eqref{eq:left-boundary-locality}. The complete argument is given in Appendix~\ref{app:left-boundary}.
\end{proof}

We next show the uniform bound on the right boundary $\operatorname{Re}z=1$.

\begin{lemma}[Uniform bound at $\operatorname{Re}z=1$]
\label{lem:right-boundary}
There exists a constant $M\geq1$ such that, for every $t\in\mathbb R$,
\begin{align}
\|F_B(1+it)\|\leq M.
\label{eq:right-boundary-bound}
\end{align}
\end{lemma}

\begin{proof}[Proof sketch]
By unitary invariance of the operator norm, $\|F_B(1+it)\| = \bigl\|\sigma_\beta (\tau_B^{-1}\otimes\sigma_{\beta,B^c}^{-1})\bigr\|$, which can be expressed as the product of an Araki's expansional~\cite{araki1969gibbs} and the inverse of its Gibbs-weighted partial trace. Both factors are uniformly bounded by Ref.~\cite[Corollaries~3.4 and~4.4, especially Eq.~(14)]{bluhm2022exponential}. The details are given in Appendix~\ref{app:right-boundary}.
\end{proof}

We can now combine the two boundary estimates to verify
Assumption~\ref{asm:QL}.

\begin{theorem}[Proof of Assumption~\ref{asm:QL}]
\label{prop:verify-QL}
Assumption~\ref{asm:QL} holds for the Hamiltonians in Eq.~\eqref{eq:hamiltonian}, with constants depending only on $\beta$ and $d$. In particular, one may take $\widetilde Q_{B,r}=\mathcal D_{S_r^c}(Q_B)$, obtained by depolarizing outside $S_r=N_{r+c_0}(\partial B)$.
\end{theorem}

\begin{proof}[Proof sketch]
For $\beta=0$, we have $Q_B=I$. Otherwise, consider
\begin{align}
G_{B,r}(z):=e^{z^2}(I-\mathcal D_{S_r^c})F_B(z).
\end{align}
Lemmas~\ref{lem:left-boundary} and~\ref{lem:right-boundary} imply
\begin{align}
\sup_{t\in\mathbb R}\|G_{B,r}(it)\| \leq Ce^{v^2/4}e^{-\nu r},
\qquad \sup_{t\in\mathbb R}\|G_{B,r}(1+it)\| \leq 2eM.
\end{align}
Hadamard's three-lines theorem (see Appendix~\ref{app:QL-cited-results}), applied to matrix elements,
therefore gives
\begin{align}
\|Q_B-\mathcal D_{S_r^c}\left(Q_B\right)\|\leq C_Qe^{-\nu r/2}.
\label{eq:QL-main-interpolated-locality}
\end{align}
Applying the same theorem directly to $F_B$, using $\|F_B(it)\|=1$ and Lemma~\ref{lem:right-boundary}, gives $\|Q_B\|\leq M^{1/2}$. Thus $\widetilde Q_{B,r}:=\mathcal D_{S_r^c}\left(Q_B\right)$ has the required support and satisfies $\|Q_B\|+\|\widetilde Q_{B,r}\|\leq2M^{1/2}$ by contractivity of $\mathcal D_{S_r^c}$. This proves Assumption~\ref{asm:QL} with $\mu_Q=\nu/2$ and $M_Q=2M^{1/2}$. The full interpolation argument is given in Appendix~\ref{app:QL-interpolation}.
\end{proof}

\subsection{Proof of Assumption~\ref{asm:GF}}
\label{sec:verify-GF}

The proof of Assumption~\ref{asm:GF} is fairly straightforward, as it follows from the two known results for 1D Gibbs states stated below. The first states an exponential decay of correlations on a general finite-range quantum spin chain, which generalizes the popular work by Araki~\cite{araki1969gibbs}.

\begin{proposition}[Exponential clustering of correlations~{\cite[Corollary~I.2]{bergamaschi2026fast}}]
\label{prop:BC-clustering}
Consider the Hamiltonians in Eq.~\eqref{eq:hamiltonian}, and let $\tau_B$ be the Gibbs state on a block $B$. There exist constants $C_{\mathrm{cl}},\mu_{\mathrm{cl}}>0$, depending only on $\beta$ and $d$, such that for every block $B$ and any two intervals $X,Y\subseteq B$,
\begin{align}
\left| \tr[\tau_B A^\dagger D] - \tr[\tau_B A^\dagger]\tr[\tau_B D] \right| \le C_{\mathrm{cl}}\|A\|\|D\| e^{-\mu_{\mathrm{cl}}\mathrm{dist}(X,Y)},
\label{eq:BC-clustering}
\end{align}
for all $A\in\mathcal{A}_X$ and $D\in\mathcal{A}_Y$. 
\end{proposition}

The second result converts exponential uniform clustering into a factorization of Gibbs states. 

\begin{proposition}[Gibbs state factorization~%
{\cite[Proposition~8.1]{bluhm2022exponential}}]
\label{prop:BCPH-factorization}
Fix $\beta\ge0$, and let $\Phi$ be a finite-range interaction on $\mathbb{Z}$ satisfying exponential uniform clustering at inverse temperature $\beta$. Then there exist constants $\widehat C,\widehat\mu>0$, depending only on $\beta$, the range and strength of the interaction, the local dimension, and the uniform-clustering constants, such that the following holds.

For every finite interval $I=XGY\subset\mathbb{Z}$ decomposed into three consecutive intervals with $|G|\ge3s\ge0$, the open-chain Gibbs state $\tau_I$ at inverse temperature $\beta$ satisfies
\begin{align}
\left\|\bigl((\tau_I)_X^{-1}\otimes(\tau_I)_Y^{-1}\bigr)
(\tau_I)_{XY}-I_{XY} \right\| \le \widehat C e^{-\widehat\mu s}.
\label{eq:BCPH-factorization}
\end{align}
Here $(\tau_I)_S$ denotes the marginal of $\tau_I$ on $S$.
\end{proposition}

The preceding result provides exponential uniform clustering for the open-chain Gibbs states considered here. The next result converts this property into an operator-norm factorization estimate.

\begin{corollary}[Two-interval Gibbs factorization]
\label{cor:two-interval-factorization}
There exist constants $C_{\mathrm{fac}},\mu_{\mathrm{fac}}>0$, depending only on $\beta$ and $d$, such that for every interval $I=XGY$ decomposed into three consecutive intervals,
\begin{align}
\|(\tau_I)_{XY}-(\tau_I)_X\otimes(\tau_I)_Y\|_1 \le C_{\mathrm{fac}}e^{-\mu_{\mathrm{fac}}|G|}.
\label{eq:two-interval-factorization}
\end{align}
\end{corollary}

\begin{proof}
By Proposition~\ref{prop:BC-clustering}, every finite restriction of the interaction satisfies Eq.~\eqref{eq:BC-clustering} with the same constants $C_{\mathrm{cl}}$ and $\mu_{\mathrm{cl}}$. In particular, the interaction satisfies exponential uniform clustering in the sense required by Proposition~\ref{prop:BCPH-factorization}. Hence, for every $s\ge 0$ such that $|G|\ge 3s$,
\begin{align}
\left\| \bigl((\tau_I)_X^{-1}\otimes(\tau_I)_Y^{-1}\bigr) (\tau_I)_{XY} -I_{XY} \right\| \le \widehat C e^{-\widehat\mu s}.
\label{eq:factorization-operator-input}
\end{align}
Setting $E_{XY}:=\bigl((\tau_I)_X^{-1}\otimes(\tau_I)_Y^{-1}\bigr)(\tau_I)_{XY}-I_{XY}$ and using H\"older's inequality, 
\begin{align}
\|(\tau_I)_{XY}-(\tau_I)_X\otimes(\tau_I)_Y\|_1 &= \|\bigl((\tau_I)_X\otimes(\tau_I)_Y\bigr)E_{XY}\|_1 \nonumber \\
& \le \|(\tau_I)_X\otimes(\tau_I)_Y\|_1 \|E_{XY}\| \nonumber \\
&= \|E_{XY}\| \nonumber \\
&\leq \widehat C e^{\widehat\mu} e^{-\widehat\mu |G|/3},
\label{eq:operator-to-trace-factorization}
\end{align}
where we set $s=\lfloor |G|/3\rfloor$ in the last line from Eq.~\eqref{eq:factorization-operator-input}. Thus Eq.~\eqref{eq:two-interval-factorization} holds by adjusting constants, e.g., $C_{\mathrm{fac}}:=\widehat C e^{\widehat\mu}$, $\mu_{\mathrm{fac}}:={\widehat\mu}/{3}$.
\end{proof}

We are now ready to verify Assumption~\ref{asm:GF}.

\begin{theorem}[Proof of Assumption~\ref{asm:GF}]
\label{thm:verify-GF}
Assumption~\ref{asm:GF} holds for the Hamiltonians in Eq.~\eqref{eq:hamiltonian}, with constants depending only on $\beta$ and $d$.
\end{theorem}

\begin{proof}
Consider a decomposition of the block into five consecutive intervals,
\begin{align}
B=L G_L V G_R R, \qquad m:=\min\{|G_L|,|G_R|\}.
\end{align}
Applying Corollary~\ref{cor:two-interval-factorization} to $B=L G_L (VG_RR)$ and tracing out $G_R$ gives
\begin{align}
\|(\tau_B)_{LVR}-(\tau_B)_L\otimes(\tau_B)_{VR}\|_1 \le C_{\mathrm{fac}}e^{-\mu_{\mathrm{fac}}|G_L|}.
\label{eq:GF-left}
\end{align}
Here and below we use contractivity of the trace norm under partial traces.

Similarly, applying Corollary~\ref{cor:two-interval-factorization} to $B=(LG_LV) G_R R$ and tracing out $G_L$ gives
\begin{align}
\|(\tau_B)_{LVR}-(\tau_B)_{LV}\otimes(\tau_B)_R\|_1 \le C_{\mathrm{fac}}e^{-\mu_{\mathrm{fac}}|G_R|}.
\label{eq:GF-right}
\end{align}
In particular, tracing out $L$ from above gives
\begin{align}
\|(\tau_B)_{VR} - (\tau_B)_V\otimes(\tau_B)_R\|_1 \le C_{\mathrm{fac}}e^{-\mu_{\mathrm{fac}}|G_R|}.
\label{eq:GF-VR}
\end{align}
Finally, applying Corollary~\ref{cor:two-interval-factorization} to $B=L (G_LVG_R) R$ gives
\begin{align}
\|(\tau_B)_{LR} - (\tau_B)_L\otimes(\tau_B)_R\|_1 \le C_{\mathrm{fac}} e^{-\mu_{\mathrm{fac}}(|G_L|+|V|+|G_R|)} \le C_{\mathrm{fac}}e^{-\mu_{\mathrm{fac}}m}.
\label{eq:GF-LR}
\end{align}
Using Eqs.~\eqref{eq:GF-left}, \eqref{eq:GF-VR}, and \eqref{eq:GF-LR}, together with the triangle inequality and multiplicativity of the trace norm under tensor products with states, we obtain
\begin{align}
\|(\tau_B)_{LRV} - (\tau_B)_{LR}\otimes(\tau_B)_V\|_1 & \le
\|(\tau_B)_{LVR} - (\tau_B)_L\otimes(\tau_B)_{VR}\|_1 \nonumber\\
&\qquad+ \|(\tau_B)_L\otimes(\tau_B)_{VR} - (\tau_B)_L\otimes(\tau_B)_V\otimes(\tau_B)_R\|_1 \nonumber\\
&\qquad+ \|(\tau_B)_L\otimes(\tau_B)_V\otimes(\tau_B)_R - (\tau_B)_{LR}\otimes(\tau_B)_V\|_1
\nonumber\\
&\qquad\le C_{\mathrm{fac}}e^{-\mu_{\mathrm{fac}}|G_L|} + C_{\mathrm{fac}}e^{-\mu_{\mathrm{fac}}|G_R|} + C_{\mathrm{fac}}e^{-\mu_{\mathrm{fac}}m}
\nonumber\\
&\qquad\le 3C_{\mathrm{fac}}e^{-\mu_{\mathrm{fac}}m}.
\end{align}
Thus Assumption~\ref{asm:GF} holds with $C_F:=3C_{\mathrm{fac}},\mu_F:=\mu_{\mathrm{fac}}$.
\end{proof}

\section{Low-depth quantum circuits for Gibbs-state preparation}
\label{sec:dynamics-implementation}

In this section, we provide an algorithm that implements the block heat-bath dynamics directly from the local Hamiltonian terms. Combined with Theorem~\ref{thm:main}, its gate count is nearly linear and its circuit depth is polylogarithmic in $n$.

\subsection{Main result}

We assume classical descriptions of the local interactions of $H$, whose entries can be computed to $p$ bits in time $\operatorname{poly}(b,p)$, where $b$ denotes the classical description length of $H$. 

\begin{theorem}[Implementation of the block heat-bath dynamics]
\label{thm:dynamics-implementation}
Fix $\beta$, $d$, and $\ell$. For every Hamiltonian $H$ in the above setting with $n>\ell$, every $t\geq0$, and every $\varepsilon\in(0,1)$, there exists a randomized quantum algorithm implementing a channel $\widetilde{\mathcal T}_t$ such that
\begin{align}
\left\|\widetilde{\mathcal T}_t-e^{t\mathcal L_\ell}\right\|_\diamond
\leq\varepsilon.
\label{eq:implementation-error}
\end{align}
Using one- and two-qubit gates on a 1D nearest-neighbor geometry, the algorithm has gate count $G$ and circuit depth $D$ satisfying
\begin{align}
G &= \mathcal O\left( \bigl(nt+\log(2/\varepsilon)\bigr) \operatorname{polylog}\left( \frac{2n(1+t)}{\varepsilon} \right) \right),
\nonumber\\
D &= \mathcal O\left( \bigl(t+\log(2n/\varepsilon)\bigr) \operatorname{polylog}\left( \frac{2n(1+t)}{\varepsilon} \right) \right).
\end{align}
It uses $\mathcal O\left( n\operatorname{polylog}\left( \frac{2n(1+t)}{\varepsilon} \right) \right)$ ancilla qubits, which may be reset and reused, and requires classical processing time $\operatorname{poly}(n,t,\log(1/\varepsilon),b)$. 
\end{theorem}

Combining this implementation with Theorem~\ref{thm:main} gives the following state-preparation result.

\begin{corollary}[Gibbs-state preparation with a polylog-depth circuit]
\label{cor:gibbs-preparation}
Fix $\beta$ and $d$, and choose the constant block length $\ell=\ell_\star$ from Theorem~\ref{thm:main}. For every $\varepsilon\in(0,1)$ and any initial state $\rho$, there exists a quantum circuit whose output state $\widetilde\sigma_\beta$ satisfies $\frac12\|\widetilde\sigma_\beta - \sigma_\beta\|_1\leq\varepsilon$. The circuit uses
\begin{align}
    G = \mathcal{O}\bigl(n\operatorname{polylog}(n/\varepsilon)\bigr),
    \qquad
    D = \mathcal{O}\bigl(\operatorname{polylog}(n/\varepsilon)\bigr)
\end{align}
one- and two-qubit gates, and $\mathcal{O}\bigl(n\operatorname{polylog}(n/\varepsilon)\bigr)$ ancilla qubits.
\end{corollary}

\begin{proof}
Let $\rho$ be any initial state and take $t=t_{\mathrm{mix}}(\varepsilon/2)=\mathcal O(\log(2n/\varepsilon))$. Then Theorem~\ref{thm:main} gives $\frac12\bigl\|e^{t\mathcal L_{\ell_\star}}(\rho)-\sigma_\beta\bigr\|_1 \leq \frac{\varepsilon}{2}$.  Implement $e^{t\mathcal L_{\ell_\star}}$ to diamond-norm error at most $\varepsilon$ using Theorem~\ref{thm:dynamics-implementation}. If $\widetilde\sigma_\beta :=\widetilde{\mathcal T}_t(\rho)$ denotes the output state, then $\frac12\bigl\|\widetilde\sigma_\beta -e^{t\mathcal L_{\ell_\star}}(\rho)\bigr\|_1 \leq \frac{\varepsilon}{2}$. The triangle inequality therefore gives $\frac12\|\widetilde\sigma_\beta-\sigma_\beta\|_1\leq\varepsilon$. Substituting $t=\mathcal O(\log(n/\varepsilon))$ into Theorem~\ref{thm:dynamics-implementation} gives the stated resource bounds.
\end{proof}

\subsection{Algorithm overview}
\label{subsec:algorithm-overview}

In this section, we briefly outline the algorithm and its intuition. All of the details are provided in Appendices~\ref{app:finite-window} and \ref{app:implementation}. We first note again that $e^{t\mathcal L_\ell}$ can be implemented by applying a Poisson-distributed number of randomly chosen block updates.

\begin{proposition}[Poisson-clock representation of $e^{t \mathcal L_{\ell}}$]
\label{prop:poisson-representation}
Let $s=nt/\ell$ and $K\sim\operatorname{Poisson}(s)$. Conditioned on $K$, let $B_1,\ldots,B_K$ be chosen i.i.d. from the $n$ blocks of length $\ell$. Then
\begin{align}
e^{t\mathcal L_\ell} = \mathbb E \left[\Phi_{B_K}\circ\cdots\circ\Phi_{B_1} \right].
\label{eq:poisson-representation}
\end{align}
\end{proposition}

\begin{proof}
A direct expansion gives
\begin{align}
e^{t\mathcal L_\ell} & = e^{-s}\sum_{k=0}^{\infty}\frac{s^k}{k!}\left(\frac{1}{n}\sum_B\Phi_B\right)^k = \sum_{k=0}^{\infty}\frac{e^{-s}s^k}{k!} \frac{1}{n^k}\sum_{B_1,\ldots,B_k} \Phi_{B_k}\circ\cdots\circ\Phi_{B_1} \nonumber \\
&= \sum_{k=0}^{\infty}\mathbb P(K=k)  \mathbb E \left[ \Phi_{B_K}\circ\cdots\circ\Phi_{B_1} \middle|  K=k \right] = \mathbb E \left[\Phi_{B_K}\circ\cdots\circ\Phi_{B_1}\right].
\end{align}
\end{proof}

It therefore suffices to implement each $\Phi_B$ approximately, with accuracy chosen according to the expected number of updates,
\begin{align}
    \mathbb E[K]=s=\frac{nt}{\ell}.
\end{align}

The central piece of our algorithm is to show that we can replace $\Phi_B$ with a Petz update constructed from a finite Hamiltonian window. The following lemma is proved in Appendix~\ref{app:finite-window}.

\begin{lemma}[Finite-window approximation of $\Phi_B$]
\label{lem:finite-window-approximation}
Fix $\beta$, $d$, and $\ell$. There exist constants $C,\mu>0$, independent of the system size, such that the following holds. For a block $B$ and an integer $r\geq0$, let $W=N_r(B)$ and let $H_W$ be the restriction of $H$ to $W$. Then
\begin{align}
\left\| \Phi_B - \left(\Phi_B^{H_W}\otimes I _{W^c}\right) \right\|_\diamond \leq C e^{-\mu r},
\label{eq:finite-window-approximation}
\end{align}
where $\Phi_B^{H_W}$ is the Petz block update associated with the Gibbs state of $H_W$.
\end{lemma}

Hence, instead of implementing $\Phi_B$ directly, we implement its finite-window approximation, as illustrated in Fig.~\ref{fig:main-petz-circuit}(a). We allocate error $\delta=\Theta(\varepsilon/(1+s))$ to each of the window approximation and local compilation, where $s=nt/\ell$ is the expected number of updates. Choosing $r=\mathcal{O}(\log(1/\delta))$ makes the window error at most $\delta$. At the mixing time $t=\mathcal{O}(\log(n/\varepsilon))$, the window size therefore satisfies $w=|W|\leq\ell+2r=\mathcal{O}(\log(n/\varepsilon))$. For fixed target accuracy, this gives the $\mathcal{O}(\log n)$ window shown in Fig.~\ref{fig:main-petz-circuit}(a).

By Proposition~\ref{prop:poisson-representation}, the expected number of updates is $s=\mathcal{O}(n\log n)$ at fixed target accuracy. Since each update acts only on a window of $\mathcal{O}(\log n)$ sites, about $n/\log n$ disjoint updates can run in parallel across the chain. We show that the sampled updates can be organized into only $\mathcal{O}(\operatorname{polylog}(n))$ layers with high probability, while preserving the sampled order of overlapping updates (Appendix~\ref{app:implementation}). It therefore remains to implement each finite-window update $\Phi_B^{H_W}$ in polylog depth.

\begin{figure}
    \centering
    \includegraphics[width=1\linewidth]{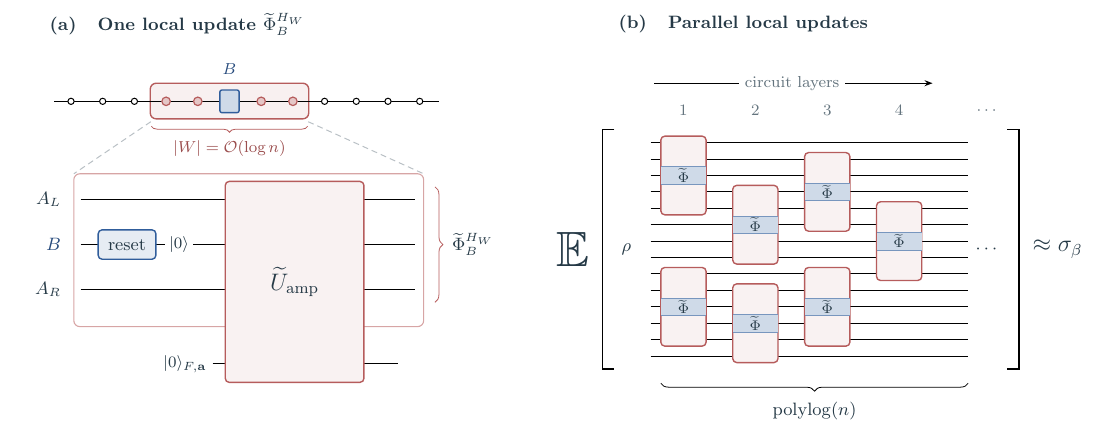}
   \caption{Local Petz updates and Gibbs-state preparation. (a) Reset the block $B$ (blue), apply the circuit implementing $\widetilde V_W$ using data registers in $W$ (red) and auxiliary registers $F,\mathbf a$, then discard the auxiliary registers. Here $A_L\cup A_R=W\setminus B$. (b) Disjoint windows are updated in parallel, while overlapping updates retain their sampled order. The expectation averages the output states over the classical update sequence. At mixing time and fixed target accuracy, $|W|=\mathcal O(\log n)$ and the circuit depth, including nearest-neighbor compilation, is $\operatorname{polylog}(n)$. Auxiliary registers are omitted in (b).}
    \label{fig:main-petz-circuit}
\end{figure}


To do so, let $A=W\setminus B$ and write $H_W=H_0+V_{\partial B}$, where $H_0=H_B+H_A$ and $V_{\partial B}$ contains the interactions connecting $B$ to $A$. The finite-window Petz update takes the form
\begin{align}
\Phi_B^{H_W}(\rho_W) = Q_B^{(W)} \bigl(I_B\otimes\tr_B[\rho_W]\bigr) {Q_B^{(W)}}^\dagger, \qquad
Q_B^{(W)} := \sigma_W^{1/2} \bigl(I_B\otimes\sigma_A^{-1/2}\bigr),
\label{eq:Phi_B-HW}
\end{align}
where $\sigma_W=e^{-\beta H_W}/\tr(e^{-\beta H_W})$ and $\sigma_A=\tr_B[\sigma_W]$.

Let us begin from observing what happens for the decoupled case $V_{\partial B}=0$. In this case, we simply have $\Phi_B^{H_W}(\rho_W) = \tau_B \otimes \tr_B(\rho_W)$, i.e., the update simply discards $B$ and replaces it with $\tau_B=e^{-\beta H_B}/\tr(e^{-\beta H_B})$. Since $|B|=\mathcal{O}(1)$, this requires only a constant-depth circuit. When $V_{\partial B} \neq 0$, the recovery must also reconstruct correlations between $B$ and $A$. In 1D, these couplings consist of at most two bounded local terms, suggesting that their effect can be incorporated efficiently. However, the technical challenge is that naively implementing the factors in Eq.~\eqref{eq:Phi_B-HW} separately does not exploit this boundary-local structure, as the smallest eigenvalue of $\sigma_A$ may be exponentially small in $|W|$, blowing up a cost to $e^{\mathcal{O}(|W|)} \sim \operatorname{poly}(n)$.

To use our intuition and isolate the boundary correction, define $E=e^{-\beta H_W/2}e^{\beta H_0/2}$ and $C=e^{-\beta H_A/2}$. The Petz amplitude admits the exact factorization
\begin{align}
    Q_B^{(W)} =E(\tau_B^{1/2}\otimes T),
    \qquad
    T=C(CSC)^{-1/2},
    \qquad
    S=\tr_B \left[(\tau_B^{1/2}\otimes I_A) E^\dagger E (\tau_B^{1/2}\otimes I_A) \right].
    \label{eq:petz-amplitude-factorization-algo}
\end{align}
The operator $E$ incorporates the boundary interaction, and the factor $T$ fills the remaining correction on $A$ required by the Gibbs marginal. Its normalization condition $T^\dagger S T=I_A$ makes the resulting recovery operator an isometry, and hence ensures that the induced recovery channel is trace preserving. When $V_{\partial B}=0$, we have $E=S=T=I$, recovering the simple replacement of $B$ by $\tau_B$.

Combining $E$ and $T$ with a purification of $\tau_B$, denoted by $|\sqrt{\tau_B}\rangle_{BF}$, defines the recovery isometry
\begin{align}
\label{eq:recovery-isometry}
V_W|\psi\rangle_A := (E_{BA}\otimes I_F)(I_B\otimes T_A\otimes I_F) \bigl(|\sqrt{\tau_B}\rangle_{BF}\otimes|\psi\rangle_A\bigr), \qquad V_W^\dagger V_W=I_A.
\end{align}
Tracing out $F$ gives the local Petz update,
\begin{align}
\Phi_B^{H_W}(\rho_W) = \tr_F\left[V_W\bigl(\tr_B\rho_W\bigr)V_W^\dagger
\right].
\end{align}
We block-encode $E$ and $T$, combine them into a projected encoding of $V_W$, and amplify it to obtain a unitary $\widetilde U_{\rm amp}$ whose action on fresh ancillas approximates $V_W$. Resetting $B$, applying $\widetilde U_{\rm amp}$, and discarding the auxiliary registers therefore implements the local update without postselection, as illustrated in Fig.~\ref{fig:main-petz-circuit}(a).

\paragraph{Block-encoding of $\boldsymbol{E}$.}

Although $E$ need not be close to the identity, its nontrivial action originates from the boundary perturbation $V_{\partial B}$. Lemma~\ref{lem:efficient-boundary-expansional} makes this observation constructive: for fixed $\beta,d,\ell$, it provides a block encoding of $E$ with constant normalization and gate cost polynomial in $w=|W|$ and $\log(1/\eta)$, where $\eta$ is the block-encoding error.

The construction has two steps. First, we block-encode the imaginary-time-evolved boundary interaction
\begin{align}
    V_0(s) := e^{-sH_0}V_{\partial B}e^{sH_0} = \sum_{k=0}^{\infty}  \frac{(-s)^k}{k!}  \operatorname{ad}_{H_0}^{k}(V_{\partial B}),  \qquad 0\le s\le\beta/2,
\end{align}
where $\operatorname{ad}_{H_0}(X)=[H_0,X]$. Only sequences of terms in $H_0$ that successively overlap the accumulated support can contribute. In 1D, these sequences can be organized as histories growing from the boundary. We bound their total absolute coefficient weight uniformly in $w$ and coherently prepare an equal superposition over the valid histories without enumerating them explicitly. Each selected history reduces to a signed product of local Pauli operators and can therefore be applied efficiently. Truncating at order $\mathcal{O}(\log(1/\eta))$ yields an efficient block encoding of $V_0(s)$, controlled by a register specifying $s$. This process is specified in Proposition~\ref{prop:boundary-coherent-encoding}.

Second, we combine these encodings through a truncated Dyson expansion. Writing $E(s)=e^{-sH_W}e^{sH_0}$, we have $\frac{d}{ds}E(s)=-E(s)V_0(s)$ and $E(0)=I$, so that
\begin{align}
   E = I+\sum_{m=1}^{\infty}(-1)^m \int_{0\le s_1\le\cdots\le s_m\le\beta/2} V_0(s_1)\cdots V_0(s_m) ds_1\cdots ds_m.
\end{align}
Heavily inspired by the truncated Dyson series constructions of Refs.~\cite{kieferov2019simulating} (see also Ref.~\cite{low2019hamiltonian}), we prepare the expansion order and discretized time labels, reversibly sort the times, and apply the corresponding block encodings in the required order. Unpreparing the label registers produces a block encoding of the truncated sum. 

The key point is that $V_{\partial B}$ has constant-size support, so at truncation order $K$ the contributing commutator histories can be generated and implemented using $\operatorname{poly}(K,w)$ gates. Their total coefficient weight is bounded independently of $w$, which implies that $K=\mathcal O(\log(1/\eta))$ suffices for accuracy $\eta$. The Dyson expansion of $E$ likewise requires only $\mathcal O(\log(1/\eta))$ orders, with each order using only polynomially many calls to these $V_0(s)$ encodings. Hence the overall cost is $\operatorname{poly}(w,\log(1/\eta))$. The detailed resource analysis and the algorithm are specified in the proof of Lemma~\ref{lem:efficient-boundary-expansional}. The details are provided in Appendix~\ref{app:implementation}.

\paragraph{Block-encoding of $\boldsymbol{T}$.}

The remaining task is to block-encode $T=C(CSC)^{-1/2}$. The bounds on $E$ and $E^{-1}$ from above construction imply that $c_-I_A\leq S\leq c_+I_A$ for positive constants $c_\pm$ independent of $w$. Also, a block encoding of $S$ follows by preparing and unpreparing $\ket{\sqrt{\tau_B}}_{BF}$ around an encoding of $E^\dagger E$. Although $S$ is well conditioned, $CSC$ may have exponentially small eigenvalues. We therefore implement the product $C(CSC)^{-1/2}$ directly, preserving the cancellation between its factors.

Define $D_u=(I_A+e^{\beta H_A+2uI_A})^{-1/2}$, $F_u=(I_A+e^{-\beta H_A-2uI_A})^{-1/2}$, and $P_u=D_uSD_u+F_u^2$. Since $D_u^2+F_u^2=I_A$, the bounds on $S$ give $\min{c_-,1}I_A\leq P_u\leq\max{c_+,1}I_A$. Thus, $P_u$ is uniformly well conditioned for every $u\in\mathbb R$. The inverse-square-root integral yields the exact representation
\begin{align}
T =\frac{2}{\pi}\int_{\mathbb R} e^u C(CSC+e^{2u}I_A)^{-1} du = \frac{2}{\pi}\int_{\mathbb R} D_uP_u^{-1}F_u du.
\label{eq:T-preconditioned-integral}
\end{align}
The operators $D_u$ and $F_u$ are bounded functions of $H_A$ that admit efficient polynomial approximations. We block-encode them directly by polynomial eigenvalue transformation of $H_A$ via Ref.~\cite{gilyen2019quantum}, rather than implementing their inverse-square-root factors separately. Combining these encodings with that of $S$ gives an encoding of $P_u$, whose inverse can be implemented efficiently using its uniform spectral bounds.

Finally, we truncate the integral in Eq.~\eqref{eq:T-preconditioned-integral} and approximate it using Gauss-Legendre quadrature. The integrand has exponentially decaying tails beyond $|u|=\mathcal O(\beta\|H_A\|)$. Since $\|H_A\|=\mathcal O(w)$, an integration range of size $\mathcal O(w+\log(1/\eta))$ and polynomially many quadrature nodes suffice for error $\eta$. The quadrature weights are positive and sum to the length of the integration interval, so combining the node encodings by linear combination of unitaries (LCU) gives a block encoding of $T$ with normalization $\mathcal O(w+\log(1/\eta))$. Since each node encoding and the number of nodes have polynomial cost in $w$ and $\log(1/\eta)$, the overall gate cost is also polynomial in these parameters (see Lemma~\ref{lem:relative-inverse}).

\paragraph{From block encodings to the local update.}

Finally, combining the block encodings of $E$ and $T$ with the preparation of $|\sqrt{\tau_B}\rangle_{BF}$ gives a projected encoding of the recovery isometry $V_W$ in Eq.~\eqref{eq:recovery-isometry}. Oblivious isometry amplification (Lemma~\ref{lem:be-isometry-amplification}) then gives a unitary whose action on the initialized input subspace approximates $V_W$, without postselection. Lemma~\ref{lem:local-compilation} shows that the resulting finite-window update $\Phi_B^{H_W}$ can be implemented using $\operatorname{poly}(w,\log(1/\delta))$ gates and ancillas. Routing within each window on a 1D nearest-neighbor architecture incurs only polynomial overhead in $w$. Since $w=\mathcal O(\log(n/\varepsilon))$ at the mixing time and disjoint updates are executed in parallel, Theorem~\ref{thm:dynamics-implementation} yields the claimed $\operatorname{polylog}(n/\varepsilon)$ circuit depth.

\section{Discussion and outlook}

In this work, we introduce a quantum block heat-bath dynamics based on Petz recovery and prove rapid mixing for the 1D spin chains considered here at every fixed finite temperature, without any commutativity assumption. Our proof establishes contraction directly in quantum $W_1$ distance using two static locality properties of 1D Gibbs states. We further provide low-depth quantum circuits for preparing the corresponding Gibbs states.

An immediate direction is to simplify our algorithm and determine the explicit dependence of its resource costs and mixing bounds on $\beta$ and $d$, which we have treated as fixed parameters throughout. In particular, it seems our current implementation incurs a gate overhead that is at least triply exponential in $\beta$. Sharper bounds on $\ell_\star$, the mixing rate, and the quantum and classical implementation costs would clarify the efficiency of our approach at low temperatures.

Another interesting direction is to extend our direct $W_1$ contraction approach to other Gibbs samplers from Refs.~\cite{chen2023an,ding2025efficient}. Their quasi-local structures suggest that similar influence estimates may be useful, but it for now seems challenging how to replace or connect the exact erasure of single-site discrepancies used in this work to there. It would also be useful to clarify how our work is precisely related to the quantum Dobrushin condition of Ref.~\cite{bakshi2026a}.

It would also be interesting to extend our results to other classes of Hamiltonians. For 1D Hamiltonians with sufficiently rapidly decaying power-law interactions, quasipolynomial-time algorithms for MPO construction and quantum Gibbs-state preparation are known~\cite{achutha2025provably}, and polynomial-bond-dimension approximations of purified Gibbs states were recently shown to exist, although efficient constructions remain open~\cite{kim2026spectral}; this raises the question of whether our dynamics and algorithm can be adapted to yield polynomial-depth state-preparation circuits. Another important direction is higher-dimensional systems, where finite-temperature phase transitions make a general all-temperature rapid-mixing result unlikely~\cite{martinelli1994on,ding2009the,gheissari2018Potts}. The natural question is then which additional conditions on the Hamiltonian or its Gibbs state suffice for fast or rapid mixing and efficient state preparation.

Finally, our contraction bound also gives convergence estimates that depend on the initial state $\rho_0$. Since $\frac12\|e^{t\mathcal L_\ell}(\rho_0)-\sigma_\beta\|_1 \leq e^{-\gamma_\ell t}\|\rho_0-\sigma_\beta\|_{W_1}$, a smaller initial $W_1$ distance can therefore reduce the required evolution time. For example, a distance of $\operatorname{polylog}(n)$ gives an $\mathcal O(\log\log n)$ mixing time, while a bounded distance gives a bound independent of $n$. Identifying physically relevant states or ensembles with these properties (in line with recent interest in initial-state-dependent mixing time ~\cite{bergamaschi2025logical,gheissari2023low,cheng2026typical,chen2026single}) is a valuable question for future work.

\paragraph{Note added.} While preparing this manuscript, Ref.~\cite{pilatowsky2026one} appeared, establishing constant-depth quantum circuits for preparing 1D Gibbs states at any fixed finite temperature and thereby improving upon the polylogarithmic-depth constructions of this work.

\section*{Acknowledgements}
We acknowledge support by the U.S. Department of Energy, Office of Science, Office of Advanced Scientific Computing Research under Contract No. DE-AC05-00OR22725 through the Accelerated Research in Quantum Computing Program MACH-Q project.

\section*{AI Disclosure}
We acknowledge the assistance of ChatGPT (5.5, 5.6 Sol and 6 Astra) in developing several proof strategies. In particular, it helped us  resolve the issue via the unitalization argument in Section~\ref{subsubsec:unitalization}, prove Assumption~\ref{asm:QL} in Section~\ref{sec:verify-QL}, and improve the algorithm design and analysis in Appendix~\ref{app:implementation}. It also flagged several issues in the original arguments and helped us in revising them. The author wrote the entire manuscript manually, and takes full responsibility for the content.

\bibliographystyle{unsrturl}
\bibliography{ref}

\clearpage

\appendix

\section{Proofs of the boundary estimates and Assumption~\ref{asm:QL}}
\label{app:QL-details}

In this section, we give the complete proofs of Lemmas~\ref{lem:left-boundary} and~\ref{lem:right-boundary}, and Theorem~\ref{prop:verify-QL}.

\subsection{Backgrounds}
\label{app:QL-cited-results}

We first list several existing results that we will use. The first two propositions are about an open chain split into two adjacent intervals $B$ and $B^c$, with uniformly bounded on-site and nearest-neighbor interactions. (We explain a straightforward extension to the periodic chain in Appendix~\ref{app:QL-boundary-estimates}.) Throughout, we fix $\beta,d$.

\begin{proposition}[Boundary locality of the Hamiltonian of mean force
{\cite[Proposition~1 and Appendix~A]{scalet2025classical}}]
\label{prop:QL-cited-HMF}
Let $H$ be an open-chain Hamiltonian satisfying the standing interaction assumptions, and let $B$ and $B^c$ be two nonempty adjacent intervals partitioning the chain. Denote by $\sigma_{\beta,B^c}$ the marginal of the full Gibbs state on $B^c$, and by $\tau_{B^c}$ the Gibbs state of the restricted Hamiltonian $H_{B^c}$.

There exist constants $c_{\mathrm{mfo}},A_{\mathrm{mfo}},a_{\mathrm{mfo}}>0$, depending only on $\beta$ and $d$, such that, for every integer $r\geq0$, there is a Hermitian operator $W_{B^c,r}\in \mathcal A_{B^c\cap N_{r+c_{\mathrm{mfo}}}(\partial B)}$ such that
\begin{equation}
\bigl\| -\log\sigma_{\beta,B^c} +\log\tau_{B^c} -W_{B^c,r}
\bigr\| \leq A_{\mathrm{mfo}}e^{-a_{\mathrm{mfo}}r}.
\label{eq:QL-cited-HMF}
\end{equation}
\end{proposition}

\begin{proposition}[Bounds for Araki expansionals and their partial traces~{\cite[Corollaries~3.4, 4.4]{bluhm2022exponential}}]
\label{prop:QL-cited-expansional}
Define
\begin{align}
\widehat E_B :=e^{-\beta H}e^{\beta(H_B+H_{B^c})}, \qquad  T_{B^c} :=\tr_B\bigl[(\tau_B\otimes I_{B^c})\widehat E_B\bigr].
\label{eq:QL-full-expansional}
\end{align}
There exists a constant $G\geq1$ such that $T_{B^c}$ is invertible and
\begin{align}
\max\bigl\{\|\widehat E_B\|,\|\widehat E_B^{-1}\|, \|T_{B^c}\|,\|T_{B^c}^{-1}\|\bigr\}\leq G.
\label{eq:QL-cited-expansional}
\end{align}
\end{proposition}

The next estimate applies directly to the physical Hamiltonian $H$, on either an open chain or a periodic chain.

\begin{proposition}[Lieb-Robinson bound {\cite[Theorem~2.3]{nachtergaele2010lieb}}]
\label{prop:QL-cited-LR}
There exist constants $C_{\mathrm{LR}},b,v_0>0$ such that, for disjoint nonempty sets $X,Y$ and operators $\mathcal O_X\in\mathcal A_X$, $\mathcal O_Y\in\mathcal A_Y$,
\begin{align}
\bigl\|[e^{-it\beta H}O_Xe^{it\beta H},O_Y]\bigr\|
\leq C_{\mathrm{LR}}|X| \|O_X\| \|O_Y\|
e^{-b\operatorname{dist}(X,Y)+v_0|t|}
\label{eq:QL-cited-LR}
\end{align}
for all $t\in\mathbb R$. These constants are uniform in the chain length, in $X,Y$, and in the interaction profile.
\end{proposition}

\begin{proposition}[Hadamard's three-lines theorem
{\cite{rudin1987real}}]
\label{prop:QL-cited-three-lines}
Let $f$ be bounded and continuous on the closed strip $0\leq\operatorname{Re}z\leq1$, and analytic in its interior.
If $m_0,m_1>0$ satisfy
\begin{align}
\sup_{t\in\mathbb R}|f(it)|\leq m_0, \qquad \sup_{t\in\mathbb R}|f(1+it)|\leq m_1,
\end{align}
then, for $0<\theta<1$,
\begin{align}
\sup_{t\in\mathbb R}|f(\theta+it)| \leq m_0^{1-\theta}m_1^\theta.
\label{eq:QL-cited-three-lines}
\end{align}
\end{proposition}

\subsection{Auxiliary boundary estimates}
\label{app:QL-boundary-estimates}

We now return to the periodic chain of the main text. 
For $\beta=0$, we have $F_B(z)=I$ and $\Delta_B=0$, so all the assertions are immediate. We assume $\beta>0$ below.

\paragraph{Application of the open-chain results.} For a fixed block $B$, fold the ring as illustrated in Fig.~\ref{fig:ring-folding}: pair sites from the endpoints of $B$ and $B^c$ toward their respective centers, and order the resulting supersites from the center of $B$ to that of $B^c$. This gives an open chain with local dimension at most $d^2$ and uniformly bounded on-site and nearest-neighbor interactions. Regrouping preserves operator norms, Gibbs states, and the partial trace over $B$, while folded boundary neighborhoods unfold to neighborhoods of $\partial B$ with a fixed support offset. Thus Propositions~\ref{prop:QL-cited-HMF} and~\ref{prop:QL-cited-expansional} apply after adjusting the constants. Unpaired center sites may be padded with inert spins, whose removal by normalized partial trace preserves the operator-norm bounds.

\begin{figure}
    \centering
    \includegraphics[width=0.9\linewidth]{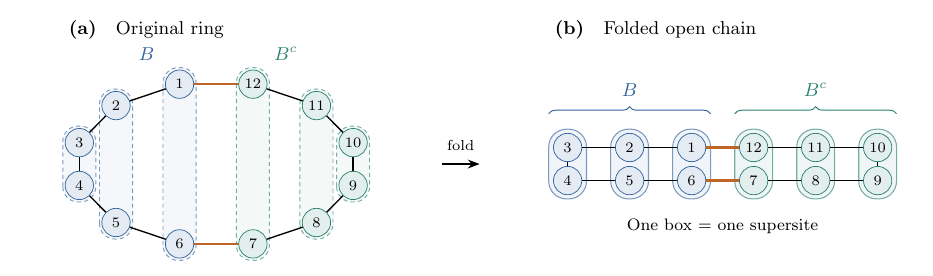}
    \caption{Folding a ring into an open chain. (a) Sites in $B$ and $B^c$ are paired from their respective endpoints toward their centers. (b) Each pair forms a supersite. All original interactions become on-site or nearest-neighbor interactions. The highlighted bonds cross the boundary. } 
    \label{fig:ring-folding}
\end{figure}
Choose a fixed integer $c_0$ large enough to contain the support of $V_{\partial B}$ and the offset in Proposition~\ref{prop:QL-cited-HMF} after this regrouping. Recall that, for $S_r:=N_{r+c_0}(\partial B)$,
\begin{align}
\mathcal D_{S_r^c}(X):=\tr_{S_r^c}(X)\otimes
\frac{I_{S_r^c}}{d^{|S_r^c|}},
\label{eq:QL-appendix-conditional-expectation}
\end{align}
with the tensor factors placed in their original site order.
Equivalently,
\begin{align}
\mathcal D_{S_r^c}(X)=\int_{\mathcal U(S_r^c)}UXU^\dagger dU,
\label{eq:QL-unitary-average}
\end{align}
where $dU$ is normalized Haar measure and each $U$ acts as the identity on $S_r$. This formula shows that $\mathcal D_{S_r^c}$ is contractive and preserves adjoints. It fixes every operator supported on $S_r$. In particular, $\mathcal D_{S_r^c}(I)=I$ and $\|I-\mathcal D_{S_r^c}\|_{\infty \to \infty}\leq2$.

Given this, we first prove the following result, which shows that $\Delta_B$, which is the difference between the effective Hamiltonians of $\sigma$ and $\tau_B \otimes \sigma_{\beta, B^c}$ is exponentially concentrated near $\partial B$.  

\begin{lemma}[Boundary estimates for the 1D Gibbs state]
\label{lem:QL-boundary-estimates}
There exist constants $A,a>0$ and $M\geq1$ such that
\begin{align}
\|\Delta_B-\mathcal D_{S_r^c}\left(\Delta_B\right)\|&\leq Ae^{-ar}, \qquad \|\Delta_B\|\leq\log M,
\label{eq:QL-logarithmic-locality}\\
\bigl\|\sigma_\beta (\tau_B^{-1}\otimes\sigma_{\beta,B^c}^{-1})\bigr\| & \leq M, \qquad \bigl\|(\tau_B\otimes\sigma_{\beta,B^c}) \sigma_\beta^{-1}\bigr\|\leq M.
\label{eq:QL-relative-norm-bounds}
\end{align}
\end{lemma}

\begin{proof}
We first prove Eq.~\eqref{eq:QL-relative-norm-bounds}. Let $G$ be the constant in Proposition~\ref{prop:QL-cited-expansional}. By partial-trace cyclicity for factors supported on $B$,
\begin{align}
T_{B^c} =\frac{1}{Z_B}
  \tr_B \left[
  e^{-\beta H_B}e^{-\beta H}e^{\beta H_B}
  \right]e^{\beta H_{B^c}}
=\frac{1}{Z_B}
  \bigl(\tr_B e^{-\beta H}\bigr)e^{\beta H_{B^c}}
=\frac{Z}{Z_BZ_{B^c}}
  \sigma_{\beta,B^c}\tau_{B^c}^{-1}.
\label{eq:QL-partial-trace-identity}
\end{align}
Together with
\begin{align}
\widehat E_B =e^{-\beta H}e^{\beta(H_B+H_{B^c})}=\frac{Z}{Z_BZ_{B^c}} \sigma_\beta(\tau_B^{-1}\otimes\tau_{B^c}^{-1}),
\end{align}
this gives
\begin{align}
\sigma_\beta (\tau_B^{-1}\otimes\sigma_{\beta,B^c}^{-1})
&=\widehat E_B(I_B\otimes T_{B^c}^{-1}),
\nonumber\\
(\tau_B\otimes\sigma_{\beta,B^c})\sigma_\beta^{-1}
&=(I_B\otimes T_{B^c})\widehat E_B^{-1}.
\label{eq:QL-relative-factorization}
\end{align}
Hence Proposition~\ref{prop:QL-cited-expansional} implies
\begin{align}
\bigl\|\sigma_\beta
(\tau_B^{-1}\otimes\sigma_{\beta,B^c}^{-1})\bigr\| &\leq \|\widehat E_B\| \|T_{B^c}^{-1}\| \leq G^2, \nonumber \\
\bigl\|(\tau_B\otimes\sigma_{\beta,B^c})
\sigma_\beta^{-1}\bigr\| &\leq \|T_{B^c}\| \|\widehat E_B^{-1}\| \leq G^2.
\end{align}
This proves Eq.~\eqref{eq:QL-relative-norm-bounds} with $M:=G^2$.

We next bound $\|\Delta_B\|$. The positive operator
\begin{align}
(\tau_B^{-1/2}\otimes\sigma_{\beta,B^c}^{-1/2})
\sigma_\beta
(\tau_B^{-1/2}\otimes\sigma_{\beta,B^c}^{-1/2})
\end{align}
is similar to
$\sigma_\beta(\tau_B^{-1}\otimes\sigma_{\beta,B^c}^{-1})$. The bounds on this relative operator and its inverse place the spectrum of the positive operator in $[M^{-1},M]$. Consequently,
\begin{align}
M^{-1}(\tau_B\otimes\sigma_{\beta,B^c}) \leq \sigma_\beta
\leq M(\tau_B\otimes\sigma_{\beta,B^c}).
\label{eq:QL-state-comparison}
\end{align}
By operator monotonicity of the logarithm,
\begin{align}
\log(\tau_B\otimes\sigma_{\beta,B^c})-(\log M)I & \leq \log\sigma_\beta \leq \log(\tau_B\otimes\sigma_{\beta,B^c})+(\log M)I, \nonumber \\
-(\log M)I &\leq \Delta_B \leq (\log M)I,
\end{align}
and therefore $\|\Delta_B\|\leq\log M$.

It remains to show that $\Delta_B$ is exponentially localized near $\partial B$. Using $H=H_B+H_{B^c}+V_{\partial B}$, 
\begin{align}
\Delta_B &= -\beta H+\beta H_B -I_B\otimes\log\sigma_{\beta,B^c} +\log \left(\frac{Z_B}{Z}\right)I \nonumber\\
&=I_B\otimes  \bigl(-\log\sigma_{\beta,B^c}+\log\tau_{B^c}\bigr) -\beta V_{\partial B} +\log \left(\frac{Z_BZ_{B^c}}{Z}\right)I.
\label{eq:QL-logarithmic-decomposition}
\end{align}
Let $W_{B^c,r}$ be the approximant from Proposition~\ref{prop:QL-cited-HMF}. By the choice of $c_0$, the map $\mathcal D_{S_r^c}$ fixes $I$, $V_{\partial B}$, and $I_B\otimes W_{B^c,r}$. Thus
\begin{align}
\|\Delta_B-\mathcal D_{S_r^c}\left(\Delta_B\right)\| &=\bigl\|(I-\mathcal D_{S_r^c})
  \bigl[I_B\otimes
  (-\log\sigma_{\beta,B^c}+\log\tau_{B^c})\bigr]\bigr\| \nonumber \\
&=\bigl\|(I-\mathcal D_{S_r^c}) \bigl[I_B\otimes (-\log\sigma_{\beta,B^c}+\log\tau_{B^c}  -W_{B^c,r})\bigr]\bigr\| \nonumber\\
&\leq 2\bigl\|-\log\sigma_{\beta,B^c}+\log\tau_{B^c} -W_{B^c,r}\bigr\| \nonumber \\
&\leq 2A_{\mathrm{mfo}}e^{-a_{\mathrm{mfo}}r} =Ae^{-ar},
\end{align}
where $A:=2A_{\mathrm{mfo}}$ and $a:=a_{\mathrm{mfo}}$.
\end{proof}

\subsection{Proof of Lemma~\ref{lem:left-boundary}}
\label{app:left-boundary}

\begin{proof}
Differentiating Eq.~\eqref{eq:QL-analytic-family} and using $\sigma_\beta^{it}=Z^{-it}e^{-it\beta H}$, we obtain
\begin{align}
\frac{d}{dt}F_B(it) &=i\sigma_\beta^{it}\Delta_B (\tau_B^{-it}\otimes\sigma_{\beta,B^c}^{-it})
\nonumber\\
&=i\bigl(e^{-it\beta H}\Delta_Be^{it\beta H}\bigr)F_B(it)
\nonumber\\
&=:iY_B(t)F_B(it), \qquad F_B(0)=I.
\label{eq:QL-relative-evolution}
\end{align}

Fix $r\geq0$ and put $q:=\lfloor r/2\rfloor$. By Lemma~\ref{lem:QL-boundary-estimates} and contractivity of $\mathcal D_{S_q^c}$,
\begin{align}
\|\Delta_B-\mathcal D_{S_q^c}(\Delta_B)\| \leq Ae^{-aq},\qquad 
\|\mathcal D_{S_q^c}(\Delta_B)\| \leq\|\Delta_B\|\leq\log M.
\label{eq:QL-initial-generator-tail}
\end{align}
If $S_r$ is the whole chain, the claim is immediate. Otherwise,
\begin{align}
\operatorname{dist}(S_q,S_r^c)\geq r-q,
\qquad
|S_q|\leq c_{\partial}(q+1)
\end{align}
for a fixed constant $c_{\partial}$. Using the unitary-average representation in Eq.~\eqref{eq:QL-unitary-average} and Proposition~\ref{prop:QL-cited-LR}, we obtain
\begin{align}
\|Y_B(t)-\mathcal D_{S_r^c}\left(Y_B(t)\right)\|
&\leq 2\|\Delta_B-\mathcal D_{S_q^c}(\Delta_B)\| +\bigl\|(I-\mathcal D_{S_r^c})
\bigl[e^{-it\beta H}(\mathcal D_{S_q^c}(\Delta_B))e^{it\beta H}\bigr]\bigr\| \nonumber\\
&\leq 2Ae^{-aq} +\sup_{U\in\mathcal U(S_r^c)}
 \bigl\|[e^{-it\beta H}(\mathcal D_{S_q^c}(\Delta_B))e^{it\beta H},U]\bigr\| \nonumber\\
&\leq 2Ae^{-aq} +C_{\mathrm{LR}}|S_q| \|\mathcal D_{S_q^c}(\Delta_B)\|
 e^{-b\operatorname{dist}(S_q,S_r^c)+v_0|t|} \nonumber\\
&\leq 2Ae^{-aq} +C_{\mathrm{LR}}c_{\partial}(q+1)(\log M)
e^{-b(r-q)+v_0|t|} \nonumber\\
&\leq C_1e^{-\nu r+v_0|t|}.
\label{eq:QL-generator-tail}
\end{align}
Here we choose any fixed $0<\nu<\frac12\min\{a,b\}$. The last inequality follows from $q\geq r/2-1$ and $r-q\geq r/2$, with the factor $q+1$ absorbed by the smaller exponential rate $\nu$.

Let $U_{B,r}(t)$ solve
\begin{align}
\frac{d}{dt}U_{B,r}(t) =i\mathcal D_{S_r^c}\left(Y_B(t)\right)U_{B,r}(t), \qquad U_{B,r}(0)=I.
\label{eq:QL-truncated-evolution}
\end{align}
Its generator is Hermitian and supported on $S_r$, so $U_{B,r}(t)$ is unitary and supported on $S_r$. Moreover,
\begin{align}
\frac{d}{dt} \bigl(U_{B,r}(t)^\dagger F_B(it)\bigr) = iU_{B,r}(t)^\dagger \bigl(Y_B(t)-\mathcal D_{S_r^c}\left(Y_B(t)\right)\bigr)F_B(it).
\end{align}
Integrating and using unitarity therefore gives
\begin{align}
\|F_B(it)-U_{B,r}(t)\| &\leq \int_{\min\{0,t\}}^{\max\{0,t\}}
\|Y_B(s)-\mathcal D_{S_r^c}Y_B(s)\| ds \nonumber\\
&\leq C_1e^{-\nu r} \int_{0}^{|t|}e^{v_0s} ds \nonumber\\
&\leq C_1|t|e^{-\nu r+v_0|t|} \nonumber\\
&\leq C_2e^{-\nu r+v|t|},
\label{eq:QL-truncated-evolution-error}
\end{align}
where $v>v_0$ is fixed and $C_2$ is chosen accordingly.

Finally, since $\mathcal D_{S_r^c}(U_{B,r}(t))=U_{B,r}(t)$,
\begin{align}
\|F_B(it)-\mathcal D_{S_r^c}\left(F_B(it)\right)\| &= \bigl\|(I-\mathcal D_{S_r^c})
  \bigl[F_B(it)-U_{B,r}(t)\bigr]\bigr\| \nonumber \\
&\leq 2\|F_B(it)-U_{B,r}(t)\| \nonumber\\
&\leq 2C_2e^{-\nu r+v|t|}.
\end{align}
This proves the lemma with $C:=2C_2$.
\end{proof}

\subsection{Proof of Lemma~\ref{lem:right-boundary}}
\label{app:right-boundary}

\begin{proof}
The powers of each fixed positive operator commute, so
\begin{align}
F_B(1+it) = \sigma_\beta^{it}
\Bigl[\sigma_\beta(\tau_B^{-1}\otimes\sigma_{\beta,B^c}^{-1})\Bigr]
(\tau_B^{-it}\otimes\sigma_{\beta,B^c}^{-it}).
\label{eq:QL-right-factorization}
\end{align}
The outer factors are unitary. Hence Lemma~\ref{lem:QL-boundary-estimates} gives
\begin{align}
\|F_B(1+it)\| = \bigl\|\sigma_\beta(\tau_B^{-1}\otimes\sigma_{\beta,B^c}^{-1})\bigr\|
\leq M.
\end{align}
\end{proof}

\subsection{Proof of Theorem~\ref{prop:verify-QL}}
\label{app:QL-interpolation}

\begin{proof}
For $\beta=0$, we have $Q_B=I$ and may take
$\widetilde Q_{B,r}=I$. So assume $\beta>0$.

Strict positivity of the states implies that $F_B(z)$ is entire.
For $0\leq x\leq1$, unitary invariance of the norm gives
\begin{align}
\|F_B(x+it)\|
&=\bigl\|\sigma_\beta^{it}
  \bigl[\sigma_\beta^x
  (\tau_B^{-x}\otimes\sigma_{\beta,B^c}^{-x})\bigr]
  (\tau_B^{-it}\otimes\sigma_{\beta,B^c}^{-it})\bigr\|
\nonumber\\
&=\bigl\|\sigma_\beta^x
  (\tau_B^{-x}\otimes\sigma_{\beta,B^c}^{-x})\bigr\|.
\label{eq:QL-strip-boundedness}
\end{align}
For each finite chain, the last expression is finite and continuous
in $x\in[0,1]$. Thus $F_B$ is bounded on the closed strip,
as required by Proposition~\ref{prop:QL-cited-three-lines}.

Fix $r\geq0$ and define
\begin{align}
G_{B,r}(z):=e^{z^2}(I-\mathcal D_{S_r^c})F_B(z).
\label{eq:QL-damped-family}
\end{align}
This family is also analytic and bounded on the strip, since
$\mathcal D_{S_r^c}$ is a fixed bounded linear map and
$|e^{(x+it)^2}|=e^{x^2-t^2}\leq e$ there.
By Lemma~\ref{lem:left-boundary},
\begin{align}
\sup_{t\in\mathbb R}\|G_{B,r}(it)\| \leq Ce^{-\nu r} \sup_{t\in\mathbb R}e^{-t^2+v|t|} =Ce^{v^2/4}e^{-\nu r}.
\label{eq:QL-damped-left}
\end{align}
Similarly, Lemma~\ref{lem:right-boundary} and $\|I-\mathcal D_{S_r^c}\|\leq2$ give
\begin{align}
\sup_{t\in\mathbb R}\|G_{B,r}(1+it)\| \leq 2M\sup_{t\in\mathbb R}e^{1-t^2} =2eM.
\label{eq:QL-damped-right}
\end{align}

Apply Proposition~\ref{prop:QL-cited-three-lines} to $z\mapsto\langle\xi,G_{B,r}(z)\eta\rangle$ for arbitrary unit vectors $\xi,\eta$. Taking the supremum over these vectors and using $F_B(1/2)=Q_B$, we obtain
\begin{align}
\|Q_B-\mathcal D_{S_r^c}\left(Q_B\right)\| &=e^{-1/4}\|G_{B,r}(1/2)\| \nonumber\\
&\leq e^{-1/4} \left(\sup_{t\in\mathbb R}\|G_{B,r}(it)\|\right)^{1/2}
  \left(\sup_{t\in\mathbb R}\|G_{B,r}(1+it)\|\right)^{1/2} \nonumber\\
&\leq e^{-1/4} \bigl(Ce^{v^2/4}e^{-\nu r}\bigr)^{1/2}(2eM)^{1/2} \nonumber\\
&=\sqrt{2MC} e^{v^2/8+1/4}e^{-\nu r/2}.
\label{eq:QL-interpolated-locality}
\end{align}
Applying Proposition~\ref{prop:QL-cited-three-lines} directly to the matrix elements of $F_B$, and using $\|F_B(it)\|=1$, also gives
\begin{align}
\|Q_B\| =\|F_B(1/2)\| \leq \left(\sup_{t\in\mathbb R}\|F_B(it)\|\right)^{1/2} \left(\sup_{t\in\mathbb R}\|F_B(1+it)\|\right)^{1/2} \leq M^{1/2}.
\label{eq:QL-interpolated-norm}
\end{align}

Set
\begin{align}
\widetilde Q_{B,r}:=\mathcal D_{S_r^c}\left(Q_B\right)
\in\mathcal A_{N_{r+c_0}(\partial B)}.
\label{eq:QL-interpolated-approximant}
\end{align}
By contractivity of $\mathcal D_{S_r^c}$,
\begin{align}
\|Q_B\|+\|\widetilde Q_{B,r}\| = \|Q_B\|+\|\mathcal D_{S_r^c}\left(Q_B\right)\| \leq 2\|Q_B\| \leq 2M^{1/2}.
\end{align}
Together with Eq.~\eqref{eq:QL-interpolated-locality}, this proves Assumption~\ref{asm:QL} with $\mu_Q:={\nu}/{2},C_Q:=\sqrt{2MC} e^{v^2/8+1/4},M_Q:=2M^{1/2}$.  All these constants, as well as $c_0$, depend only on $\beta$ and $d$.
\end{proof}

\section{Finite-window approximation of block updates}
\label{app:finite-window}

In this section, we prove Lemma~\ref{lem:finite-window-approximation}. Throughout this appendix, $\beta,d,\ell$ are fixed, and constants are independent of the system size. For a Hamiltonian $J$, let $\Delta_B(J)$ denote the logarithmic correction computed from $J$. The boundary estimates of Lemma~\ref{lem:QL-boundary-estimates} also apply after deleting interaction terms, since this preserves the interaction bound.

\subsection{Stability under a distant cut}

\begin{lemma}[Stability under a distant cut]
\label{lem:FW-distant-cut}
Let $W=N_r(B)$ and define $H^{(0)}:=H_W+H_{W^c}$. There exist constants $C,a>0$ such that, for every integer $r\geq0$,
\begin{align}
\|\Delta_B(H)-\Delta_B(H^{(0)})\| \leq Ce^{-ar}.
\label{eq:FW-distant-cut}
\end{align}
\end{lemma}

\begin{proof}
The claim is immediate if $W=\Lambda_n$. Bounded values of $r$ are covered by increasing $C$, using the uniform bound on $\|\Delta_B\|$ from Lemma~\ref{lem:QL-boundary-estimates}. We thus assume that $r$ is sufficiently large and $W\subsetneq\Lambda_n$.

We first consider deleting a single bond $e$ crossing $\partial W$. Let $J_1$ denote the Hamiltonian before deleting a bond $e$ crossing $\partial W$, and set $J_0:=J_1-v_e$, where $v_e$ is the interaction on that bond. Choose the endpoint $x\in W^c$ of $e$ and attach a control qubit $F$ to the spin at $x$. Define $\widehat J=J_0\otimes I_F+v_e\otimes|1\rangle\langle1|_F$. Thus the two values of the control select $J_0$ and $J_1$. This remains a nearest-neighbor Hamiltonian with the same interaction bound. To obtain uniform local dimension $2d$, attach inert qubits to all other sites. Their maximally mixed Gibbs factors cancel from the logarithmic correction; we suppress these inert tensor factors below.

Let $\rho_j=e^{-\beta J_j}/Z_j$, where $Z_j=\tr[e^{-\beta J_j}]$, and write $\Delta_j=\Delta_B(J_j)$. The Gibbs state of $\widehat J$ is $\widehat\rho=\sum_{j=0}^1p_j\rho_j\otimes |j\rangle\langle j|_F$, with $p_j=Z_j/(Z_0+Z_1)$. Since $F$ lies outside $B$ and $(J_0)_B=(J_1)_B$, the terms $\log p_j$ cancel between $\log\widehat\rho$ and the logarithm of its complementary marginal. Consequently,
\begin{align}
\widehat\Delta_B
:=\Delta_B(\widehat J)
=\Delta_0\otimes|0\rangle\langle0|_F
+\Delta_1\otimes|1\rangle\langle1|_F.
\end{align}

Apply Lemma~\ref{lem:QL-boundary-estimates} to the enlarged
chain, with constants $A_*,a_*>0$ and boundary offset $c_*$
corresponding to local dimension $2d$.
Set $s=\lfloor r/2\rfloor$ and
$\widehat S_s=N_{s+c_*}(\partial B)$.
For sufficiently large $r$, the control site $x$ lies outside
$\widehat S_s$. Hence depolarization over
$\widehat S_s^c$ gives
$\mathcal D_{\widehat S_s^c}(\widehat\Delta_B)=Q_s\otimes I_F$
for some Hermitian $Q_s$, and the lemma gives
$\|\widehat\Delta_B-Q_s\otimes I_F\|
\le A_*e^{-a_*s}$.
Compressing this difference to either control sector yields
$\|\Delta_j-Q_s\|\le A_*e^{-a_*s}$ for $j=0,1$.
Therefore,
\begin{align}
\|\Delta_B(J_1)-\Delta_B(J_0)\| \le \|\Delta_1-Q_s\|+\|\Delta_0-Q_s\| \le 2A_*e^{-a_*s}.
\end{align}

Apply this estimate successively to the at most two bonds crossing $\partial W$. The constants are uniform for the intermediate Hamiltonians, since deleting bonds preserves the interaction bound. The triangle inequality and $s=\lfloor r/2\rfloor$ prove Eq.~\eqref{eq:FW-distant-cut} after adjusting the constants.
\end{proof}

\subsection{Proof of Lemma~\ref{lem:finite-window-approximation}}

\begin{proof}
Let $W=N_r(B)$ and define $H^{(0)}=H_W+H_{W^c}$ and $H^{(1)}=H$. We use superscripts $(0)$ and $(1)$ for quantities computed from these Hamiltonians. If $W=\Lambda_n$, the claim is immediate. It suffices to consider $r$ larger than a fixed constant: the remaining radii are covered by increasing $C$, since the diamond distance between two channels is at most $2$.

By Lemma~\ref{lem:FW-distant-cut},
\begin{align}
\|\Delta_B^{(1)}-\Delta_B^{(0)}\| \leq Ce^{-ar}.
\label{eq:FW-logarithmic-stability}
\end{align}

We next transfer this estimate to the Petz amplitudes. Let $\sigma^{(j)}$ be the Gibbs state of $H^{(j)}$, with marginal $\sigma_{B^c}^{(j)}=\operatorname{tr}_B[\sigma^{(j)}]$. Since the two Hamiltonians agree inside $B$, they share the same internal Gibbs state $\tau_B$. For $j\in\{0,1\}$, the analytic family in Eq.~\eqref{eq:QL-analytic-family} therefore takes the form $F_j(z)=(\sigma^{(j)})^z\bigl(\tau_B^{-z}\otimes(\sigma_{B^c}^{(j)})^{-z}\bigr)$. By Eq.~\eqref{eq:conditional-petz-amplitude}, the corresponding conditional Petz amplitude is $Q_B^{(j)}:=(\sigma^{(j)})^{1/2}\bigl(\tau_B^{-1/2}\otimes(\sigma_{B^c}^{(j)})^{-1/2}\bigr)=F_j(1/2)$. Equation~\eqref{eq:QL-relative-evolution} then gives
\begin{align} 
\frac{d}{dt}F_j(it) = i\left(e^{-it\beta H^{(j)}}\Delta_B^{(j)}e^{it\beta H^{(j)}}\right)F_j(it), \qquad F_j(0)=I. 
\label{eq:FW-relative-evolution} 
\end{align} 

We first compare the two evolutions on the imaginary axis. Following the notation in Eq.~\eqref{eq:QL-relative-evolution}, write $Y_B^{(j)}(u):=e^{-iu\beta H^{(j)}}\Delta_B^{(j)}e^{iu\beta H^{(j)}}$ and $V:=\beta(H^{(1)}-H^{(0)})$. For $t\geq0$, we have
\begin{align}
\|F_1(it)-F_0(it)\| &\leq \int_0^t \|Y_B^{(1)}(u)-Y_B^{(0)}(u)\| du \nonumber\\
&\leq t\|\Delta_B^{(1)}-\Delta_B^{(0)}\|+\int_0^t \left\|e^{-iu\beta H^{(1)}}\Delta_B^{(0)}e^{iu\beta H^{(1)}}-Y_B^{(0)}(u)\right\| du \nonumber\\
&\leq t\|\Delta_B^{(1)}-\Delta_B^{(0)}\|+\int_0^t\int_0^u \|[V,Y_B^{(0)}(s)]\| ds du.
\label{eq:FW-evolution-comparison}
\end{align}
The first inequality follows from Eq.~\eqref{eq:FW-relative-evolution} and unitarity. The second separates the change in $\Delta_B$ from the change in the physical evolution, and the third applies Duhamel's formula to compare the evolutions generated by $\beta H^{(1)}$ and $\beta H^{(0)}$.

It remains to bound the commutator in the last line. Choose $q=\lfloor r/2\rfloor$ and set $A_q:=\mathcal D_{S_q^c}(\Delta_B^{(0)})$. The operator $A_q$ is supported on $S_q$, at distance at least $r-q-\mathcal O(1)$ from the removed bonds supporting $V$. Splitting $\Delta_B^{(0)}=A_q+(\Delta_B^{(0)}-A_q)$ and using Eqs.~\eqref{eq:QL-logarithmic-locality} and~\eqref{eq:QL-cited-LR}, we obtain
\begin{align}
\|[V,Y_B^{(0)}(s)]\| &\leq 2\|V\| \|\Delta_B^{(0)}-A_q\|+\left\|[V,e^{-is\beta H^{(0)}}A_qe^{is\beta H^{(0)}}]\right\| \nonumber\\
&\leq Ce^{-aq}+C(q+1)e^{-b(r-q)+v_0|s|} \nonumber\\
&\leq Ce^{-\nu r+v_0|s|}.
\label{eq:FW-cut-commutator}
\end{align}
Here $\|V\|\leq2\beta$, and $\|A_q\|$ is uniformly bounded. The factor $q+1$ comes from the size of $S_q$ and is absorbed by choosing a smaller decay rate $\nu>0$.

Substituting Eq.~\eqref{eq:FW-cut-commutator} into Eq.~\eqref{eq:FW-evolution-comparison} and using Eq.~\eqref{eq:FW-logarithmic-stability} gives the desired bound for $t\geq0$. The same argument with reversed time covers $t<0$. Thus, for every $t\in\mathbb R$,
\begin{align}
\|F_1(it)-F_0(it)\| &\leq |t|\|\Delta_B^{(1)}-\Delta_B^{(0)}\|+C\int_0^{|t|}u e^{-\nu r+v_0u} du \nonumber\\
&\leq C(|t|+t^2)e^{-\nu r+v_0|t|} \nonumber\\
&\leq Ce^{-\nu r+v|t|}.
\label{eq:FW-left-boundary}
\end{align}
In the second inequality, we decrease $\nu$ if necessary, and in the last, we choose $v>v_0$ to absorb the polynomial factor in $|t|$.

On the other boundary of the strip, Eq.~\eqref{eq:right-boundary-bound} gives $\|F_j(1+it)\|\leq M$. Set $G(z)=e^{z^2}(F_1(z)-F_0(z))$. As in Eq.~\eqref{eq:QL-strip-boundedness}, $G$ is analytic and bounded on the closed strip $0\leq\operatorname{Re}z\leq1$ for each finite chain. Its boundary norms satisfy
\begin{align}
    \sup_{t\in\mathbb R}\|G(it)\| = \sup_{t\in\mathbb R} e^{-t^2}\|F_1(it)-F_0(it)\|
    \leq Ce^{-\nu r} \sup_{t\in\mathbb R}e^{-t^2+v|t|} =
    Ce^{v^2/4}e^{-\nu r},
    \\[2pt]
    \sup_{t\in\mathbb R}\|G(1+it)\| \leq  \sup_{t\in\mathbb R}  e^{1-t^2}
    \left(\|F_1(1+it)\|+\|F_0(1+it)\|\right) \leq 2eM.
\end{align}
Applying Proposition~\ref{prop:QL-cited-three-lines} to the matrix elements of $G$ and taking the supremum over unit vectors gives
\begin{align}
    \|Q_B^{(1)}-Q_B^{(0)}\|
    &= e^{-1/4}\|G(1/2)\| \nonumber\\
    &\leq  e^{-1/4} \left(\sup_{t\in\mathbb R}\|G(it)\|\right)^{1/2}
    \left(\sup_{t\in\mathbb R}\|G(1+it)\|\right)^{1/2}  \nonumber\\
    &\leq  Ce^{-\nu r/2}.
    \label{eq:FW-amplitude-stability}
\end{align}
Finally, $H^{(0)}$ and $H^{(1)}$ have the same internal Gibbs state $\tau_B$. The factorization \eqref{eq:petz-amplitude-factorization} therefore gives
\begin{align}
    \Phi_B^{(j)}(\rho) = Q_B^{(j)} \left[\tau_B\otimes\tr_B\rho\right] (Q_B^{(j)})^\dagger.
\end{align}
Moreover, the Gibbs state of $H^{(0)}=H_W+H_{W^c}$ factorizes between $W$ and $W^c$, so Eq.~\eqref{eq:block-petz-map} implies $\Phi_B^{(0)}=\Phi_B^{H_W}\otimes\mathrm{id}_{W^c}$.

The map $\rho\mapsto\tau_B\otimes\tr_B\rho$ is CPTP. Expanding the difference of the two conjugations and applying the trace-norm multiplication inequality, also after tensoring with an arbitrary reference system, we conclude that
\begin{align}
\left\|\Phi_B - \left(\Phi_B^{H_W}\otimes\mathrm{id}_{W^c}\right) \right\|_\diamond
&=  \|\Phi_B^{(1)}-\Phi_B^{(0)}\|_\diamond \nonumber\\
&\leq \left(\|Q_B^{(1)}\|+\|Q_B^{(0)}\|\right)  \|Q_B^{(1)}-Q_B^{(0)}\|  \nonumber\\
&\leq 2\sqrt M \|Q_B^{(1)}-Q_B^{(0)}\|  \nonumber\\
&\leq  Ce^{-\nu r/2},
\end{align}
where the last two lines use Eqs.~\eqref{eq:QL-interpolated-norm} and~\eqref{eq:FW-amplitude-stability}. This proves Eq.~\eqref{eq:finite-window-approximation} with $\mu=\nu/2$.
\end{proof}

\section{Low-depth quantum circuits for Gibbs-state preparation}
\label{app:implementation}

In this section, we provide details of our algorithm to prepare Gibbs state via polylogarithmic-depth quantum circuits and thus prove Theorem~\ref{thm:dynamics-implementation}. Throughout, $\beta,d,\ell$ are fixed. We assume polynomial-time access to the precision bits of $\beta$, rational inputs $t,\varepsilon$, and a fixed universal gate set with polynomial-time-computable entries. The local interaction terms are given by classical descriptions from which their entries can be computed to $p$ bits of accuracy in time $\operatorname{poly}(b,p)$, where $b$ is the input bit length. 

We first state the central result, which states that the map $\Phi_B^{H_W}$ can be implemented by using only $\mathcal{O}(\text{poly}(|W|)$ gates.

\begin{lemma}[Local compilation]
\label{lem:local-compilation}
Let $W=N_r(B)$ and $w=|W|$. For every $\delta\in(0,1)$, the channel $\Phi_B^{H_W}$ presented in Lemma~\ref{lem:finite-window-approximation} can be implemented to diamond-norm error at most $\delta$ using
\begin{align}
\mathcal{O}\left(\operatorname{poly}\bigl(w,1+\log(1/\delta)\bigr)\right)
\label{eq:local-compilation-cost}
\end{align}
one- and two-qubit gates and ancilla qubits. The classical processing time is bounded by the same expression times $\operatorname{poly}(b)$.
\end{lemma}


\subsection{Roadmap on implementing $\Phi_B^{H_W}$}

Let $A=W\setminus B$, and denote by $\sigma_W=e^{-\beta H_W}/\tr[e^{-\beta H_W}]$ the Gibbs state of $H_W$, with marginal $\sigma_A=\tr_B[\sigma_W]$. Our goal is to factor the Petz recovery amplitude
\begin{align}
Q_B^{(W)} = \sigma_W^{1/2}(I_B\otimes\sigma_A^{-1/2})
\end{align}
into operators that admit efficient block encodings.

Let $H_B$ and $H_A$ contain the terms of $H_W$ supported entirely in $B$ and $A$, respectively, and write
\begin{align}
H_W=H_0+V_{\partial B},
\qquad H_0=H_B+H_A.
\end{align}
The boundary interaction $V_{\partial B}$ contains at most two local terms. Define
\begin{align}
E=e^{-\beta H_W/2}e^{\beta H_0/2},
\qquad C=e^{-\beta H_A/2},
\end{align}
and recall that $\tau_B=e^{-\beta H_B}/Z_B$, where $Z_B=\tr_B[e^{-\beta H_B}]$. The operator $E$ accounts for the boundary couplings, while $C$ describes the imaginary-time evolution on $A$. Define also
\begin{align}
T=C(CSC)^{-1/2}, \qquad S=\tr_B\left[ (\tau_B^{1/2}\otimes I_A)E^\dagger E (\tau_B^{1/2}\otimes I_A)\right].
\label{eq:relative-marginal}
\end{align}
Then, a direct calculation shows that the Petz recovery amplitude admits the exact factorization
\begin{align}
Q_B^{(W)}  = \sigma_W^{1/2}(I_B\otimes\sigma_A^{-1/2}) = E(\tau_B^{1/2}\otimes T).
\label{eq:exact-petz-factorization}
\end{align}
Let $F$ be a copy of $B$, and define the normalized purification
\begin{align}
|\sqrt{\tau_B}\rangle_{BF} := \sum_a (\tau_B^{1/2}|a\rangle_B)\otimes|a\rangle_F.
\end{align}
Using Eq.~\eqref{eq:exact-petz-factorization}, define
$V_W:\mathcal H_A\to\mathcal H_W\otimes\mathcal H_F$ by
\begin{align}
V_W|\psi\rangle_A
:=(E_{BA}\otimes I_F)
(I_B\otimes T_A\otimes I_F)
\bigl(|\sqrt{\tau_B}\rangle_{BF}\otimes|\psi\rangle_A\bigr),
\label{eq:structured-isometry}
\end{align}
where tensor factors are identified by their subsystem labels.
The definition of $S$ gives $V_W^\dagger V_W=T^\dagger ST=I_A$, so $V_W$ is an isometry. By Eq.~\eqref{eq:exact-petz-factorization}, tracing out $F$ implements the Petz recovery map, and hence
\begin{align}
\Phi_B^{H_W}(\rho_{BA}) = \tr_F\left[V_W\bigl(\tr_B[\rho_{BA}]\bigr)V_W^\dagger\right].
\label{eq:block-update-isometry}
\end{align}
To implement this channel, we cover the following three steps. 

\begin{enumerate}
\item \textbf{Block encode the boundary correction $E$.}
Lemma~\ref{lem:efficient-boundary-expansional} constructs a block encoding of $E$ with constant normalization and polynomial cost.

\item \textbf{Block encode the relative factor $T$.}
Combine the encoding of $E^\dagger E$ with preparation and unpreparation of $|\sqrt{\tau_B}\rangle_{BF}$ to encode $S$. The bounds on $E$ and $E^{-1}$ imply that $S$ is bounded above and below by positive constants. Lemma~\ref{lem:relative-inverse}, applied with $K=\beta H_A/2$, then gives an encoding of $T$ with polynomial normalization and cost.

\item \textbf{Amplify and implement the block update.}
Combine the encodings of $E$ and $T$ with the purification preparation to obtain a projected unitary encoding of $V_W$, and apply oblivious isometry amplification without postselection. Erase the original block $B$, apply the amplified circuit with fresh ancillas, and discard $F$ and the ancilla registers. The amplification and error analysis are given in the proof of Lemma~\ref{lem:local-compilation}.
\end{enumerate}

\subsection*{Block-encoding preliminaries}
\label{subsec:block-encoding-preliminaries}

We first collect the block-encoding conventions and standard constructions used below. 

\begin{definition}[Block encoding]
A unitary $U$ is an $(\alpha,a,\epsilon)$ block encoding of an operator $A$ if
\begin{align}
\left\|A-\alpha(\langle0^a|\otimes I) U(|0^a\rangle\otimes I)
\right\|\leq\epsilon.
\end{align}
\end{definition}

We choose normalization parameters satisfying $\|A\|\leq\alpha$. An encoding with normalization $\alpha$ can be padded to any $\alpha'\geq\alpha$ by appending a qubit rotation with zero-to-zero amplitude $\alpha/\alpha'$. This preserves the unnormalized error. Ancilla reigsters may likewise be padded with unused qubits.

The following statements are specialized to the forms needed in our constructions.

\begin{lemma}[Linear combinations~{\cite[Lemma 52]{gilyen2019quantum}}]
\label{lem:be-linear-combination}
Suppose $U_j$ is an $(\alpha_j,a,\epsilon_j)$ block encoding of $A_j$, after padding the ancilla registers to a common size. Let $w_j\geq0$, set $\alpha:=\sum_jw_j\alpha_j>0$, and suppose that a unitary $P$ prepares $P|0\rangle=\sum_j\sqrt{w_j\alpha_j/\alpha} |j\rangle$. Preparing this label, applying $U_j$ controlled by $j$, and unpreparing the label gives a block encoding of $\sum_jw_jA_j$ with normalization $\alpha$ and error at most $\sum_jw_j\epsilon_j$. For $M$ terms, the construction uses $a+\lceil \log_2 M\rceil$ ancilla qubits and one call each to the selection circuit, $P$, and $P^\dagger$.
\end{lemma}

\begin{lemma}[Products~{\cite[Lemma 53]{gilyen2019quantum}}]
\label{lem:be-product}
Suppose $U_A$ and $U_B$ are $(\alpha,a,\epsilon_A)$ and $(\beta,b,\epsilon_B)$ block encodings of $A$ and $B$, respectively. Applying $U_B$ followed by $U_A$, with distinct ancilla registers, gives an $(\alpha\beta,a+b,\alpha\epsilon_B+\beta\epsilon_A)$ block encoding of $AB$.
\end{lemma}

We also use compression by a prepared state. If $P_\chi|0\rangle=|\chi\rangle$, preparing and unpreparing $|\chi\rangle$ around an encoding of $A$ gives an encoding of $(\langle\chi|\otimes I)A(|\chi\rangle\otimes I)$ with the same normalization and no increase in block error, assuming exact state preparation.

\begin{lemma}[Polynomial eigenvalue transformation~{\cite[Theorem 56]{gilyen2019quantum}}]
\label{lem:be-polynomial-transform}
Let $U$ be an $(\alpha,a,\epsilon)$ block encoding of a Hermitian operator $A$. Let $p$ be a real degree-$d$ polynomial satisfying $|p(x)|\leq1/2$ for all $x\in[-1,1]$. For any $\delta\in(0,1)$, one can construct a $(1,a+2,4d\sqrt{\epsilon/\alpha}+\delta)$ block encoding of $p(A/\alpha)$.

The circuit uses $\mathcal O(d)$ calls to $U$ and $U^\dagger$, including one controlled call, and $\mathcal O((a+1)d)$ additional one- and two-qubit gates. Its description can be computed classically in time $\operatorname{poly}(d,\log(1/\delta))$. 
\end{lemma}

\begin{lemma}[Reciprocal approximation~{\cite[Corollary 67]{gilyen2019quantum}}]
\label{lem:be-reciprocal-polynomial}
For $h,\epsilon\in(0,1/2]$, there exists a real polynomial
$q$ of degree $\mathcal O(h^{-1}\log(1/\epsilon))$ such that
\begin{align}
\sup_{x\in[h,1]} \left|q(x)-\frac{h}{2x}\right|
\leq\epsilon, \qquad \sup_{x\in[-1,1]}|q(x)|\leq1.
\end{align}
Such an approximation may be chosen with either even or odd parity. In particular, $q/2$ satisfies the boundedness condition of Lemma~\ref{lem:be-polynomial-transform}.
\end{lemma}

\begin{lemma}[Oblivious isometry amplification~{\cite[Theorem 28]{gilyen2019quantum}}]
\label{lem:be-isometry-amplification}
Let $U$ be unitary, let $\Pi_{\rm in},\Pi_{\rm out}$ be orthogonal projectors, and let $W:\operatorname{im}(\Pi_{\rm in}) \to\operatorname{im}(\Pi_{\rm out})$ be an isometry, extended by zero outside $\operatorname{im}(\Pi_{\rm in})$. Suppose that, for an odd positive integer $m$,
\begin{align}
\left\|\Pi_{\rm out}U\Pi_{\rm in} - \sin \left(\frac{\pi}{2m}\right)W \right\| \leq\epsilon.
\end{align}
Then there is an amplified unitary $\widetilde U$ satisfying
\begin{align}
\left\| \Pi_{\rm out}\widetilde U\Pi_{\rm in} - W \right\| \leq 2m\epsilon.
\end{align}
The construction uses $m$ calls to $U$ or $U^\dagger$, $\mathcal O(m)$ projector-controlled operations and single-qubit gates, and at most one additional ancilla. 
\end{lemma}

\subsection{Block encoding the boundary correction}
\label{subsec:boundary-expansional}

We first construct a block encoding of $E=e^{-\beta H_W/2}e^{\beta H_0/2}$. The following lemma gives a block encoding of $E$ with normalization independent of $w$.

\begin{lemma}[Block encoding of the boundary expansional $E$]
\label{lem:efficient-boundary-expansional}
Fix $\beta, d$. 
There is a constant $C_E\geq1$, independent of $w$, such that $\|E\|,~\|E^{-1}\|\leq C_E$. Moreover, for every $\eta\in(0,1)$, a $(C_E,q_E,\eta)$ block encoding of $E$ can be constructed using $\mathcal{O}(\operatorname{poly}(w,1+\log(1/\eta)))$ one- and two-qubit gates, as well as $q_E$ ancilla qubits. The classical processing time is also $\mathcal{O}(\operatorname{poly}(w,1+\log(1/\eta),b))$.
\end{lemma}

If $\beta=0$, then $E=I$, so the claim is immediate. So we assume $\beta>0$. For $0\leq s\leq\beta/2$, define
\begin{align}
\label{eq:V0s}
V_0(s):=e^{-sH_0}V_{\partial B}e^{sH_0}.
\end{align}
We first construct a constant-normalized block encoding of $V_0(s)$ and then use it in the Dyson expansion to prove Lemma~\ref{lem:efficient-boundary-expansional}.

\begin{proposition}[Block encoding of $V_0(s)$]
\label{prop:boundary-coherent-encoding}
Under the assumptions of Lemma~\ref{lem:efficient-boundary-expansional}, there is a constant $L\geq1$ such that, for every $\xi\in(0,1)$, an $(L,q_V,\xi)$ block encoding of $V_0(s)$ can be constructed using $\operatorname{poly}(w,\log(1/\xi))$ one- and two-qubit gates and ancilla qubits. The construction acts coherently on a binary time register representing $s\in[0,\beta/2]$ with $\mathcal O(\log(1/\xi))$ bits. The classical processing time is polynomial in $w$, $\log(1/\xi)$, and $b$.
\end{proposition}

\begin{proof}
We first bound the coefficient weight of a truncated commutator expansion, then implement it using the circuit in Fig.~\ref{fig:v0-block-encoding}.

\paragraph{Commutator expansion and coefficient bounds.}
Write $H_0=\sum_j h_j$ and $V_{\partial B}=v_1+v_2$, where each $v_\nu$ is supported on a boundary bond (missing boundary terms are set to zero). 
We write each term in the Pauli-basis, $h_j=\sum_P c_{j,P}P$ and $v_\nu=\sum_P b_{\nu,P}P$. Since these operators act on at most two sites and have norm at most one, their Pauli coefficient $\ell_1$ norms are bounded by a constant. Choose a constant $g\geq1$ strictly larger than this bound, and pad each local coefficient list with cancelling $\pm I$ terms to make its total weight exactly $g$. Repeated Pauli operators in these lists retain separate labels. The local coefficients can be computed classically to any prescribed precision in polynomial time.

With $\operatorname{ad}_{H_0}(X):=[H_0,X]$, Eq.~\eqref{eq:V0s} becomes
\begin{align}
V_0(s) = \sum_{\nu=1}^{2}\sum_{k=0}^{\infty} \frac{(-s)^k}{k!}\operatorname{ad}_{H_0}^{k}(v_\nu).
\label{eq:boundary-commutator-expansion}
\end{align}
For each $\nu$, expanding the nested commutators gives ordered sequences of local bonds, which we call histories. For a nonzero contribution, each selected bond must intersect the interval spanned by the initial bond supporting $v_\nu$ and all previously selected bonds. The selected bond either lies inside this interval or extends it by one site, and we call the latter a growth step. After $m$ growth steps, there are at most $2+m$ interior bonds and at most two bonds enlarging the interval. We pad the choices with zero terms so that these counts are always exactly $2+m$ and two, respectively. 
For counting purposes, a padded growth choice still increments $m$.

Let $F_{k,m}$ count these (padded) histories of length $k$ with $m$ growth steps, for a fixed initial boundary, and set $F_k=\sum_m F_{k,m}$. With $F_{0,m}=\mathbf1_{m=0}$ and $F_{k,-1}=0$,
\begin{align}
F_{k+1,m} = (2+m)F_{k,m}+2F_{k,m-1}, \qquad \sum_{k=0}^{\infty}F_k\frac{x^k}{k!} = \exp\bigl(2x+2(e^x-1)\bigr).
\label{eq:boundary-history-generating-function}
\end{align}
Each commutator contributes two ordered products, and the local Pauli coefficients have $\ell_1$ weight at most $g$. Thus a length-$k$ history has total coefficient weight $g(2g)^k$ before cancellations. Consequently, the total weight of the expansion is bounded by
\begin{align}
2g\sum_{k=0}^{\infty}F_k\frac{(2gs)^k}{k!} \leq 2g\exp\bigl(2g\beta+2(e^{g\beta}-1)\bigr) =:L.
\label{eq:boundary-coefficient-bound}
\end{align}
This bounds both the LCU normalization and the Pauli coefficient $\ell_1$ norm after combining repeated terms. The tail beyond order $K$ has weight at most
\begin{align}
2g\sum_{k>K}F_k\frac{(2gs)^k}{k!} \leq 2^{-K}2g\exp\bigl(4g\beta+2(e^{2g\beta}-1)\bigr).
\label{eq:boundary-commutator-tail}
\end{align}
Thus $K=\mathcal O(\log(1/\xi))$ suffices for truncation error at most $\xi/3$.

\begin{figure}
    \centering
    \includegraphics[width=1\linewidth]{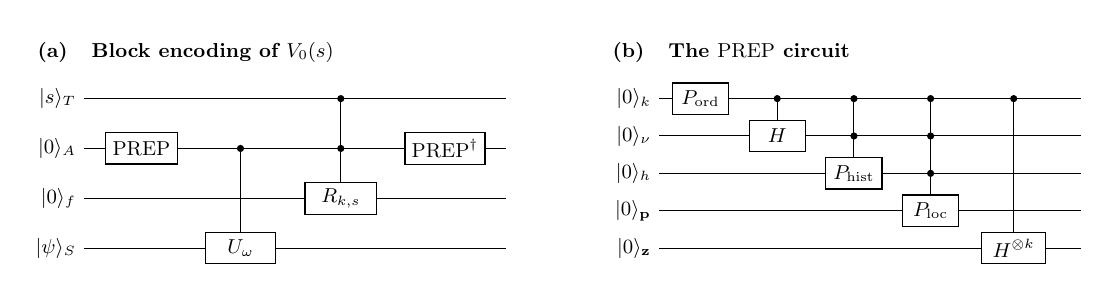}
    \caption{Block encoding of $V_0(s)$. (a) Prepared labels control the Pauli operations, while a flag rotation supplies the time dependence. Unpreparing the labels gives Eq.~\eqref{eq:boundary-truncated-block}. (b) Label preparation: choose the expansion order, a sequence of local terms, their Pauli components, and the left/right commutator branches, as in Eq.~\eqref{eq:boundary-prepared-state}. Dummy branches and ancilla registers are omitted.}
\label{fig:v0-block-encoding}
\end{figure}

\paragraph{Label preparation: Fig.~\ref{fig:v0-block-encoding}(b).}
The label register $A$ stores $\omega=(k,\nu,h,\mathbf p,\mathbf z)$. Here $k$ is the commutator order, $\nu\in\{1,2\}$ selects the initial boundary term $v_\nu$, and $h=(j_1,\ldots,j_k)$ records the selected local terms. The labels $\mathbf p=(p_0,\ldots,p_k)$ select one Pauli from each of $v_\nu,h_{j_1},\ldots,h_{j_k}$, while $\mathbf z=(z_1,\ldots,z_k)$ selects left or right multiplication at each commutator. For example, $k=2$, $\nu=1$, and $h=(j_1,j_2)$ specify $[h_{j_2},[h_{j_1},v_1]]$; after selecting Paulis $P_0,P_1,P_2$, the choices $\mathbf z=(L,R)$ give $-P_1P_0P_2$, where $L$ and $R$ denote the branches $PX$ and $-XP$ of $[P,X]$, respectively. (See Example~\ref{ex:V0} for more concrete details.) Registers have fixed capacity for order $K$ ($K$ bond labels, $K+1$ Pauli labels, and $K$ direction labels).

First, $P_{\rm ord}$ prepares order $k$ with probability $\lambda_k/L$, where $\lambda_k:=2gF_k(g\beta)^k/k!$. Since $\sum_{k=0}^{K}\lambda_k\leq L$, the remaining probability is assigned to a dummy label $\perp$. On each non-dummy branch, a Hadamard gate prepares $\nu\in\{1,2\}$ uniformly.

Next, $P_{\rm hist}$ prepares a uniform padded history conditional on $k$ and $\nu$. Let $N_{r,m}$ denote the number of possible continuations by $r$ further steps after $m$ growth steps, including the padded zero choices. These counts satisfy
\begin{align}
N_{0,m}=1,\qquad
N_{r,m}=(2+m)N_{r-1,m}+2N_{r-1,m+1}.
\label{eq:boundary-history-completions}
\end{align}
The table for $r+m\leq K$ has $\mathcal O(K^2)$ entries, each with $\mathcal O(K\log(K+2))$ bits, and $N_{k,0}=F_k$. At each step, choose the interior branch with probability $(2+m)N_{r-1,m}/N_{r,m}$ and the growth branch with probability $2N_{r-1,m+1}/N_{r,m}$, then choose uniformly among its padded slots. The probabilities along any complete history telescope to $1/F_k$. The selected bond indices are computed by reversibly updating the interval endpoints.

For the resulting history $(j_1,\ldots,j_k)$, $P_{\rm loc}$ prepares the local Pauli labels $p_0,\ldots,p_k$ with probabilities $|b_{\nu,p_0}|/g$ and $|c_{j_t,p_t}|/g$, respectively. Zero terms use cancelling coefficient lists of weight $g$. Finally, Hadamard gates prepare the $k$ active direction bits $z_1,\ldots,z_k$ uniformly.


Denote $a_{\nu,h,\mathbf p} :=|b_{\nu,p_0}|\prod_{t=1}^k|c_{j_t,p_t}|$
and $p_\perp:=1-\sum_{k=0}^K\lambda_k/L$. Suppressing unused registers fixed to zero, the preparation proceeds as follows:
\begin{align}
|0\rangle_A
&\xrightarrow{P_{\rm ord}} \sum_{k=0}^K \sqrt{\frac{\lambda_k}{L}} |k\rangle_A +\sqrt{p_\perp} |\perp\rangle_A \nonumber\\
&\xrightarrow{H_\nu} \sum_{k,\nu} \sqrt{\frac{\lambda_k}{2L}} |k,\nu\rangle_A +\sqrt{p_\perp} |\perp\rangle_A \nonumber\\
&\xrightarrow{P_{\rm hist}} \sum_{k,\nu,h}
\sqrt{\frac{\lambda_k}{2LF_k}}  |k,\nu,h\rangle_A +\sqrt{p_\perp} |\perp\rangle_A \nonumber\\
&\xrightarrow{P_{\rm loc}} \sum_{k,\nu,h,\mathbf p}
\sqrt{\frac{\lambda_k a_{\nu,h,\mathbf p}} {2LF_k g^{k+1}}} |k,\nu,h,\mathbf p\rangle_A +\sqrt{p_\perp} |\perp\rangle_A \nonumber\\
&\xrightarrow{H_{\mathbf z}} \sum_{\omega} \sqrt{p_\omega} |\omega\rangle_A + \sqrt{p_\perp} |\perp\rangle_A = \mathrm{PREP}|0\rangle_A,
\label{eq:boundary-prepared-state}
\end{align}
where $H_\nu$ prepares the boundary label uniformly and $H_{\mathbf z}$ applies Hadamard gates to the $k$ active direction bits, both conditional on a non-dummy branch. Here $a_\omega:=a_{\nu,h,\mathbf p}$ and $p_\omega=\lambda_k a_\omega/ (2^{k+1}LF_k g^{k+1}) =(\beta/2)^k a_\omega/(k!L)$.

\paragraph{Selected operations: Fig.~\ref{fig:v0-block-encoding}(a).}
The label $\omega$ determines a unitary $U_\omega$: the ordered product of the selected Pauli operators, including their coefficient signs, the commutator signs, and $(-1)^k$. Specifically, starting from $Q_0=P_{p_0}$, set $Q_t=P_{p_t}Q_{t-1}$ if $z_t=0$ and $Q_t=Q_{t-1}P_{p_t}$ if $z_t=1$. No explicit multiplication of Pauli matrices is needed. To implement $Q_k$, first apply $P_{p_t}$ for $t=k,\ldots,1$, conditional on $z_t=1$; then apply $P_{p_0}$; finally apply $P_{p_t}$ for $t=1,\ldots,k$, conditional on $z_t=0$. The coefficient signs and the factor $(-1)^{k+\sum_t z_t}$ are implemented by phases controlled by the labels. Using controls that restrict each step to $t\leq k$, this gives a fixed circuit with $2K+1$ controlled local Pauli slots, of which at most $k+1$ are active on each label branch, with polynomial overhead for the label controls and bond locations.

The time register $T$ controls a rotation on an additional flag $f$ satisfying $\langle0|R_{k,s}|0\rangle=(2s/\beta)^k$, with value one for $k=0$. On the dummy branch, this flag is flipped. Applying $\mathrm{PREP}$, the selected operations, and $\mathrm{PREP}^{\dagger}$ therefore gives the desired truncated version of $V_0(s)$:
\begin{align}
L\langle0|_{Af}U_s|0\rangle_{Af}
&= L\langle0|_{Af}
(\mathrm{PREP}^{\dagger}\otimes I_f)
\mathrm{SELECT}_s
(\mathrm{PREP}\otimes I_f)
|0\rangle_{Af}
\nonumber\\
&= L\sum_\omega
\sqrt{p_\omega}\sqrt{p_\omega} 
\langle0|R_{k,s}|0\rangle U_\omega
\nonumber\\
&= L\sum_\omega
\frac{(\beta/2)^k a_\omega}{k!L}
\left(\frac{2s}{\beta}\right)^k U_\omega
\nonumber\\
&= \sum_{k=0}^{K}\frac{s^k}{k!}
\sum_{\nu=1}^{2}
\sum_{h,\mathbf p,\mathbf z}
a_{(k,\nu,h,\mathbf p,\mathbf z)}
U_{(k,\nu,h,\mathbf p,\mathbf z)}
\nonumber\\
&= \sum_{\nu=1}^{2}\sum_{k=0}^{K}
\frac{(-s)^k}{k!}
\operatorname{ad}_{H_0}^{k}(v_\nu).
\label{eq:boundary-truncated-block}
\end{align}
The dummy branch contributes zero because $\langle0|X|0\rangle=0$. The last equality collects the Pauli labels, commutator directions, and padded histories.

\paragraph{Precision and complexity.}
Choose $K=\mathcal O(1+\log(1/\xi))$ so that the truncation error in Eq.~\eqref{eq:boundary-commutator-tail} is at most $\xi/3$. Approximate each unpadded local Pauli coefficient list to $\ell_1$ error at most $\delta$, preserving real coefficients, and then pad it to total weight $g$. The fixed slack in $g$ ensures that this padding is possible for sufficiently small $\delta$. Telescoping the coefficient products bounds the resulting change in the truncated operator by $L(K+1)\delta/g$. We therefore choose $\delta\leq g\xi/[3L(K+1)]$, making this error at most $\xi/3$.

The label preparations use conditional rotations determined by polynomial-size tables. For a distribution on $M$ labels, the uniformly controlled rotation construction of Ref.~\cite{mottonen2005transformation} uses $\mathcal O(M)$ single-qubit rotations and CNOTs after padding to a power of two. If the distribution depends on an additional index with $S$ possible values, the cost is $\mathcal O(SM)$. Thus $P_{\rm ord}$ costs $\mathcal O(K)$ gates. At each history step, the interior/growth choice depends on $\mathcal O(K^2)$ possible pairs $(r,m)$, while the uniform-slot choice has $\mathcal O(K)$ outcomes and $\mathcal O(K)$ possible values of $m$. Consequently, $P_{\rm hist}$ costs $\mathcal O(K^3)$ gates over all steps. Each local Pauli list has constant size and is selected by one of $\mathcal O(w)$ bond indices, so $P_{\rm loc}$ costs $\mathcal O(Kw)$ gates. These preparations require no enumeration of complete histories.

The selected Pauli operations likewise access only the labels of the relevant slot, together with the order label. Scanning the $\mathcal O(w)$ possible bond locations for each of the $\mathcal O(K)$ slots gives $\mathcal O(Kw)$ controlled local operations. Decomposing their controls and reversibly updating the interval endpoints adds polynomial overhead in $\log(K+w+2)$. All preparation ancilla registers are included in the ancilla register, retained during the selected operations, and uncomputed by $\mathrm{PREP}^{\dagger}$.

For the time-controlled flag, reversibly compute the angle $2\arccos[(2s/\beta)^k]$, apply the corresponding rotation, and uncompute the arithmetic work. The case $k=0$ is implemented directly. The angle is computed to sufficient absolute accuracy uniformly over the represented time values, including the endpoints. This requires only polynomial resources in $K$, the time-register size, and the number of precision bits.

Let $U_s^{\rm id}$ denote the ideal circuit for the truncated expansion with the approximated local coefficients. Choose the arithmetic precision and rotation-synthesis accuracy so that the implemented circuit satisfies $\sup_s\|\widetilde U_s-U_s^{\rm id}\|\leq\xi/(3L)$. For a circuit with $G_0$ elementary operations before gate synthesis, an error budget of $\xi/(3LG_0)$ per operation suffices by telescoping. Since $L$ is constant and $G_0$ is polynomial in the relevant parameters, the required precision is logarithmic in $G_0/\xi$. Combining the truncation, coefficient, and circuit errors gives $\|V_0(s)-L\langle0|_{Af}\widetilde U_s|0\rangle_{Af}\|\leq\xi$, uniformly over the represented times.

Including preparation and unpreparation, reversible arithmetic, label controls, and gate synthesis, a conservative gate-count bound is
\begin{align}
G_V
=
\mathcal O \left(
(K^3+Kw) 
\operatorname{poly} \left(
\log(K+w+2), 1+\log(1/\xi)
\right)
\right)
=
\operatorname{poly} \left(w,1+\log(1/\xi)\right).
\label{eq:G_V}
\end{align}
The circuit depth and ancilla count obey the same polynomial bound. The classical tables and rotation angles can also be computed in polynomial time in these parameters and $b$.

\end{proof}

We below give a simple example to explicitly show how this block-encoding process works.

\begin{example}[Example of block-encoding $V_0(s)$]
\label{ex:V0}
Consider a four-qubit chain with $B=\{1\}$ and $H_0=h_1+h_2$, where $h_1=X_2X_3$, $h_2=Z_3Z_4$, and $V_{\partial B}=v=Z_1Z_2$. We illustrate the construction for the second-order term $[H_0,[H_0,v]]$.

Two single-qubit labels $a,b\in\{1,2\}$ select the local terms in $[h_b,[h_a,v]]$, and two direction bits $\mathbf z=(z_1,z_2)$ select the commutator branches. Four Hadamard gates prepare these labels uniformly. The selected unitary is one of the four signed Pauli products
\begin{align}
U_{ab,00}=h_bh_av, \qquad U_{ab,01}=-h_avh_b,
\qquad U_{ab,10}=-h_bvh_a,
\qquad U_{ab,11}=vh_ah_b.
\end{align}
Each product is implemented by applying three local Pauli operators in the specified order, together with its sign. For example, the label $(a,b,\mathbf z)=(1,2,01)$ selects $-h_1vh_2$: apply $h_2$, then $v$, then $h_1$, and include a minus sign on this branch.

Writing $\mathrm{PREP}_2$ for the four Hadamard gates, the circuit acts on an arbitrary system state as
\begin{align}
|0^4\rangle_A|\psi\rangle &\xrightarrow{\mathrm{PREP}_2}
\frac14
\sum_{a,b=1}^{2}\sum_{\mathbf z\in\{0,1\}^2}
|a,b,\mathbf z\rangle_A|\psi\rangle
\nonumber\\
&\xrightarrow{\mathrm{SELECT}_2}
\frac14
\sum_{a,b=1}^{2}\sum_{\mathbf z\in\{0,1\}^2}
|a,b,\mathbf z\rangle_A U_{ab,\mathbf z}|\psi\rangle
\nonumber\\
&\xrightarrow{\mathrm{PREP}_2^\dagger}
|0^4\rangle_A
\frac1{16}[H_0,[H_0,v]]|\psi\rangle
+|\perp_\psi\rangle,
\end{align}
where $(\langle0^4|_A\otimes I)|\perp_\psi\rangle=0$. Indeed, summing over the direction bits gives $\sum_{\mathbf z}U_{ab,\mathbf z}=[h_b,[h_a,v]]$, and summing over $a,b$ gives $[H_0,[H_0,v]]$.

Thus the circuit combines all sixteen signed products without precomputing their matrices or collecting their Pauli coefficients. Only the local terms and the rule for selecting their order and signs are needed. This example isolates $k=2$ (the full construction also prepares the order label and supplies the Taylor coefficients through the order weights and the time-controlled flag rotation). 
\end{example}

We now use Proposition~\ref{prop:boundary-coherent-encoding}
to prove Lemma~\ref{lem:efficient-boundary-expansional}. Our encoding scheme is heavily inspired by the construction of Ref.~\cite{kieferov2019simulating}.

\begin{proof}[Proof of Lemma~\ref{lem:efficient-boundary-expansional}]
If $\beta=0$, then $E=I$, so assume $\beta>0$. Let $L$ be the coefficient bound from Eq.~\eqref{eq:boundary-coefficient-bound}. For $E(s):=e^{-sH_W}e^{sH_0}$, direct differentiation gives $E'(s)=-E(s)V_0(s)$ and $(E(s)^{-1})'=V_0(s)E(s)^{-1}$. Since $\|V_0(s)\|\leq L$, integration yields $\|E(s)\|,\|E(s)^{-1}\|\leq e^{sL}$. Thus $C_E:=e^{\beta L/2}$ bounds both $\|E\|$ and $\|E^{-1}\|$.

\paragraph{Dyson expansion and time discretization.}
Iterating the integral equation for $E(s)$ gives
\begin{align}
E =\sum_{m=0}^{\infty}\frac{(-\beta/2)^m}{m!}
\mathbb E \left[V_0(t_{(1)})\cdots V_0(t_{(m)})\right],
\label{eq:boundary-expansional-dyson}
\end{align}
where $t_1,\ldots,t_m$ are times sampled i.i.d. from $[0,\beta/2]$, sorted as $t_{(1)}\leq\cdots\leq t_{(m)}$. The tail beyond order $J$ is at most $\sum_{m>J}(\beta L/2)^m/m!\leq2^{-J}e^{\beta L}$. Choose $J=\mathcal O(\log(1/\eta))$ to make this error at most $\eta/4$.

The norm of the $k$th term in the commutator expansion is bounded by $2gF_k(2gs)^k/k!$. Differentiating these bounds term by term gives a convergent series, uniformly for $s\in[0,\beta/2]$. Hence there is a constant $L_1$, independent of $w$, such that $\sup_{s\in[0,\beta/2]}\|V_0'(s)\|\leq L_1$. Replace each time by the left endpoint of its cell in a uniform mesh of $M=2^{q_t}$ points. Every sorted time changes by at most $\beta/(2M)$, so telescoping the products and summing their weights bounds the discretization error by $\beta^2L_1C_E/(4M)$. Choose $q_t=\mathcal O(\log(1/\eta))$ to make this at most $\eta/4$. The resulting operator is
\begin{align}
E_{J,M} :=\sum_{m=0}^{J}\frac{(-\beta/2)^m}{m!}
\mathbb E_M \left[V_0(t_{(1)})\cdots V_0(t_{(m)})\right],
\label{eq:boundary-discrete-dyson}
\end{align}
where $\mathbb E_M$ averages over independent uniform indices $j_i\in\{0,\ldots,M-1\}$ with $t_i=\beta j_i/(2M)$, including repeated indices.

\begin{figure}
    \centering
    \includegraphics[width=1\linewidth]{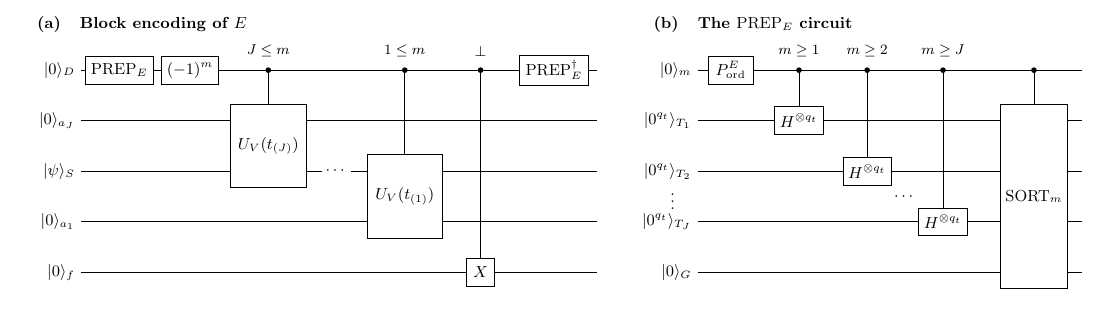}
    \caption{Block encoding of $E$. (a) Prepared labels control the ordered calls to $U_V(t)$, each using a separate ancilla. The phase $(-1)^m$ supplies the Dyson sign, and the dummy branch flips the flag. Unpreparing the labels gives Eq.~\eqref{eq:boundary-dyson-block} for ideal inner blocks. (b) Label preparation: choose the expansion order, prepare uniform time indices, and sort them reversibly, as in Eq.~\eqref{eq:boundary-dyson-prepared-state}. Sorting data are retained in $G$ until unpreparation.}
\label{fig:e-block-encoding}
\end{figure}

\paragraph{Label preparation: Fig.~\ref{fig:e-block-encoding}(b).}
Set $p_m:=(\beta L/2)^m/(m!C_E)$ for $0\leq m\leq J$ and $p_\perp:=1-\sum_{m=0}^{J}p_m$. The circuit $P_{\rm ord}^{E}$ prepares the order with probabilities $p_m$, assigning the remaining probability to a dummy label $\perp$. Conditional on a non-dummy order $m$, $H^{\otimes m q_t}$ applies Hadamard gates to the first $m$ time slots, preparing independent uniform indices in $\{0,\ldots,M-1\}$. (Here $H^{\otimes m q_t}$ denotes the controlled operation that applies $H^{\otimes q_t}$ to slot $i$ if $i\leq m$ on a non-dummy branch, and acts trivially on all other slots.) The circuit $\mathrm{SORT}$ then sorts the active indices reversibly, retaining its comparison bits in a ancilla register $G$.

Write $\mathbf j=(j_1,\ldots,j_m)$, $\mathbf j^\uparrow=(j_{(1)},\ldots,j_{(m)})$ for the sorted tuple, and $g_m(\mathbf j)$ for the retained sorting data. With $D$ containing the order, time, and sorting-ancilla registers, the preparation is
\begin{align}
|0\rangle_D
&\xrightarrow{P_{\rm ord}^{E}}
\sum_{m=0}^{J}\sqrt{p_m} |m,0,0\rangle_D
+\sqrt{p_\perp} |\perp,0,0\rangle_D
\nonumber\\
&\xrightarrow{H^{\otimes m q_t}}
\sum_{m=0}^{J}\sqrt{\frac{p_m}{M^m}}
\sum_{\mathbf j\in\{0,\ldots,M-1\}^{m}}
|m,\mathbf j,0\rangle_D
+\sqrt{p_\perp} |\perp,0,0\rangle_D
\nonumber\\
&\xrightarrow{\mathrm{SORT}}
\sum_{m=0}^{J}\sqrt{\frac{p_m}{M^m}}
\sum_{\mathbf j\in\{0,\ldots,M-1\}^{m}}
|m,\mathbf j^\uparrow,g_m(\mathbf j)\rangle_D
+\sqrt{p_\perp} |\perp,0,0\rangle_D
\nonumber\\
&=\mathrm{PREP}_E|0\rangle_D.
\label{eq:boundary-dyson-prepared-state}
\end{align}
Unused time slots are fixed to zero, and $m=0$ has a single empty tuple. The sorted indices specify the times $t_{(i)}=\beta j_{(i)}/(2M)$. The sorting data $g_m(\mathbf j)$ distinguish input tuples that produce the same sorted tuple, making sorting reversible and preserving their multiplicities in the average $\mathbb E_M$. These data are retained until $\mathrm{PREP}_E^\dagger$.


\paragraph{Selected products: Fig.~\ref{fig:e-block-encoding}(a).}
Let $U_V(t)$ denote the block encoding supplied by Proposition~\ref{prop:boundary-coherent-encoding}. First consider ideal inner blocks $V_0(t)/L$. For order $m$, apply $U_V(t_{(m)}),\ldots,U_V(t_{(1)})$ in that chronological order, using a separate ancilla register $a_j$ for each call. Their joint zero block is then $V_0(t_{(1)})\cdots V_0(t_{(m)})/L^m$. The fixed circuit has $J$ possible calls, run from $j=J$ down to $1$, active only on non-dummy branches with $j\leq m$. Apply the phase $(-1)^m$ on each non-dummy branch and flip the flag $f$, initially in $|0\rangle$, only on the dummy branch. Let $\mathrm{SELECT}_E$ denote these operations together with the controlled encoding calls, and define $U_E:=(\mathrm{PREP}_E^\dagger\otimes I_{\mathbf a f S}) \mathrm{SELECT}_E (\mathrm{PREP}_E\otimes I_{\mathbf a f S})$, where $\mathbf a=(a_1,\ldots,a_J)$. The label register $D$ is preserved by $\mathrm{SELECT}_E$, and the dummy branch contributes zero because $\langle0|X|0\rangle=0$.

Taking the zero block of each inner encoding and substituting $p_m=(\beta L/2)^m/(m!C_E)$, we obtain
\begin{align}
\langle0|_{D\mathbf a f}U_E|0\rangle_{D\mathbf a f}
&= \sum_{m=0}^{J}p_m(-1)^m
\mathbb E_M \left[\prod_{j=1}^{m}\bigl(\langle0|_{a_j}U_V(t_{(j)})|0\rangle_{a_j}\bigr)
\right] \nonumber\\
&= \sum_{m=0}^{J}p_m(-1)^m
\mathbb E_M \left[\frac{V_0(t_{(1)})\cdots V_0(t_{(m)})}{L^m}
\right] \nonumber\\
&= \frac{1}{C_E}
\sum_{m=0}^{J}\frac{(-\beta/2)^m}{m!}
\mathbb E_M \left[V_0(t_{(1)})\cdots V_0(t_{(m)})
\right] = \frac{E_{J,M}}{C_E}.
\label{eq:boundary-dyson-block}
\end{align}

\paragraph{Precision and complexity.}
Let $\widetilde U_V(t)$ be the implemented inner encoding and set $\widetilde V_0(t):= L\langle0|_a\widetilde U_V(t)|0\rangle_a$. Choose the inner-encoding and time-conversion accuracies so that $\|\widetilde V_0(t)-V_0(t)\|\leq\xi$ uniformly over the mesh. The time-conversion error is controlled by $L_1$. Since a block of a unitary is a contraction, $\|\widetilde V_0(t)\|\leq L$, and also $\|V_0(t)\|\leq L$. Telescoping an $m$-fold product therefore bounds its error by $mL^{m-1}\xi$. Summing with the Dyson weights gives
\begin{align}
\sum_{m=1}^{J}
\frac{(\beta/2)^m}{m!} mL^{m-1}\xi
\leq \frac{\beta C_E\xi}{2}.
\end{align}
Choose $\xi\leq\min\{1/2,\eta/(2\beta C_E)\}$, so this contribution is at most $\eta/4$. Implement the remaining outer operations with total unitary error at most $\eta/(4C_E)$. Together with the truncation and discretization errors, each bounded by $\eta/4$, the unnormalized block error is at most $\eta$.

The outer order preparation uses $\mathcal O(J)$ CNOTs and one-qubit rotations by the same state-preparation construction used above. Preparing the time registers requires $\mathcal O(Jq_t)$ controlled Hadamard gates. Sorting uses $\mathcal O(J^2)$ reversible comparisons and conditional swaps of $q_t$-bit registers; each has gate cost polynomial in $q_t$ and $\log(J+2)$. The comparison bits are retained until unpreparation. There are at most $J$ controlled calls to $U_V$, each with cost $G_V(w,\xi)$ in Eq.~\eqref{eq:G_V}. Including unpreparation and finite-precision synthesis, the total gate count is bounded by
\begin{align}
\mathcal O \left( J G_V(w,\xi) + J^2\operatorname{poly} \left( q_t,\log(J+2),1+\log(1/\eta) \right) \right).
\end{align}
Since $J$, $q_t$, and $\log(1/\xi)$ are $\mathcal O(1+\log(1/\eta))$, this is $\operatorname{poly}(w,1+\log(1/\eta))$. The ancilla count is polynomial in the same parameters, and the classical processing time is polynomial also in $b$. Thus the construction gives the claimed $(C_E,q_E,\eta)$ block encoding of $E$.
\end{proof}

Below, we give a simple example of this process. 

\begin{example}[Assembling a Dyson approximation from block encodings]
For illustration, take $\beta=2$, retain Dyson orders $m\leq2$, and use the two time points $0$ and $1/2$. Write $P=V_0(0)$ and $G=V_0(1/2)$, and assume ideal block encodings $U_0,U_{1/2}$ with normalization $L$.

For $m=1$, preparing a time qubit in $(|0\rangle+|1\rangle)/\sqrt2$, applying the corresponding block encoding, and unpreparing the time qubit gives $(P+G)/(2L)$ in the all-zero ancilla block. The Dyson sign contributes an additional minus sign.

For $m=2$, two Hadamard gates prepare $(|00\rangle+|01\rangle+|10\rangle+|11\rangle)/2$. We reversibly sort the time labels and call the larger-time block encoding first, using separate ancilla registers for the two calls. The four time branches therefore contribute $P^2/L^2$, $PG/L^2$, $PG/L^2$, and $G^2/L^2$. For example, the labels $10$ are sorted to $01$; applying $U_{1/2}$ followed by $U_0$ gives $PG/L^2$. Sorting work is retained until the sorting is undone. Unpreparing the time labels then gives $(P^2+2PG+G^2)/(4L^2)$.

To combine the orders, set $C_E=e^L$ and prepare $m=0,1,2$ with probabilities $1/C_E$, $L/C_E$, and $L^2/(2C_E)$, respectively. The remaining probability is assigned to a dummy branch whose encoded block is zero. Applying the corresponding circuits and unpreparing the order register gives
\begin{align}
\langle0|U_E|0\rangle
&= \frac1{C_E}I + \frac{L}{C_E} \left(-\frac{P+G}{2L}\right)
+ \frac{L^2}{2C_E} \left(\frac{P^2+2PG+G^2}{4L^2}\right)
\nonumber\\
&= \frac1{C_E} \left( I-\frac{P+G}{2} +\frac{P^2+2PG+G^2}{8} \right).
\end{align}
Thus the factors $L^m$ in the order weights cancel the normalization losses from the $m$ block-encoding calls, leaving the required Dyson coefficients. The operator in parentheses is the second-order Dyson approximation to $E=e^{-H_W}e^{H_0}$ with a two-point time discretization. The full construction increases the truncation order and time resolution and controls the implementation errors to achieve the prescribed accuracy.
\end{example}

\subsection{Block encoding the relative inverse-square-root factor}
\label{subsec:relative-inverse}

We next construct a block encoding of $T=C(CSC)^{-1/2}$. In our application, $C=e^{-K}$ with $K=\beta H_A/2$. Although $CSC$ may have exponentially small eigenvalues, the complete factor $T$ remains bounded under the assumed spectral bounds on $S$. Encoding $(CSC)^{-1/2}$ separately would introduce a large normalization before its growth is canceled by the prefactor $C$. We preserve this cancellation by rewriting $T$ as an integral of products involving only bounded factors and uniformly well-conditioned inverses. Each product can then be block encoded efficiently, and a quadrature rule combines them into an encoding of $T$, as shown in Fig.~\ref{fig:t-block-encoding}.

\begin{lemma}[Relative inverse square root]
\label{lem:relative-inverse}
Let $K$ be Hermitian with $\|K\|\leq L_0$ and admit an efficient block encoding with normalization $\mathcal O(1+L_0)$. Let $C=e^{-K}$, and suppose that $S$ admits an efficient block encoding with constant normalization and satisfies $c_-I\leq S\leq c_+I$ for fixed constants $0<c_-\leq c_+<\infty$.

For every $\eta\in(0,1)$, an $(\alpha_T,q_T,\eta)$ block encoding of $T=C(CSC)^{-1/2}$ can be constructed with
\begin{align}
\alpha_T = \mathcal O\bigl(1+L_0+\log(1/\eta)\bigr).
\label{eq:T-normalization}
\end{align}
The gate count, ancilla count, and classical processing time are polynomial in $1+L_0$, $\log(1/\eta)$, and the corresponding resources for the supplied encodings. Input encoding errors of $\eta^4/\operatorname{poly}(1+L_0+\log(1/\eta))$ suffice, for a sufficiently large fixed polynomial.
\end{lemma}

\begin{proof}
For real $u$, define
\begin{align}
D_u:=(I+e^{2(K+uI)})^{-1/2},
\qquad
F_u:=(I+e^{-2(K+uI)})^{-1/2},
\qquad
P_u:=D_uSD_u+F_u^2.
\label{eq:preconditioned-integrand}
\end{align}
Set $c:=\min(c_-,1)$ and $C_+:=\max(c_+,1)$. Since $D_u^2+F_u^2=I$, we have $cI\leq P_u\leq C_+I$. Thus $P_u$ has uniformly bounded condition number, independently of $u$, $L_0$, and the dimension.

\paragraph{A preconditioned resolvent representation.}
For a positive definite operator $A$, the inverse-square-root integral gives $A^{-1/2}=(2/\pi)\int_0^\infty(A+x^2I)^{-1} dx$. Applying this to $A=CSC$ and substituting $x=e^u$, we obtain
\begin{align}
T
&=C(CSC)^{-1/2} \nonumber\\
&=\frac{2}{\pi}\int_0^\infty
C(CSC+x^2I)^{-1} dx \nonumber\\
&=\frac{2}{\pi}\int_{\mathbb R}
e^uC(CSC+e^{2u}I)^{-1} du.
\label{eq:relative-resolvent-unpreconditioned}
\end{align}
To precondition the resolvent, set $R_u:=(C^2+e^{2u}I)^{1/2}$. Since $C=e^{-K}$, the definitions in Eq.~\eqref{eq:preconditioned-integrand} give $D_u=CR_u^{-1}$ and $F_u=e^uR_u^{-1}$. As $C$ and $R_u$ commute, we have $R_uP_uR_u=R_uD_uSD_uR_u+R_uF_u^2R_u = CSC+e^{2u}I$.  Consequently, 
\begin{align}
e^uC(CSC+e^{2u}I)^{-1} = e^uCR_u^{-1}P_u^{-1}R_u^{-1} =D_uP_u^{-1}F_u.
\end{align}
Substituting into Eq.~\eqref{eq:relative-resolvent-unpreconditioned} yields
\begin{align}
T =\frac{2}{\pi}\int_{\mathbb R} D_uP_u^{-1}F_u du.
\label{eq:relative-resolvent}
\end{align}
Also, $TT^\dagger = C(CSC)^{-1}C =CC^{-1}S^{-1}C^{-1}C =S^{-1}$, and hence $\|T\|\leq c_-^{-1/2}$.

\begin{figure}
    \centering
    \includegraphics[width=1\linewidth]{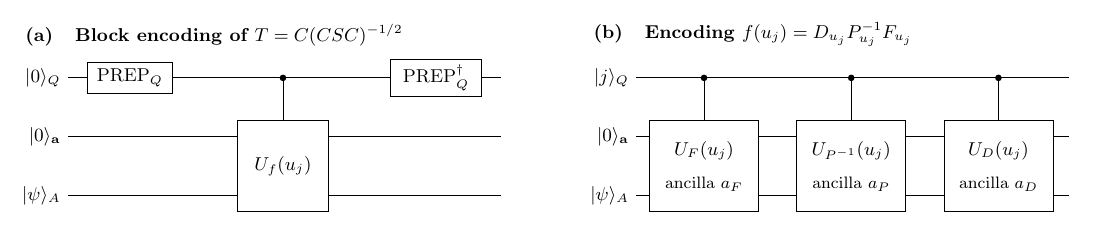}
   \caption{Block encoding of $T=C(CSC)^{-1/2}$. (a) The quadrature label selects $U_f(u_j)$; preparation and unpreparation supply the weights in Eq.~\eqref{eq:relative-lcu-block}. (b) At each node, the circuit applies $U_F(u_j)$, $U_{P^{-1}}(u_j)$, and $U_D(u_j)$ in that order. Each acts on the system and its indicated ancilla register, giving the product block in Eq.~\eqref{eq:relative-integrand-block}. Here $\mathbf a=(a_F,a_P,a_D)$, and dots denote register-valued controls.} 
\label{fig:t-block-encoding}
\end{figure}

\paragraph{Quadrature and label preparation: Fig.~\ref{fig:t-block-encoding}(a).}
Let $f(u):=D_uP_u^{-1}F_u$. The spectral bounds on $C$ give $\|D_u\|\leq\min\{1,e^{L_0-u}\}$ and $\|F_u\|\leq\min\{1,e^{L_0+u}\}$.
Together with $\|P_u^{-1}\|\leq1/c$, 
\begin{align}
\left\|T-\frac{2}{\pi}\int_{-R}^{R}f(u) du \right\| \leq \frac{2}{\pi c}
\left( \int_R^\infty e^{L_0-u} du +\int_{-\infty}^{-R}e^{L_0+u} du
\right) =\frac{4e^{L_0-R}}{\pi c}.
\end{align}
Choose $R:=L_0+\log(32/(c\eta))+1$, so that this truncation error is at most $\eta/4$.

Let $u_j,\omega_j$ be the nodes and weights of $m$-point Gauss-Legendre quadrature on $[-R,R]$. The weights are positive and sum to $2R$. For any operator-valued polynomial $p$ of degree at most $2m-1$, quadrature is exact on $p$. Consequently,
\begin{align}
\left\| \frac{2}{\pi}\int_{-R}^{R}f(u) du - \frac{2}{\pi}\sum_{j=1}^m\omega_jf(u_j) \right\| \leq \frac{8R}{\pi} \sup_{u\in[-R,R]}\|f(u)-p(u)\|.
\end{align}

To choose $p$, extend the integrand to $f(z):=e^zC(CSC+e^{2z}I)^{-1}$. For $z=u+i\theta$, set $M_{u,\theta}:=D_uSD_u+e^{2i\theta}F_u^2$. Then $f(u+i\theta)=e^{i\theta}D_uM_{u,\theta}^{-1}F_u$. For $|\theta|\leq\pi/8$, the bounds on $S$ and $D_u^2+F_u^2=I$ give $\operatorname{Re}M_{u,\theta}\geq cI/\sqrt2$. Hence $f$ is analytic in $|\operatorname{Im}z|<\pi/8$ and bounded by $\sqrt2/c$ on the closed strip.

Under the rescaling $u=Rx$, the Bernstein ellipse with parameter $\rho:=\exp(\pi/(16R))$ lies inside this strip. Chebyshev approximation gives a polynomial $p$ of degree at most $2m-1$ for which
\begin{align}
\sup_{u\in[-R,R]}\|f(u)-p(u)\| \leq \frac{2\sqrt2}{c(\rho-1)}\rho^{-(2m-1)}.
\end{align}
The preceding quadrature bound therefore gives error at most $C_{\rm GL}R^2e^{-\pi m/(8R)}$, where $C_{\rm GL}$ depends only on $c$. Choosing $m=\lceil(8R/\pi)\log(4C_{\rm GL}R^2/\eta)\rceil =\mathcal O(R\log(R/\eta))$ yields
\begin{align}
T_{R,m}:=\frac{2}{\pi}\sum_{j=1}^m\omega_j f(u_j), \qquad \left\| T_{R,m}-\frac{2}{\pi}\int_{-R}^{R}f(u) du \right\| \leq\frac{\eta}{4}.
\label{eq:relative-quadrature}
\end{align}

Once this is constructed, these quadrature weights determine the label preparation:
\begin{align}
|0\rangle_Q \xrightarrow{\mathrm{PREP}_Q} \sum_{j=1}^{m}\sqrt{\frac{\omega_j}{2R}} |j\rangle_Q.
\label{eq:relative-prepared-state}
\end{align}
Each label $j$ selects the integrand encoding at $u_j$. Preparation and unpreparation then supply the weight $\omega_j/(2R)$ in the encoded block. The nodes and weights are computed classically, and the label state is prepared using $\mathcal O(m)$ controlled rotations determined by conditional probability sums. 

\paragraph{Encoding one integrand: Fig.~\ref{fig:t-block-encoding}(b).}
For each quadrature node $u_j$, we construct a block encoding of $f(u_j)=D_{u_j}P_{u_j}^{-1}F_{u_j}$. Starting from the block encodings of $K$ and $S$, we encode $D_{u_j}$ and $F_{u_j}$, combine them to encode $P_{u_j}=D_{u_j}SD_{u_j}+F_{u_j}^2$, and then encode $P_{u_j}^{-1}$. Multiplying these three factor encodings gives $U_f(u_j)$. We explain this more specifically below.

First, fix a node $u=u_j$ and pad the supplied block encoding of $K$ to a known normalization $\Lambda\geq\max\{1,L_0\}$ with $\Lambda=\mathcal O(1+L_0)$. Define $d_u(x):=(1+e^{2(x+u)})^{-1/2}$ and $\varphi_u(x):=(1+e^{-2(x+u)})^{-1/2}$, so that $D_u=d_u(K)$ and $F_u=\varphi_u(K)$. Both functions are analytic and uniformly bounded in $|\operatorname{Im}x|\leq\pi/8$. Chebyshev approximation therefore gives real polynomials $q_{D,u}$ and $q_{F,u}$ approximating $d_u(\Lambda x)$ and $\varphi_u(\Lambda x)$ on $[-1,1]$ to error $\epsilon_{\rm app}\in(0,1/2]$, with degree $d_{\rm out}=\mathcal O(\Lambda\log(2\Lambda/\epsilon_{\rm app}))$. Because these functions take values in $[0,1]$ on the real axis, $|q_{D,u}(x)|,|q_{F,u}(x)|\leq1+\epsilon_{\rm app}$ on $[-1,1]$, so each polynomial divided by $4$ satisfies the boundedness condition of Lemma~\ref{lem:be-polynomial-transform}. Applying that lemma to the encoding of $K$ gives encodings $U_D(u)$ and $U_F(u)$ of $D_u$ and $F_u$ with normalizations $\alpha_D=\alpha_F=4$.

Second, use Lemma~\ref{lem:be-product} with distinct ancilla registers to encode $D_uSD_u$ and $F_u^2$. If $\alpha_S$ is the normalization of the supplied encoding of $S$, these products have normalizations $a:=\alpha_D^2\alpha_S$ and $b:=\alpha_F^2$. Prepare a branch qubit in $\sqrt{a/(a+b)}|0\rangle+\sqrt{b/(a+b)}|1\rangle$, select the corresponding product encoding, and unprepare the branch qubit. Lemma~\ref{lem:be-linear-combination} then gives an encoding of $P_u=D_uSD_u+F_u^2$ with normalization $a+b$. Pad this encoding to the constant normalization $\alpha_P:=\max\{a+b,2C_+\}$.

Third, set $h:=c/\alpha_P$. Since $cI\leq P_u\leq C_+I$, the spectrum of $P_u/\alpha_P$ lies in $[h,1/2]$, where $h>0$ is independent of $u$ and $L_0$. Lemma~\ref{lem:be-reciprocal-polynomial} gives a real polynomial $q_{\rm inv}$ approximating $h/(2x)$ on $[h,1]$ to error $\epsilon_{\rm app}$, bounded by one on $[-1,1]$, with degree $d_{\rm inv}=\mathcal O(h^{-1}\log(2/\epsilon_{\rm app}))=\mathcal O(\log(2/\epsilon_{\rm app}))$. Apply Lemma~\ref{lem:be-polynomial-transform} to $q_{\rm inv}/2$ and the encoding of $P_u$. Its output block approximates $(h/4)(P_u/\alpha_P)^{-1}=(c/4)P_u^{-1}$, giving an encoding $U_{P^{-1}}(u)$ with normalization $\alpha_{\rm inv}:=4/c$.

Finally, apply $U_F(u)$, $U_{P^{-1}}(u)$, and $U_D(u)$ in that circuit order, using distinct ancilla registers $\mathbf a=(a_F,a_P,a_D)$. By Lemma~\ref{lem:be-product}, their ideal zero block is
\begin{align}
\langle0|_{\mathbf a}U_f(u)|0\rangle_{\mathbf a} = \frac{D_u}{\alpha_D}\frac{P_u^{-1}}{\alpha_{\rm inv}}\frac{F_u}{\alpha_F} = \frac{f(u)}{\alpha_f}, \qquad \alpha_f:=\alpha_D\alpha_{\rm inv}\alpha_F.
\label{eq:relative-integrand-block}
\end{align}
Thus $\alpha_f$ is independent of $u$ and $L_0$. Compile $U_f(u_j)$ for each quadrature node and pad the ancilla registers to a common size, allowing selection by $j$ in the next step.

\paragraph{Assembly and accuracy: Fig.~\ref{fig:t-block-encoding}(a).} 
Prepare the quadrature label using $\mathrm{PREP}_Q$ from Eq.~\eqref{eq:relative-prepared-state}, apply $U_f(u_j)$ controlled by $j$, and unprepare the label. The resulting unitary satisfies
\begin{align}
\langle0|_{Q\mathbf a}U_T|0\rangle_{Q\mathbf a} = \sum_{j=1}^{m}\frac{\omega_j}{2R} \frac{f(u_j)}{\alpha_f} =\frac{\pi}{4R\alpha_f} T_{R,m} = \frac{T_{R,m}}{\alpha_T},
\label{eq:relative-lcu-block}
\end{align}
where $\alpha_T:=4R\alpha_f/\pi=\mathcal O(1+L_0+\log(1/\eta))$.

Choose each implemented integrand encoding $\widetilde U_f(u_j)$ to satisfy $\|f(u_j)-\alpha_f\langle0|_{\mathbf a} \widetilde U_f(u_j)|0\rangle_{\mathbf a}\|\leq\xi$, where $\xi:=\pi\eta/(16R)$. Since the quadrature weights are positive and sum to $2R$, these errors contribute at most $(2/\pi)\sum_j\omega_j\xi=(4R/\pi)\xi=\eta/4$ to the final unnormalized block error.

To justify the required input precision, suppose that the supplied encodings of $K$ and $S$ have unnormalized errors at most $\epsilon_{\rm in}$. For a degree-$d$ polynomial, Lemma~\ref{lem:be-polynomial-transform} gives output error $4d\sqrt{\epsilon/\alpha}+\delta$ from an input $(\alpha,a,\epsilon)$ encoding, where $\delta$ is the additional transformation-synthesis error. Choose this synthesis error at most $\epsilon_{\rm app}$ in each transformation. The resulting unnormalized errors in $D_u,F_u$ are then $\mathcal O(d_{\rm out}\sqrt{\epsilon_{\rm in}/\Lambda} +\epsilon_{\rm app})$. The product and sum constructions give an encoding of $P_u$ with error $\epsilon_P=\mathcal O( d_{\rm out}\sqrt{\epsilon_{\rm in}/\Lambda} +\epsilon_{\rm app}+\epsilon_{\rm in})$. Applying the inverse polynomial gives unnormalized error $\mathcal O(d_{\rm inv}\sqrt{\epsilon_P/\alpha_P} +\epsilon_{\rm app})$. The final product with $D_u,F_u$ adds their errors with only constant factors.

These two successive square-root error bounds show that input errors of $\epsilon_{\rm in}\leq \eta^4/\operatorname{poly}(1+L_0+\log(1/\eta))$ suffice to achieve integrand error at most $\xi$, for a sufficiently large fixed polynomial. Choosing $\epsilon_{\rm app}$ at the same scale also suffices. The degrees remain polynomial in $1+L_0+\log(1/\eta)$, and the bound $P_u\geq cI$ keeps the inverse-approximation constants independent of the smallest eigenvalue of $C$.

Round the quadrature nodes and weights to sufficient precision that their effect on the quadrature sum is at most $\eta/8$, preserving positive weights with total $2R$. The uniform strip bound gives a uniform derivative bound for $f$ on the real axis, so polynomially many precision bits suffice for the nodes as well as the weights. Implement the remaining label preparation and selection controls with total unitary error at most $\eta/(8\alpha_T)$, contributing another $\eta/8$ to the unnormalized block error. Together with integral truncation, quadrature, and integrand errors, each at most $\eta/4$, the total error is at most $\eta$.

The polynomial degrees, quadrature node counts, and required precision bits are all bounded by polynomials in $1+L_0+\log(1/\eta)$. The quadrature data and scalar approximation polynomials can be computed to the required precision in polynomial classical time, and Lemma~\ref{lem:be-polynomial-transform} provides polynomial-time classical construction of the corresponding transformation circuits. Including the resources of the supplied encodings, the gate count, ancilla count, and classical processing time satisfy the claimed polynomial bounds.
\end{proof}

\subsection{Isometry amplification and implementation of $\widetilde{\Phi}_B^{H_W}$}
\label{subsec:isometry-amplification}

We now combine the preceding constructions to prove Lemma~\ref{lem:local-compilation}, following Fig.~\ref{fig:isometry-local-compilation}. 

\begin{proof}[Proof of Lemma~\ref{lem:local-compilation}]
Fix $\delta\in(0,1)$ and set $\gamma:=\delta^2/(128\pi)$. We first construct a projected unitary encoding of $V_W$ with unnormalized error at most $\gamma$. We then amplify this encoding and discard the environment and encoding ancillas to implement the desired block update.

\paragraph{Constructing $S$ and encoding $T$.}
Lemma~\ref{lem:efficient-boundary-expansional} provides a block encoding of $E$ with constant normalization $C_E$ and bounds $\|E\|,\|E^{-1}\|\leq C_E$. To construct a block encoding of $S$ defined in Eq.~\eqref{eq:relative-marginal}, let $P_\tau|0\rangle_{BF}=|\sqrt{\tau_B}\rangle_{BF}$ and express $S$ as
\begin{align}
S=(\langle\sqrt{\tau_B}|_{BF}\otimes I_A)(E^\dagger E\otimes I_F)(|\sqrt{\tau_B}\rangle_{BF}\otimes I_A).
\end{align}
By Lemma~\ref{lem:be-product}, encodings of $E$ and $E^\dagger$ using separate ancilla registers give an encoding of $E^\dagger E$ with normalization $C_E^2$. Preparing and unpreparing the purification then gives an encoding of $S$ with the same normalization. Since $|B|$, $d$, and $\beta$ are fixed, the purification can be computed and prepared to state-vector error $\epsilon$ with cost polynomial in $b$ and $1+\log(1/\epsilon)$. Thus $S$ can be encoded to arbitrary precision with polynomial resources.

Moreover, the bounds on $E$ and $E^{-1}$ imply $C_E^{-2}I\leq E^\dagger E\leq C_E^2I$. Compressing these inequalities by the normalized purification gives
\begin{align}
C_E^{-2}I_A\leq S\leq C_E^2I_A.
\label{eq:S-conditioning}
\end{align}

Now apply Lemma~\ref{lem:relative-inverse} with $K=\beta H_A/2$, whose norm and local-term block-encoding normalization are $\mathcal O(w)$. This gives a block encoding of $T$ with normalization $\alpha_T=\mathcal O(1+w+\log(1/\gamma))$ and unnormalized error at most $\gamma/(C_E+1)^2$. The finer input accuracy required by the lemma is achieved using correspondingly accurate local-term encodings of $K$, encodings of $E$, and purification preparations. All required precision bits and circuit costs remain polynomial in $w$ and $1+\log(1/\gamma)$. Finally, $TT^\dagger=S^{-1}$ and Eq.~\eqref{eq:S-conditioning} imply $\|T\|\leq C_E$.

\begin{figure}
    \centering
    \includegraphics[width=\linewidth]{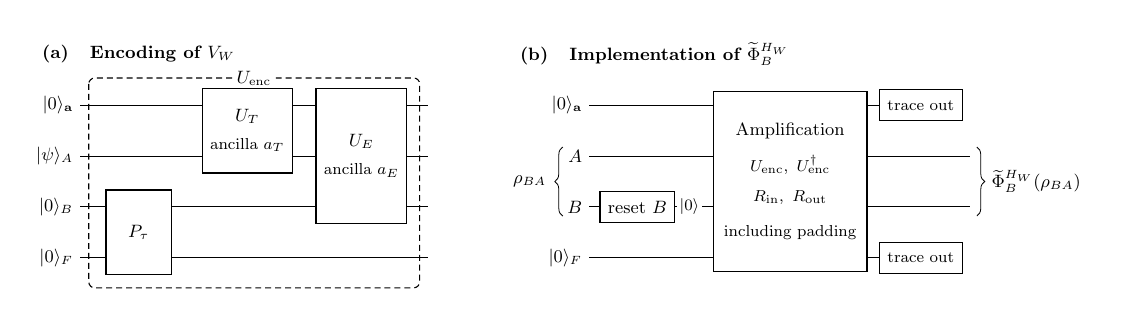}
    \caption{Isometry amplification and local compilation. (a) Prepare $|\sqrt{\tau_B}\rangle_{BF}$, then apply the encodings of $T_A$ and $E_{BA}$ using separate ancilla registers. Projecting only the ancilla registers gives Eq.~\eqref{eq:complete-isometry-encoding}. (b) Reset $B$ and amplify the complete encoding using its inverse and the projectors in Eq.~\eqref{eq:isometry-projectors}. Padding and amplification ancillas are included in $\mathbf a$. Discarding $F,\mathbf a$ implements Eq.~\eqref{eq:compiled-block-update} without postselection.}
    \label{fig:isometry-local-compilation}
\end{figure}

\paragraph{The complete encoding: Fig.~\ref{fig:isometry-local-compilation}(a).}
The circuit $U_{\rm enc}$ applies $P_\tau$, then the block encoding of $T_A$, and finally that of $E_{BA}$. The latter two use disjoint ancilla registers $a_T,a_E\subset\mathbf a$. By Lemma~\ref{lem:be-product}, this gives a projected unitary encoding of $V_W$ from Eq.~\eqref{eq:structured-isometry}, with normalization $\alpha:=C_E\alpha_T$.

Let $\widetilde E$ and $\widetilde T$ denote the zero blocks of the implemented encodings multiplied by $C_E$ and $\alpha_T$, respectively, with $\|\widetilde E-E\|\leq\epsilon_E$ and $\|\widetilde T-T\|\leq\epsilon_T$. Let $|\widetilde\tau\rangle_{BF}$ be the normalized state actually prepared, with $\||\widetilde\tau\rangle-|\sqrt{\tau_B}\rangle\|\leq\epsilon_\tau$. Since the two encodings use disjoint ancilla registers, the scaled projected block of $U_{\rm enc}$ equals $(\widetilde E_{BA}\otimes I_F)(I_B\otimes\widetilde T_A\otimes I_F)(|\widetilde\tau\rangle_{BF}\otimes I_A)$. Expanding its difference from $V_W$ into contributions from the three approximations and using $\|E\|,\|T\|\leq C_E$ gives
\begin{align}
\left\|\alpha(\langle0|_{\mathbf a}\otimes I_{ABF})U_{\rm enc}(|0\rangle_{\mathbf a BF}\otimes I_A)-V_W\right\| \leq \epsilon_E(C_E+\epsilon_T)+C_E\epsilon_T+C_E^2\epsilon_\tau \leq\gamma,
\label{eq:complete-isometry-encoding}
\end{align}
where the last inequality follows by choosing $\epsilon_E,\epsilon_T,\epsilon_\tau\leq\gamma/(C_E+1)^2$. The preceding constructions achieve these accuracies with polynomial resources.

To apply isometry amplification to this encoding, define the input and output projectors
\begin{align}
\Pi_{\rm in}:=|0\rangle\langle0|_{\mathbf a BF}\otimes I_A, \qquad \Pi_{\rm out}:=|0\rangle\langle0|_{\mathbf a}\otimes I_{ABF}.
\label{eq:isometry-projectors}
\end{align}
The input projector fixes all auxiliary registers to zero, while the output projector fixes only the encoding ancillas and leaves the full isometry output $ABF$ unrestricted. The reflections $R_{\rm in/out}:=2\Pi_{\rm in/out}-I$ and the corresponding projector-controlled gates can be implemented using reversible all-zero tests with polynomial overhead. Concretely, a linear-size Toffoli circuit computes the all-zero condition into a flag qubit; after the required flag-controlled operation, the computation is reversed to restore the flag and workspace to zero. All padding and amplification ancillas introduced below are included in $\mathbf a$, with the projectors extended accordingly.

\paragraph{Oblivious isometry amplification: Fig.~\ref{fig:isometry-local-compilation}(b).}
To match the normalization required by Lemma~\ref{lem:be-isometry-amplification}, choose an odd integer $m\geq\pi\alpha/2$ with $m=\mathcal O(\alpha)$ and set $\alpha':=\csc(\pi/(2m))$. Since $\alpha'\geq2m/\pi\geq\alpha$, a single-qubit rotation with zero-to-zero amplitude $\alpha/\alpha'$ pads the normalization to $\alpha'$. The resulting projected block approximates $V_W/\alpha'=\sin(\pi/(2m))V_W$ with error at most $\gamma/\alpha'$.

Let $U_{\rm pad}$ denote the padded circuit, with its additional qubit included in $\mathbf a$ and in the projectors. Embed $V_W$ in the circuit space by defining $\mathcal V_W:=(|0\rangle_{\mathbf a}\otimes V_W)(\langle0|_{\mathbf a BF}\otimes I_A)$. Since $V_W^\dagger V_W=I_A$, we have $\mathcal V_W^\dagger\mathcal V_W=\Pi_{\rm in}$ and $\Pi_{\rm out}\mathcal V_W=\mathcal V_W$. Equation~\eqref{eq:complete-isometry-encoding} and normalization padding give $\|\Pi_{\rm out}U_{\rm pad}\Pi_{\rm in}-\sin(\pi/(2m))\mathcal V_W\|\leq\gamma/\alpha'$, satisfying the hypothesis of Lemma~\ref{lem:be-isometry-amplification}.

Lemma~\ref{lem:be-isometry-amplification} therefore applies to this padded encoding. It uses $m$ calls to $U_{\rm pad}$ and its inverse, together with $\mathcal O(m)$ projector-controlled operations and single-qubit gates. Let $U_{\rm amp}$ denote the amplified unitary and define its induced isometry by $V_{\rm amp}:=U_{\rm amp}(|0\rangle_{\mathbf a BF}\otimes I_A)$. Set $J_W:=(\langle0|_{\mathbf a}\otimes I_{ABF})V_{\rm amp}$. Then
\begin{align}
\|J_W-V_W\|
&\leq\frac{2m\gamma}{\alpha'}
\leq\pi\gamma=:e,
\label{eq:amplification-success-error}
\end{align}
where the second inequality uses $2m\sin(\pi/(2m))\leq\pi$. The amplification acts on the complete padded circuit $U_{\rm pad}$, including the purification preparation.

To bound the full output error, we also control the component outside the output subspace. Since $V_{\rm amp}^\dagger V_{\rm amp}=I_A$, $V_W^\dagger V_W=I_A$, and $\|J_W\|\leq1$,
\begin{align}
\|(I-\Pi_{\rm out})V_{\rm amp}\|^2
&=\|I_A-J_W^\dagger J_W\|
\leq2\|J_W-V_W\|
\leq2e.
\label{eq:amplification-leakage}
\end{align}
Combining the errors inside and outside that subspace gives
\begin{align}
\|V_{\rm amp}-|0\rangle_{\mathbf a}\otimes V_W\|
&\leq e+\sqrt{2e}
\leq\frac{\delta^2}{128}+\frac{\delta}{8}
<\frac{\delta}{4}.
\label{eq:amplified-isometry-error}
\end{align}
This square-root contribution is why the initial encoding error was chosen to be of order $\delta^2$.

Finally, synthesize the remaining rotations and elementary gates to total unitary error at most $\delta/4$. For $G$ such gates, accuracy $\delta/(4G)$ per gate suffices. Let $\widetilde U_{\rm amp}$ denote the resulting implemented unitary circuit, and define its induced isometry by $\widetilde V_W:=\widetilde U_{\rm amp}(|0\rangle_{\mathbf a BF}\otimes I_A)$. Then
\begin{align}
\|\widetilde V_W-|0\rangle_{\mathbf a}\otimes V_W\|
&\leq\delta/2.
\label{eq:full-isometry-error}
\end{align}

\paragraph{The resulting block update.}
Reset the original block $B$ to $|0\rangle_B$, initialize $F$ and the ancilla registers to zero, apply $\widetilde U_{\rm amp}$, and discard $F,\mathbf a$. This implements the CPTP map
\begin{align}
\widetilde\Phi_B^{H_W}(\rho_{BA}) := \operatorname{tr}_{\mathbf a F}\left[ \widetilde V_W \bigl(\operatorname{tr}_B\rho_{BA}\bigr)
\widetilde V_W^\dagger \right].
\label{eq:compiled-block-update}
\end{align}
The diamond-norm distance between channels induced by two isometries is at most twice their operator-norm distance. Using contractivity under the initial partial trace over $B$ and the final partial trace over $F,\mathbf a$, Eq.~\eqref{eq:full-isometry-error} yields
\begin{align}
\|\widetilde\Phi_B^{H_W}-\Phi_B^{H_W}\|_\diamond
&\leq2\|\widetilde V_W-|0\rangle_{\mathbf a}\otimes V_W\| \leq\delta.
\end{align}
No postselection is required.

Finally, $1+\log(1/\gamma)=\mathcal O(1+\log(1/\delta))$ and $m=\mathcal O(1+w+\log(1/\delta))$. Each encoding call, projector-controlled operation, and state preparation has polynomial cost in $w$ and $1+\log(1/\delta)$. Gate synthesis preserves this bound. Thus the complete update uses $\operatorname{poly}(w,1+\log(1/\delta))$ gates and ancilla qubits, with classical processing time polynomial also in $b$.
\end{proof}

\subsection{Proof of Theorem~\ref{thm:dynamics-implementation}}
\label{app:implementation-complexity}

We now combine the compiled local updates with the Poisson representation and parallel scheduling. Here $b$ denotes the total classical description length of $H$ and the rational inputs $t,\varepsilon$. Recall that $s=nt/\ell$, and set $R:=2n(1+t)/\varepsilon$.
\begin{proof}
For $t=0$, use the identity channel. Assume $n>\ell$.

\paragraph{Finite-window updates.}
Let $C_{\rm win}\geq1$ and $\mu>0$ be the constants from Lemma~\ref{lem:finite-window-approximation}. Choose $\delta:=\frac{\varepsilon}{8(1+s)}, r:=\left\lceil\frac1\mu\log\frac{C_{\rm win}}{\delta}\right\rceil$. For each block $B$, set $W=N_r(B)$ and compile $\Phi_B^{H_W}$ to diamond-norm error $\delta$ using Lemma~\ref{lem:local-compilation}, as in Fig.~\ref{fig:isometry-local-compilation}(b). Extend this circuit by the identity on $W^c$ and denote the resulting channel by $\widetilde\Phi_{B,r}$. The triangle inequality gives $\|\widetilde\Phi_{B,r}-\Phi_B\|_\diamond\leq\delta+C_{\rm win}e^{-\mu r}\leq2\delta$. 

The window size satisfies $w:=|W|\leq\ell+2r$. Since $\ell$ is fixed, both $w$ and $1+\log(1/\delta)$ are $\mathcal O(\log R)$. Thus each compiled update has gate count, depth, and workspace bounded by $\operatorname{polylog}R$.

\paragraph{Bounded sampling.}
By Proposition~\ref{prop:poisson-representation}, the exact channel is the average of compositions of $K\sim\operatorname{Poisson}(s)$ block updates, whose labels are drawn independently and uniformly. Cap the number of events at $N:=\lceil2s+2\log(16/\varepsilon)\rceil$. Markov's inequality then gives $\Pr(K>N)\leq2^{-N}\mathbb E[2^K]=e^s2^{-N}\leq\varepsilon/16$.

Let $K_{\rm cap}:=\min\{K,N\}$. Its probabilities are $p_k=e^{-s}s^k/k!$ for $0\leq k<N$, with the remaining probability assigned to $N$. Compute these probabilities using $p_0=e^{-s}$ and $p_{k+1}=sp_k/(k+1)$, and approximate the cumulative probabilities $F_k:=\Pr(K_{\rm cap}\leq k)$ on a common dyadic grid. Choose a nondecreasing approximation $\widehat F_k$ with $\widehat F_N=1$ and $|\widehat F_k-F_k|\leq\varepsilon/[16(N+1)^2]$, including both numerical and rounding errors. Since each probability is a difference of consecutive cumulative probabilities, the total variation error is at most $\sum_{k=0}^{N-1}|\widehat F_k-F_k|\leq\varepsilon/[16(N+1)]$. Inverse-CDF sampling on the dyadic grid uses a bounded number of uniform bits. The computation requires polynomially many bits and classical operations in $s$, $N$, $\log(1/\varepsilon)$, and $b$.

For each block label, draw $q$ uniform bits, interpret them as $Z\in\{0,\ldots,2^q-1\}$, and return $1+\lfloor nZ/2^q\rfloor$. The resulting distribution differs from uniform by at most $n2^{-q}$ in total variation. Choose $q:=\lceil\log_2(16n(N+1)/\varepsilon)\rceil$. There are at most $N$ label draws and one count draw, so their combined sampling error is at most $\varepsilon/16$. No rejection sampling is needed.

\paragraph{Parallel scheduling: Fig.~\ref{fig:main-petz-circuit}(b).}
At most $Q:=2w-1$ block labels have windows intersecting any given window. Construct a directed dependency graph whose vertices are the sampled events, with an edge from event $i$ to event $j$ whenever $i<j$ and their windows overlap. Compiled updates on disjoint windows commute and use separate ancilla registers, so any schedule respecting these dependencies implements the same compiled composition.

Let $L$ be the longest path in the dependency graph of the uncapped ideal sequence, measured in events, with $L=0$ for the empty sequence. A path of length $m$ selects an increasing subsequence of $m$ events whose consecutive windows overlap. There are at most $nQ^{m-1}$ possible label sequences. Since $\mathbb E[\binom Km]=s^m/m!$, counting these subsequences and using $m!\geq(m/e)^m$ gives
\begin{align}
\Pr(L\geq m) \leq \mathbb E\left[\binom Km\right]\frac{nQ^{m-1}}{n^m} \leq \frac nQ\left(\frac{eQt}{\ell m}\right)^m.
\label{eq:implementation-dependency-tail}
\end{align}
Now, setting $m_\star:=\lceil2eQt/\ell+\log_2(16n/\varepsilon)\rceil+1$ gives $\Pr(L\geq m_\star)\leq(n/Q)2^{-m_\star}\leq\varepsilon/16$. Thus, since $Q=\mathcal O(w)=\mathcal O(\log R)$, $\mathcal O(t\log R+\log(n/\varepsilon))$ layers suffice with probability at least $1-\varepsilon/16$. 

First, we classically compute the layers. Initialize all site records to zero and process the sampled events in order. Assign each event one plus the largest layer previously recorded at any site in its window, and update the records at all sites in that window to the assigned layer. The resulting number of layers equals the dependency height, and events in the same layer have disjoint windows. If the number of layers is at least $m_\star$, apply the identity channel. Otherwise, execute the layers successively, performing all updates within each layer in parallel. This procedure takes polynomial in $n$ classical time and bounds the number of executed layers on every run.

\paragraph{1D nearest-neighbor implementation.}
Place each logical site and its $\operatorname{polylog}R$ ancilla qubits in a physical cell on a line or ring matching the logical geometry. Each window then occupies a connected region containing $\operatorname{polylog}R$ qubits. Route each two-qubit gate using $\mathcal O(\operatorname{polylog}R)$ nearest-neighbor SWAPs, restoring all register positions afterward. Since each compiled update has $\operatorname{polylog}R$ gates, its routed gate count and depth remain $\operatorname{polylog}R$. Disjoint windows use disjoint physical regions and can therefore execute in parallel, giving at most $m_\star$ update layers on every accepted run.

\paragraph{Total error.}
Couple the finite-precision samples to the exact capped sequence so that the probability of a sampling mismatch is at most $\varepsilon/16$. Let $G$ denote the event that $K\leq N$, the samples match, and the ideal dependency height satisfies $L<m_\star$. On $G$, the algorithm implements the compiled updates in an equivalent order. By the union bound, $\Pr(G^c)\leq3\varepsilon/16$.

Let $\widetilde\Phi_\omega$ denote the channel implemented on a realization $\omega$ of the coupled samples. On $G$, telescoping the composition bounds its distance from $\Phi_{B_K}\circ\cdots\circ\Phi_{B_1}$ by the sum of the individual update errors, since every constituent channel has diamond norm one. On $G^c$, the distance is at most two because both channels are CPTP. Thus
\begin{align}
\|\widetilde{\mathcal T}_t-e^{t\mathcal L_\ell}\|_\diamond
&=\left\|\mathbb E_\omega\left[
\widetilde\Phi_\omega-\Phi_{B_K}\circ\cdots\circ\Phi_{B_1}
\right]\right\|_\diamond
&&\text{(Proposition~\ref{prop:poisson-representation})} \nonumber\\
&\leq\mathbb E_\omega\left[\left\|\widetilde\Phi_\omega-\Phi_{B_K}\circ\cdots\circ\Phi_{B_1}\right\|_\diamond
\right] \nonumber\\
&\leq\mathbb E_\omega\left[\mathbf1_G\sum_{j=1}^{K}
\|\widetilde\Phi_{B_j,r}-\Phi_{B_j}\|_\diamond
\right]+2\Pr(G^c) \nonumber\\
&\leq2\delta \mathbb E[K\mathbf1_G]+2\Pr(G^c)
&&\text{($\|\widetilde\Phi_{B,r}-\Phi_B\|_\diamond\leq2\delta$)} \nonumber\\
&\leq2\delta \mathbb E[K]+2\cdot\frac{3\varepsilon}{16}
&&\text{(union bound)} \nonumber\\
&=\frac{s\varepsilon}{4(1+s)}+\frac{3\varepsilon}{8}
\leq\frac{5\varepsilon}{8}\leq\varepsilon.
\label{eq:implementation-total-error}
\end{align}

\paragraph{Resources.}
At most $N=\mathcal O(s+\log(2/\varepsilon))$ updates are executed. Each routed update uses $\operatorname{polylog}R$ gates and has depth $\operatorname{polylog}R$. Since $Q=\mathcal O(w)$ and $\ell$ is fixed, $m_\star=\mathcal O(wt+\log(2n/\varepsilon))$. Using $s=nt/\ell$ and $w=\mathcal O(\log R)$, we obtain
\begin{align}
G
&\leq N\operatorname{polylog}R
=\mathcal O\left(\bigl(nt+\log(2/\varepsilon)\bigr)\operatorname{polylog}R\right),
\label{eq:implementation-gate-count}\\
D &\leq m_\star\operatorname{polylog}R = \mathcal O\left(\bigl(t+\log(2n/\varepsilon)\bigr)\operatorname{polylog}R\right).
\label{eq:implementation-depth-count}
\end{align}
The factor $w=\mathcal O(\log R)$ in the depth bound is absorbed into the polylogarithmic factor.

Each of the $n$ physical cells contains $\operatorname{polylog}R$ ancilla qubits. Resetting and reusing this workspace therefore requires $\mathcal O(n\operatorname{polylog}R)$ ancilla qubits. Compiling and routing at most $n$ distinct local circuits requires $n\operatorname{polylog}R\operatorname{poly}(b)$ classical time. Including bounded sampling and dependency scheduling, the total classical processing time is $\operatorname{poly}(n,t,1+\log(1/\varepsilon),b)$. Since $R=2n(1+t)/\varepsilon$, these are the bounds claimed in Theorem~\ref{thm:dynamics-implementation}.
\end{proof}

\section{Remarks}
\addtocontents{toc}{\protect\setcounter{tocdepth}{1}}

In this section, we remark several points that may be worth noting. 

\subsection{Worst-case optimality of rapid mixing}
\label{subsec:worst-case-optimality}

The logarithmic dependence on $n$ and $\varepsilon$ in Theorem~\ref{thm:main} is worst-case optimal, as shown below by the trivial case $H=0$.

\begin{proposition}[Worst-case optimality of $\mathcal{O}(\log(n/\varepsilon))$ mixing]
\label{prop:worst-case-optimality}
Fix $d\geq2$ and an integer block length $\ell\geq1$. For $H=0$ on a periodic chain of length $n\geq2\ell$, the dynamics generated by $\mathcal L_\ell$ in Eq.~\eqref{eq:block-petz-generator} satisfies, at every $0\leq\beta<\infty$,
\begin{align}
t_{\mathrm{mix}}(\varepsilon) \geq \frac12\log\frac n\varepsilon-C_{\ell,d}, \qquad 0<\varepsilon\leq\frac14,
\label{eq:worst-case-optimality}
\end{align}
where $C_{\ell,d}$ is independent of $n$, $\varepsilon$, and $\beta$. In particular, for a block length independent of $n$, the worst-case mixing time is $\Omega(\log(n/\varepsilon))$.
\end{proposition}

\begin{proof}
For $H=0$, the Gibbs state is $\sigma_\beta=I/d^n$ at every inverse temperature, and the Petz update simply replaces the state on the block by the maximally mixed state,
\begin{align}
\Phi_B(\rho)=\frac{I_B}{d^{|B|}}\otimes\tr_B[\rho].
\label{eq:optimality-product-update}
\end{align}
Let $\rho_0=(|0\rangle\langle0|)^{\otimes n}$, and choose $m:=\lfloor n/\ell\rfloor$ sites $S:=\{1,1+\ell,\ldots,1+(m-1)\ell\}$. Consecutive sites in $S$, including the last and first sites around the cycle, are separated by at least $\ell$. Thus no block of length $\ell$ contains two sites of $S$. Since each site belongs to exactly $\ell$ blocks, whose Poisson clocks have rate $1/\ell$, the selected sites are resampled independently at total rate one. Consequently, writing $\rho_t=e^{t\mathcal L_\ell}(\rho_0)$ and $\omega:=I/d$, their reduced state is
\begin{align}
(\rho_t)_S=\omega_t^{\otimes m},\qquad
\omega_t:=e^{-t}|0\rangle\langle0|+(1-e^{-t})\omega.
\label{eq:optimality-independent-sites}
\end{align}
These states are diagonal in the product basis, so their trace distance is the total variation distance between the corresponding probability distributions. By Cauchy-Schwarz, 
\begin{align}
\frac12\|\omega_t-\omega\|_1 = \left(1-\frac1d\right)e^{-t}
\leq\sqrt{1-(\tr[\omega_t^{1/2}\omega^{1/2}])^2}.
\end{align}
Also, for probability distributions $p,q$, the inequality $\min\{p(x),q(x)\}\leq\sqrt{p(x)q(x)}$ gives $\frac12\|p-q\|_1\geq1-\sum_x\sqrt{p(x)q(x)}$. Applying this to the product states in Eq.~\eqref{eq:optimality-independent-sites}, together with contractivity of the trace distance under partial trace, yields
\begin{align}
\frac12\|\rho_t-\sigma_\beta\|_1
&\geq\frac12\|\omega_t^{\otimes m}-\omega^{\otimes m}\|_1 \nonumber \\
& \geq 1-(\tr[\omega_t^{1/2}\omega^{1/2}])^m \nonumber\\
&\geq1-\exp\left[-\frac m2\left(1-\frac1d\right)^2e^{-2t}\right].
\label{eq:optimality-trace-lower-bound}
\end{align}
It follows that, for $0<\varepsilon<1$,
\begin{align}
t_{\mathrm{mix}}(\varepsilon)
\geq\frac12\log\left(\frac{m(1-1/d)^2}{2[-\log(1-\varepsilon)]}\right).
\label{eq:optimality-mixing-lower-bound}
\end{align}
For $0<\varepsilon\leq1/4$, we have $-\log(1-\varepsilon)\leq4\varepsilon/3$ and $m\geq n/(2\ell)$. Substituting these bounds into Eq.~\eqref{eq:optimality-mixing-lower-bound} gives Eq.~\eqref{eq:worst-case-optimality}, with $C_{\ell,d}:=\frac12\log\left({16\ell}/(3(1-1/d)^2)\right)$, independent of $n$, $\varepsilon$, and $\beta$.
\end{proof}

Taking $\ell=\ell_\star$ from Theorem~\ref{thm:main}, Proposition~\ref{prop:worst-case-optimality} shows that the $\mathcal O(\log(n/\varepsilon))$ upper bound has optimal order uniformly over the Hamiltonian family.

\subsection{Extension to open boundary conditions}
\label{subsec:open-chains}

Here we give a short description on extending the argument to the open chain $\Lambda_n^{\mathrm{op}}:=\{1,\ldots,n\}$, with Hamiltonian $H^{\mathrm{op}}=\sum_{i=1}^{n-1}h^{(i)}_{i,i+1}$, $\|h^{(i)}_{i,i+1}\|\leq1$, and Gibbs state $\sigma_\beta^{\mathrm{op}}$. Throughout this subsection, distances and neighborhoods are taken in the open-chain geometry.

To retain total update rate one at every site, include truncated endpoint blocks. For $1\leq\ell<n$, define $B_s^{\mathrm{op}}:=\{s,\ldots,s+\ell-1\}\cap\Lambda_n^{\mathrm{op}}$ for $s=2-\ell,\ldots,n$. Using the Petz updates $\Phi_{B_s^{\mathrm{op}}}^{\mathrm{op}}$ from Eq.~\eqref{eq:block-petz-map}, with reference state $\sigma_\beta^{\mathrm{op}}$, set
\begin{align}
\mathcal L_\ell^{\mathrm{op}} := \frac1\ell\sum_{s=2-\ell}^{n} \left(\Phi_{B_s^{\mathrm{op}}}^{\mathrm{op}}-\mathcal I\right).
\label{eq:open-chain-generator}
\end{align}
The proofs of the two assumptions apply with uniform constants depending only on $\beta$ and $d$: each block has at most two interactions crossing its boundary, and its internal Gibbs state $\tau_B$ is an open-chain Gibbs state. Hence Theorem~\ref{thm:static-to-influence} supplies a uniformly exponentially decaying $\eta$.

It remains to check the block counting. A site $i$ belongs to $B_s^{\mathrm{op}}$ precisely when $i-\ell+1\leq s\leq i$, so it belongs to exactly $\ell$ blocks. For each $r\geq1$, a block at distance $r$ from $i$ must end at $i-r$ or begin at $i+r$, with at most one such block on either side. Thus,
\begin{align}
\sum_{s:i\in B_s^{\mathrm{op}}}\frac1\ell=1,
\qquad
\sum_{s:i\notin B_s^{\mathrm{op}}}
\eta\bigl(\mathrm{dist}(i,B_s^{\mathrm{op}})\bigr)
\leq2\sum_{r\geq1}\eta(r)=:a_\eta.
\label{eq:open-chain-counting}
\end{align}
The Euler-step proof of Theorem~\ref{thm:influence-to-curvature} therefore applies with $0\leq h\leq\ell/(n+\ell-1)$. For any fixed $\ell>a_\eta$, it gives curvature at least $\gamma_\ell:=1-a_\eta/\ell>0$. Corollary~\ref{cor:abstract-trace-mixing} then yields $\sup_\rho\frac12\|e^{t\mathcal L_\ell^{\mathrm{op}}}(\rho)-\sigma_\beta^{\mathrm{op}}\|_1\leq ne^{-\gamma_\ell t}$, and hence rapid mixing.

\subsection{The case $n=\ell_\star$}

In Theorem~\ref{thm:main}, we implicitly assumed $n>\ell_\star$, as holds in the thermodynamic limit. For $n\leq\ell_\star$, we adopt the whole-chain reset generator $\mathcal L_{\ell_\star}:=n^{-1}(\Phi_{\Lambda_n}-\mathcal I)$, where $\Phi_{\Lambda_n}(\rho)=\sigma_\beta$. This agrees with the original definition when $n=\ell_\star$ and extends it to $n<\ell_\star$. Hence $e^{t\mathcal L_{\ell_\star}}(\rho)=e^{-t/n}\rho+(1-e^{-t/n})\sigma_\beta$, so $t_{\mathrm{mix}}(\varepsilon)\leq n\log(1/\varepsilon)=\mathcal O(\log(n/\varepsilon))$, since $n\leq\ell_\star$ is bounded by a constant. 

\end{document}